\documentclass{article}
\usepackage[utf8]{inputenc}
\usepackage{amsmath, amsthm}
\usepackage{amsfonts}
\usepackage{amssymb}

\usepackage{bbm}
\usepackage{hyperref}	
\usepackage{epigraph}	
\newtheorem{theo}{Theorem}
\newtheorem{open}{Open problem}
\newcommand\numberthis{\addtocounter{equation}{1}\tag{\theequation}}

\newcommand{\PP}{\mathcal{P}_{se}}

\newcommand{\PPp}{\mathcal{P}_{se}^{\Lambda>0}}
\newcommand{\PPm}{\mathcal{P}_{se}^{\Lambda\leq 0}}

\usepackage{blindtext}
\usepackage[left=1.7cm,right=1.7cm,top=1.6cm,bottom=1.6cm]{geometry}

\usepackage{authblk}
\usepackage{enumerate}

\usepackage{enumerate}

\theoremstyle{plain}
\newtheorem{thm}{Theorem}[section]

\newtheorem{lem}[thm]{Lemma}

\newtheorem{prop}[thm]{Proposition}

\newtheorem{cor}[thm]{Corollary}

\theoremstyle{remark}
\newtheorem{rmk}{Remark}[section]

\theoremstyle{definition}

\newtheorem{defn}{Definition}[section]

\newcommand{\ep}{\epsilon} \newcommand{\ls}{\lesssim}
\newcommand{\RR}{\mathbb{R}}

\newcommand{\dphi}{\delta \phi}

\newcommand{\ddphi}{\delta \dot{\phi}}
\newcommand{\domega}{\delta \Omega^2}
\newcommand{\dr}{\delta r}
\newcommand{\ddr}{\delta \dot{r}}
\newcommand{\tuone}{\tilde{u}_1}
\newcommand{\tutwo}{\tilde{u}_2}
\newcommand{\tvone}{\tilde{v}_1}
\newcommand{\tvtwo}{\tilde{v}_2}
\newcommand{\phil}{\phi_{\mathcal{L}}}
\newcommand{\sr}{\mathfrak{s}}
\newcommand{\dphil}{\dot{\phi}_{\mathcal{L}}}

\theoremstyle{plain}

\theoremstyle{remark}

\theoremstyle{definition}

\newcommand{\s}{s_{lin}}

\usepackage{bm}

\usepackage{color}

\usepackage{float}

\addtocounter{tocdepth}{-2}
\usepackage{graphicx}
\usepackage{authblk}
 
\numberwithin{equation}{section}
\begin{document}
	\title{Violent nonlinear collapse in the interior of charged hairy black holes}

	\author{Maxime~Van~de~Moortel\thanks{mmoortel@princeton.edu}}
	
	\affil{\small  Department of Mathematics, Princeton University, 
		Washington~Road,~Princeton~NJ~08544,~United~States~of~America \vskip.1pc}

	\maketitle
	\begin{abstract}
		We construct a new one-parameter family, indexed by $\epsilon$, of two-ended, spatially-homogeneous black hole interiors solving the Einstein--Maxwell--Klein--Gordon equations with a (possibly zero) cosmological constant $\Lambda$ and bifurcating off a Reissner--Nordstr\"{o}m-(dS/AdS) interior ($\ep=0$). For all small $\epsilon \neq 0$, we prove that, although the black hole is charged, its terminal boundary is an everywhere-\textit{spacelike} Kasner singularity foliated by spheres of zero radius $r$. 
		
		Moreover, smaller perturbations (i.e.\ smaller $|\ep|$) are \textit{more singular than larger ones}, in the sense that the Hawking mass and the curvature blow up following a power law of the form  $r^{-O(\epsilon^{-2})}$ at the singularity $\{r=0\}$. This unusual property  originates  from a dynamical phenomenon -- \textit{violent nonlinear collapse} --  caused by the almost formation of a Cauchy horizon to the past of the spacelike singularity $\{r=0\}$. This phenomenon was previously described numerically in the physics literature and referred to as ``the collapse of the Einstein--Rosen bridge''.

		While we cover all values of $\Lambda \in \RR$, the case $\Lambda<0$ is of particular significance to the AdS/CFT correspondence. % in which charged black holes are dual to the grand canonical ensemble of $N$ holographic theories for large $N$.
		Our result can also be viewed in general as a first step towards the understanding of the interior of hairy black holes.
		
		%Lastly, modifying our metric by a domain of dependence argument and using the author's earlier work, we also provide the first construction of a two-ended spacetime featuring both a spacelike singularity and a Cauchy horizon. 

	\end{abstract} 
	\section{Introduction}\label{intro.section}
	
	The no-hair conjecture is a well-known statement in the Physics literature, broadly claiming that all stationary black holes are solely described by their mass, angular momentum and charge (namely they belong to the Kerr--(Newman) or the Reissner--Nordstr\"{o}m family), see the review \cite{Chrusciel_review} and references therein. In (electro)-vacuum, celebrated uniqueness theorems \cite{uniqueness2,uniqueness3,uniqueness5,uniqueness1,uniqueness4} preclude the existence of asymptotically flat ``hairy'' black holes.  However, there exists a plethora of literature on hairy black holes for relatively exotic matter models: Arguably the most emblematic known hairy black holes are static solutions coupled with non-abelian gauge theories, satisfying the Einstein--Yang--Mills equations  \cite{Bizon,P9,P4,P8} or the Einstein--Yang--Mills equations coupled with a Higgs or Dilaton field \cite{P9,P2,Sarbach}.
	
	In the present study, we consider a typical matter model, the Einstein--Maxwell--Klein--Gordon equations, with a cosmological constant $\Lambda \in \RR$, and a scalar field $\phi$ obeying the linear Klein--Gordon equation with mass $m^2 \in \RR-\{0\}$:
	
	\begin{align}\label{E1} & Ric_{\mu \nu}(g)- \frac{1}{2}R(g)g_{\mu \nu}+ \Lambda g_{\mu \nu}= \mathbb{T}^{EM}_{\mu \nu}+  \mathbb{T}^{KG}_{\mu \nu} , \\ & \label{E2} \mathbb{T}^{EM}_{\mu \nu}=2\left(g^{\alpha \beta}F _{\alpha \nu}F_{\beta \mu }-\frac{1}{4}F^{\alpha \beta}F_{\alpha \beta}g_{\mu \nu}\right), \hskip 7 mm \nabla^{\mu} F_{\mu \nu}=0, \\ &  \label{E3}  \mathbb{T}^{KG}_{\mu \nu}= 2\left( \nabla_{\mu}\phi\nabla_{\nu}\phi -\frac{1}{2}(g^{\alpha \beta} \nabla_{\alpha}\phi \nabla_{\beta}\phi + m ^{2}|\phi|^2  )g_{\mu \nu} \right), \\ &\label{E5} g^{\mu \nu} \nabla_{\mu} \nabla_{\nu}\phi = m ^{2} \phi.   \end{align}

	A uniqueness result of Bekenstein \cite{Bekenstein}   precludes the existence of asymptotically flat ($\Lambda=0$) hairy black holes for the above system (at least when $m^2  \geq 0$). However, there has been recent significant interest in asymptotically AdS static hairy black holes when $\Lambda<0$, in connection with the AdS/CFT correspondence \cite{r1,r2,r3,numerics.uncharged, numerics.RN, numerics.charged}.  In our main theorem below, we will consider the subject under a dynamical perspective 
	and study rigorously the time-evolution of characteristic initial data consisting of a constant scalar field $\phi \equiv \ep$ of small amplitude $\ep\neq0 $ on a two-ended event horizon. The resulting spacetime is a one-parameter family  bifurcating from the Reissner--Nordstr\"{o}m-(dS/AdS) interior metric, which we interpret as the interior region of a charged and static hairy black hole. We will however limit our study to the black hole interior and not concern ourselves with the construction of the asymptotically AdS black hole exterior  (see Figure \ref{Fig_Two}), since we do not want to impose the sign of the cosmological constant in the present work.
	
	The resulting family of metrics $g_{\ep}$ we construct are charged $\ep$-perturbations of the Reissner--Nordstr\"{o}m-(dS/AdS) interior, but surprisingly \textit{do not admit a Cauchy horizon}; instead, their singularity is everywhere spacelike (Figure \ref{Fig3}).
	
	Furthermore, we show that the evolution problem obeys highly nonlinear dynamics leading to a more violent singularity than expected. %Finally, we also address the  limit $ \ep \rightarrow 0$, and show that it is non-unique; moreover, convergence is not uniform.

	%, and only holds in a very rough topology given by the Bounded Mean Oscillations norm (BMO norm).

	\begin{theo} \textup{[Rough version]} \label{thm.intro}
		Fix the following characteristic initial data on bifurcating event horizons $\mathcal{H}^+_1 \cup \mathcal{H}^+_2$: \begin{align}
			& \label{phi.hair}\phi \equiv \epsilon,\\ &g= g_{RN} \color{black} + O(\epsilon^2) \color{black},
		\end{align} where $ g_{RN}$ is the Reissner--Nordstr\"{o}m-(dS/AdS) metric \eqref{RN} with sub-extremal parameters $(M,e,\Lambda)$. The Maximal Globally Hyperbolic Development $(\mathcal{M}_{\ep},g_{\ep})$ of this data is a spatially-homogeneous spacetime with topology $\RR \times \mathbb{S}^2$.
		
		Then, for almost every $(M,e,\Lambda,m^2)$, there exists $\ep_0(M,e,\Lambda,m^2)>0$ such that for all $0<|\ep| < \ep_0$, the spacetime $(\mathcal{M}_{\ep},g_{\ep})$ ends at a \emph{spacelike singularity $\mathcal{S}:=\{r=0\}$}, where $r$ is the area-radius of the $ \mathbb{S}^2$  sphere.  Moreover:

		\begin{enumerate}[i.]
			\item \underline{Almost formation of a Cauchy horizon}: $g_{\ep}$ is uniformly close to Reissner--Nordstr\"{o}m-(dS/AdS) locally (in Reissner--Nordstr\"{o}m-(dS/AdS) time) and moreover $g_{\ep}$ converges weakly to Reissner--Nordstr\"{o}m-(dS/AdS) as $\ep \rightarrow 0$.

			\item \underline{Singular power-law inflation}\label{power.law.statement}: the Hawking mass $\rho$ and the Kretschmann scalar $\mathfrak{K}$ blow up at $\mathcal{S}:=\{r=0\}$ as: \begin{equation} \begin{split} \label{rho.intro}
					\rho(r) \approx  		r^{- b_-^{-2}  \epsilon^{-2}+ O(\ep^{-1})}, \hskip 5 mm \mathfrak{K}(r) \approx  		r^{- 2 b_-^{-2}   \epsilon^{-2} +O(\ep^{-1})},
				\end{split}
			\end{equation} and we call $b_-(M,e,\Lambda,m^2) \neq 0$ the resonance parameter. $\approx$ means equivalent  as $r \rightarrow 0$, up to a constant.
			
			\item \label{Linfty.statement}  \underline{Kasner-type behavior}: near the singularity $\mathcal{S}:=\{r=0\}$, $g_{\ep}$ is uniformly close to a Kasner-like metric $\tilde{g}_{\ep}^{Kas}$:
			\begin{align}  \label{Kasner.intro.eq}
				&\tilde{g}_{\ep}^{Kas}= -d\tilde{\tau}^2 + \tilde{\tau}^{2 (1-4b_-^2 \ep^2)+O(\ep^3)} d\rho^2 +  \tilde{\tau}^{4b_-^2 \ep^2+O(\ep^3)}\cdot  r_-^2 \cdot ( d\theta^2 + \sin^2(\theta) d\psi^2),\ \\ & \phi(\tilde{\tau}) = 2 b_- \cdot [  \ep + O(\ep^2)] 
				\cdot  \log(\tilde{\tau}^{-1}), 
			\end{align} where $r_-(M,e,\Lambda)>0$ is the radius of the Cauchy horizon of the unperturbed metric \eqref{RN}$=g_{\ep=0}$.
			%	\item \label{BMO.statement}  \underline{$BMO$ convergence}: near the singularity $\mathcal{S}:=\{r=0\}$, $g_{\ep}$ converges to Minkowski in the $BMO$ norm as $\ep \rightarrow 0$.

		\end{enumerate}
	\end{theo}
	
	\begin{figure} 
		
		\begin{center}
			
			\includegraphics[width=67 mm, height=73 mm]{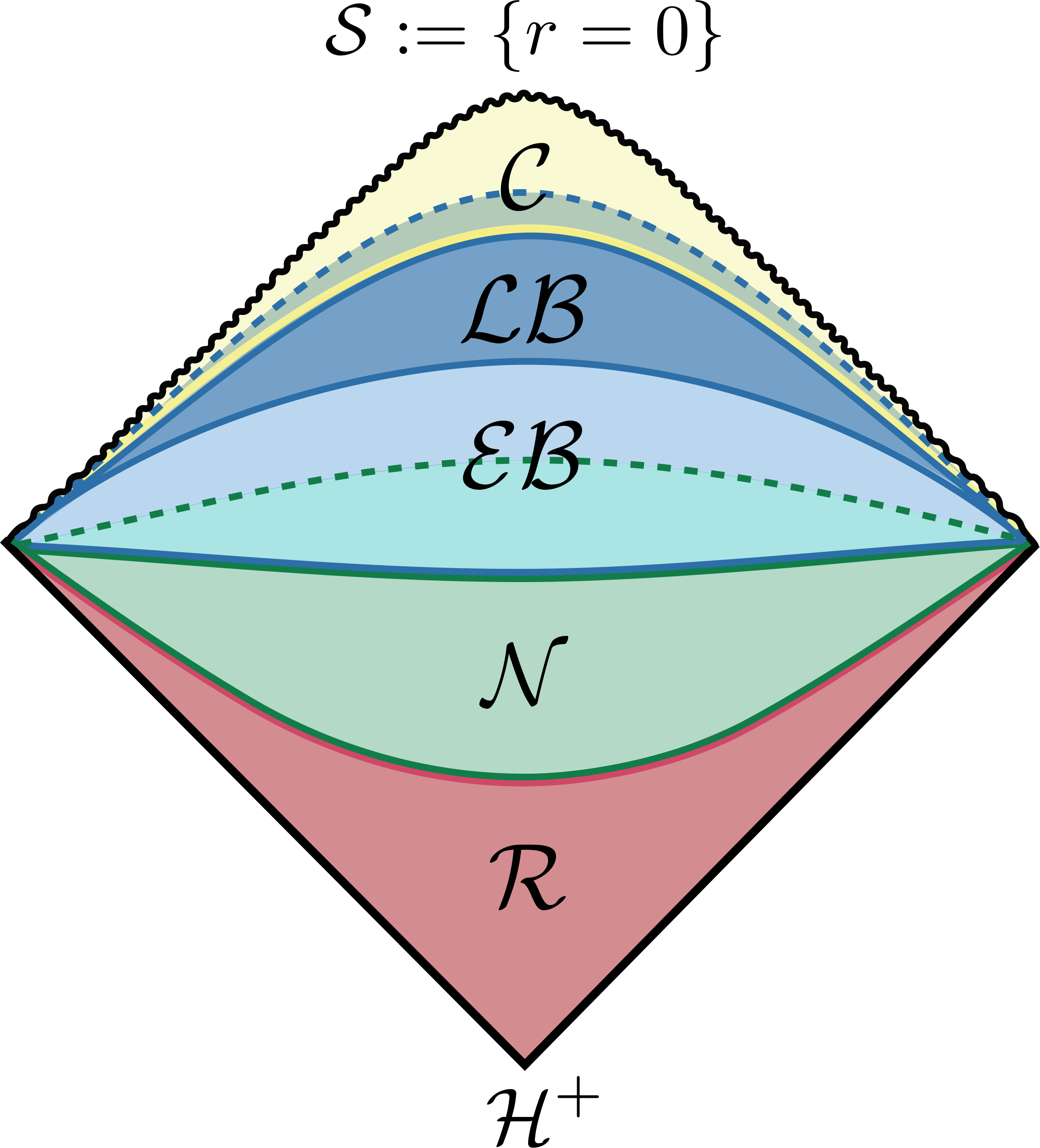}
			
		\end{center}

		\caption{Penrose diagram of $g_{\ep}$ from Theorem~\ref{thm.intro}, where $\mathcal{S}=\{r=0\}$ is a spacelike singularity. The spacetime regions $\mathcal{R}$, $\mathcal{N}$, $\mathcal{EB}$, $\mathcal{LB}$ and $\mathcal{C}$ are introduced in Section~\ref{strategy.section}. Note that   $\mathcal{EB}$ \color{black} overlaps with \color{black} $\mathcal{N}$ and $\mathcal{LB}$ \color{black} overlaps with \color{black} $\mathcal{C}$.}
		%\caption{Penrose diagram of the metric $g_{\ep}$ from Theorem \ref{thm.intro}}
		\label{Fig3}
	\end{figure} 
	
	We emphasize that we do not fix the sign of $\Lambda$, or of $m^2$: if $\Lambda=0$ (respectively $\Lambda>0$, $\Lambda<0$), the spacetime metric $g_{\ep=0}$ from which $g_{\ep}$ bifurcates is a Reissner--Nordstr\"{o}m (respectively RN-de-Sitter, RN-anti-de-Sitter) interior metric.

	We now make a few essential remarks on our spacetime $(\mathcal{M}_{\ep},g_{\ep})$, and announce the outline of the introduction: 
	
	\begin{itemize}
		\item 
		If we consider instead black holes solutions of \eqref{E1}--\eqref{E5}  \textit{relaxing to Reissner--Nordstr\"{o}m}, namely if instead of \eqref{phi.hair}, we have the following non-hairy behavior of the scalar field  $$\phi \rightarrow 0 \text{ towards } i^+ \text{ in the black hole exterior,}$$ then  the black hole interior \textit{does admit} a Cauchy horizon, in complete contrast to the spacetime $(\mathcal{M}_{\ep},g_{\ep})$ of Theorem~\ref{thm.intro}, see Section~\ref{AF.comparision} and Figure \ref{Fig.relax}.

		\item Asymptotically AdS black holes play an important role in AdS/CFT, in connection to two of the most celebrated problems in the quantum aspects of gravity: the information paradox \cite{HartmanMaldacena,info.paradoxHawking1,info.paradoxHawking2} on the one hand, and the probing of the singularity at the black hole terminal boundary on the other hand \cite{numerics.uncharged,numerics.RN,numerics.charged} (see Section~\ref{ADS.section}).	 For both problems, many standard results in AdS/CFT (like the computations of entanglement entropy \cite{HartmanMaldacena}) are considered on the Reissner--Nordstr\"{o}m-AdS spacetime, which is very specific in that it has a Cauchy horizon, contrary to $g_{\ep}$. It would be interesting to examine these results on the more general 
		spacetime $(\mathcal{M}_{\ep},g_{\ep})$ constructed in Theorem~\ref{thm.intro}.
		
		\item Although $g_{\ep}$ has a spacelike singularity, a Cauchy horizon \textit{almost forms}, i.e.\ $g_{\ep}$ is close to Reissner--Nordstr\"{o}m-(dS/AdS) for large intermediate times (at which a Reissner--Nordstr\"{o}m metric would be ``close to its Cauchy horizon''). Moreover, we have weak convergence of $g_{\ep}$  to Reissner--Nordstr\"{o}m-(dS/AdS), see Section~\ref{RN.stab.intro}.
		
		\item The almost formation of the Cauchy horizon dramatically impacts the singularity structure and is responsible for what we call \textit{violent nonlinear collapse}, reflected by the $r^{-b_-^{-2}\ep^{-2}}$ power-law inflation rates from \ref{rho.intro}. Such singular rates (note that $\ep^{-2}\rightarrow \infty$, as $\ep \rightarrow 0$) conjecturally do not occur for hairy perturbations of Schwarzschild-(dS/AdS), see Section~\ref{violent.section}.
		
		\item The other effect of violent nonlinear collapse is to make $g_{\ep}$ uniformly close (near the singularity) to a \emph{Kasner metric} (see Section~\ref{Kasner.intro.section}) of  exponents $(1-4b_-^2\ep^2 +O(\ep^3),\ 2b_-^2\ep^2 +O(\ep^3),\ 2b_-^2\ep^2 +O(\ep^3))$. Note that an exact Kasner metric of exponents $(1-4b_-^2\ep^2,\ 2b_-^2\ep^2,\ 2b_-^2\ep^2)$ (recall that $b_-(M,e,\Lambda,m^2 )\neq 0$) converges (locally) to Minkowski in  $L^p$ norm, but not uniformly as $\ep \rightarrow 0$.% which is consistent with statement \ref{BMO.statement} of Theorem~\ref{thm.intro}.
		
		\item In Section~\ref{comparison.section}, we discuss other results and analogies between $g_{\ep}$ and hairy black holes for other matter models.
		\item The author hopes that the present study will pave the way towards other interesting problems such as \begin{enumerate}[A.]
			\item Constructing a static, asymptotically AdS black hole exterior matching with the interior metric $g_{\ep}$.
			\item Understanding the stability of the metric $g_{\ep}$ with respect to \textit{non spatially-homogeneous perturbations}.
			\item  Understanding the singularity inside rotating (charged or uncharged) hairy black holes.
			\item Studying similar models which admit BKL-type oscillations, a topic of great importance in cosmology.
		\end{enumerate} For further developments on the above problems, we refer the reader to Section~\ref{open.problems}.
		\item   Concerning the proof, we want to emphasize two important aspects (see Section~\ref{strategy.section} for more details): 
		
		\begin{enumerate}[a.]
			\item \textit{The importance of distinguishing different time scales}, in particular short times in which Reissner--Nordstr\"{o}m-(dS/AdS) enjoys Cauchy stability, intermediate times where a typical blue-shift instability kicks in 
			and late times where the nonlinearity dominates and monotonicity takes over, leading to collapse to $\{r=0\}$.
			
			\item We exploit a \textit{linear instability} \cite{ChristophYakov} for the Klein--Gordon equation on  Reissner--Nordstr\"{o}m-(dS/AdS). This instability is important at early/intermediate times and relies on a scattering resonance (absent in the case $m^2=\Lambda=0$) giving $b_- \neq 0$, that occurs for \textit{almost every} (but not all) parameters $(M,e,\Lambda,m^2)$.
		\end{enumerate}

	\end{itemize}

	\subsection{Comparison with non-hairy black holes relaxing to Reissner--Nordstr\"{o}m} \label{AF.comparision}
	
	Dafermos--Luk proved in \cite{KerrStab} the stability of Kerr's Cauchy horizon with respect to vacuum perturbations \textit{relaxing to Kerr} at a fast, integrable rate (consistent with the fast rates one would obtain in the exterior problem \cite{MihalisStabExt}, \cite{MihalisStabExtKerr}). For the model \eqref{E1}--\eqref{E5}, the relaxation is conjectured to occur at a slower rate if $ \Lambda=0$ (see the heuristics/numerics from \cite{Phycists2}, \cite{KonoplyaZhidenko}, \cite{KoyamaTomimatsu}), which is a serious obstruction to asymptotic stability, even in spherical symmetry. Nevertheless, the author proved that the Reissner--Nordstr\"{o}m Cauchy horizon is stable with respect  to spherically symmetric perturbations, providing they decay at a (slow) inverse polynomial relaxation rate consistent with the conjectures:

	\begin{thm}[\cite{Moi}] \label{CH.thm} Consider regular spherically symmetric characteristic data on  $\mathcal{H}^+ \cup \underline{C}_{in}$, where $\mathcal{H}^+:= [1,+\infty)_v  \times \mathbb{S}^2$, converging to a sub-extremal Reissner--Nordstr\"{o}m-(dS/AdS) at the following rate:  for some $s>\frac{1}{2}$ and for all $v \in \mathcal{H}^+$:
		\begin{equation} \label{decay}
			|\phi|(v) + |\partial_v \phi|(v) \ls v^{-s},
		\end{equation} where $v$ is a standard Eddington--Finkelstein advanced time-coordinate. Then, restricting $\underline{C}_{in}$ to be sufficiently short, the future domain of dependence of $\mathcal{H}^+ \cup \underline{C}_{in}$  is bounded by a Cauchy horizon $\mathcal{CH}^+$, namely a null boundary emanating from $i^+$ and foliated by spheres of strictly positive area-radius $r$, as depicted in Figure~\ref{Fig.relax}.
	\end{thm}
	
	\begin{figure}[H]
		
		\begin{center}
			
			\includegraphics[width= 58 mm, height=40
			mm]{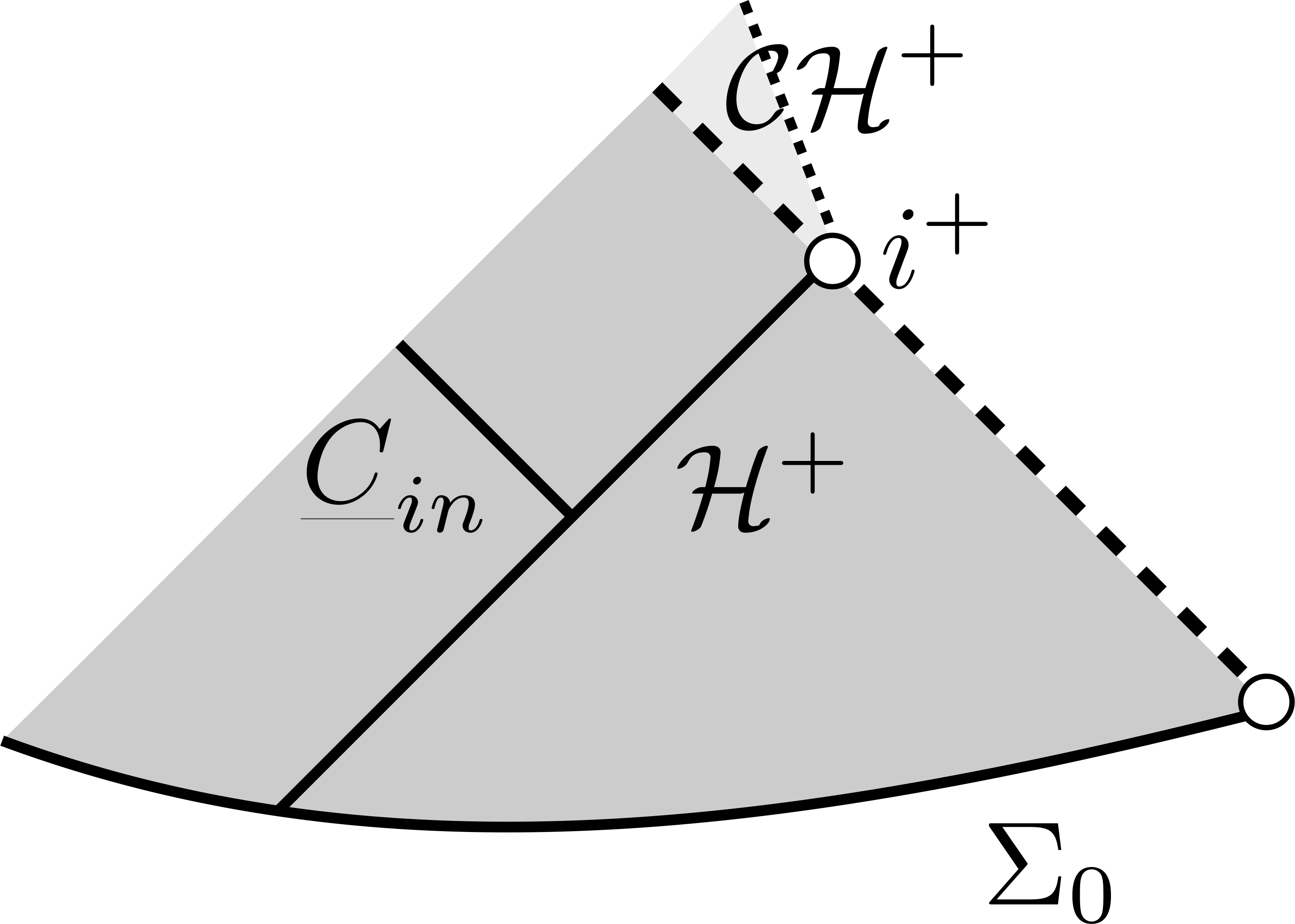}
			
		\end{center}

		\caption{Penrose diagram of the spacetime corresponding to Theorem~\ref{CH.thm}}
		\label{Fig.relax}
	\end{figure} 
	
	\begin{itemize}
		
		\item Theorem~\ref{CH.thm} stands in \textit{complete contrast} with Theorem~\ref{thm.intro}: to summarize, the Cauchy horizon is stable if $\phi$ decays to zero (even at a slow rate as in \eqref{decay}) but is  unstable if $\phi\equiv \ep$ is small but does not decay.
		\item Note that for $\Lambda=0$ and asymptotically flat Cauchy data, since  by \cite{Bekenstein}, Reissner--Nordstr\"{o}m is the only static solution of \eqref{E1}--\eqref{E5}, it means that the decay of $\phi$ to $0$ indeed quantifies the relaxation of the black hole exterior towards a Reissner--Nordstr\"{o}m metric. In a slight abuse of terminology, we will say that a (spherically symmetric) black hole relaxes to  Reissner--Nordstr\"{o}m when the scalar field $\phi$ tends to $0$ towards $i^+$ in black hole exterior.
		
		\item Note that the polynomial decay as in  \eqref{decay} is conjecturally sharp in the exterior \cite{Moi2}  if $\Lambda=0$, but not %sharp to hold in the exterior for regular data if $\Lambda \geq 0$, but is not sharp 
		if $\Lambda>0$ because the decay is exponential \cite{HintzVasy}; moreover \eqref{decay} is not satisfied if $\Lambda<0$, as the decay is logarithmic \cite{HolzegelSmulevici}.

		\item If we assume integrable decay for this model, i.e.\ $s>1$ (an unrealistic assumption if $\Lambda=0$, in view of the conjectured rates in the exterior), then we prove in \cite{Moi} that the metric  is continuously extendible at the Cauchy horizon. Under similar assumptions, Dafermos--Luk reached the same conclusion for perturbations of Kerr \cite{KerrStab} without symmetry: since the integrable decay $s>1$ becomes a realistic assumption in the vacuum case \cite{MihalisStabExtKerr,BH.stab1,KS2}, the result of Dafermos--Luk \cite{KerrStab} also falsifies the $C^0$-formulation of Strong Cosmic Censorship (by means of decay).
		
		\item Nevertheless, if $\frac{1}{2}  < s \leq 1$, there is a large class of $\phi$-data obeying \eqref{decay} such that the Cauchy horizon exists, but admits a novel \emph{null contraction singularity} that renders the metric $C^0$-inextendible, and $\phi$ is unbounded in amplitude \cite{Moi3Christoph,Moi4Christoph}. However, uniform boundedness and continuous extendibility  of the metric and $\phi$ hold for a sub-class of \textit{oscillating event horizon data}, as proven in \cite{Moi3Christoph}. Since these oscillations at the event horizon are conjectured to be generic \cite{HodPiran1,KoyamaTomimatsu}, the work \cite{Moi3Christoph}  also falsifies the $C^0$-formulation of Strong Cosmic Censorship in spherical symmetry by means of the \emph{oscillations} of the perturbation (as opposed to by means of decay).
		
		\item If we assume an averaged version of \eqref{decay} as a lower bound (still consistent with the conjectured rates) then we proved in \cite{Moi,Moi4} that the curvature blows up at the Cauchy horizon and that mass inflation occurs (i.e.\ the Hawking mass blows up): this shows that the metric is (future) $C^{2}$-inextendible \cite{Moi4}. This statement is called the $C^2$ version  of the Strong Cosmic Censorship conjecture, which is thus true for this matter model.
		
		\item Theorem~\ref{CH.thm} is a local result, which is independent of the topology of the Cauchy data. It turns out that solutions of \eqref{E1}--\eqref{E5} with $F \neq 0$ in spherical symmetry are constrained to have two-ended data topology $\RR \times \mathbb{S}^2$, a setting which is not realistic to study astrophysical gravitational collapse (but may be in other settings, see Section~\ref{ADS.section}). In fact, Theorem~\ref{CH.thm} and all the results claimed in this section were originally proven for the Einstein--Maxwell--Klein--Gordon system where the scalar field  $\phi$ is also allowed to be \textit{charged}, see \cite{Moi}, \cite{Moi4}. Unlike model \eqref{E1}--\eqref{E5}, this charged model allows for one-ended data with topology $\RR^3$ admitting a regular center and provides an acceptable model to study spherical collapse as a global phenomenon. The one-ended collapse case brings a striking conclusion: under the above assumptions, while a  Cauchy horizon $\mathcal{CH}^+$ is present  near $i^+$ (in the domain of dependence of $\underline{C}_{in}$, see Figure \ref{Fig.relax}), it \textit{breaks down} globally and a crushing singularity forms near the center, as we prove in \cite{r=0}.

	\end{itemize}

	\subsection{Boundary value problem and potential applications to AdS/CFT} \label{ADS.section}

	Two-ended stationary black holes (``eternal'') play a pivotal role in the celebrated Anti-de-Sitter/Conformal Field Theory correspondence \cite{Maldacena.ADSCFT}, see for instance \cite{Maldacena.BH,Witten}. In the AdS/CFT dictionary, charged black holes correspond thermodynamically to the Grand Canonical Ensemble of $N$ holographic theories at the boundary for large $N$   \cite{Emparan,HawkingReall}.
	
	One important open problem regarding the quantum aspects of gravity is the \emph{information paradox} resulting from the apparent loss of information due to the quantum evaporation of black holes by Hawking radiation \cite{info.paradoxHawking1,info.paradoxHawking2}. Recently, the quantum concept of entanglement entropy \cite{HartmanMaldacena} was used in an approach \cite{info.paradox2} to explain this paradox.  Most computations in this context are, however,  made on the Reissner--Nordstr\"{o}m-AdS electro-vacuum solution, which is highly non-generic, and moreover its terminal boundary is a smooth Cauchy horizon $\mathcal{CH}^+_A \cup \mathcal{CH}^+_B$ (see Figure \ref{Fig.RADS}).
	
	\begin{figure}[H]
		
		\begin{center}
			
			\includegraphics[width=63 mm, height=70 mm]{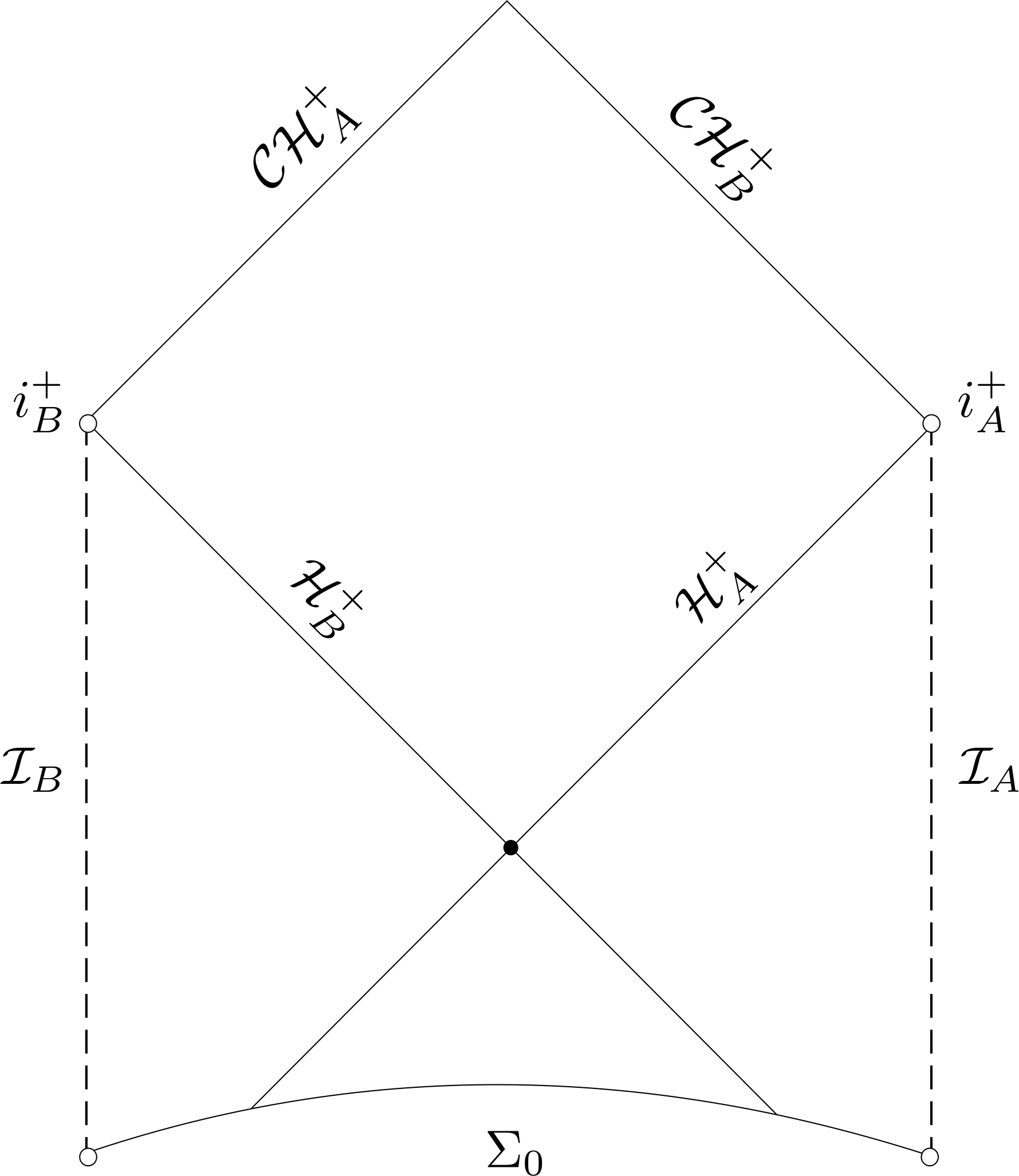}
			
		\end{center}

		\caption{Penrose diagram of the Reissner--Nordstr\"{o}m-AdS spacetime}
		\label{Fig.RADS}
	\end{figure}  
	
	In complete contrast, the terminal boundary of our charged hairy black holes $g_{\ep}$ from Theorem~\ref{thm.intro} is not a Cauchy horizon, but a spacelike singularity $\mathcal{S}$ instead (compared with Figure \ref{Fig3}). It would be interesting to construct a static, asymptotically AdS black hole exterior corresponding to $g_{\ep}$ and carry out the computations of standard quantum quantities, such as the entanglement entropy, in this setting. The metric asymptotics we derive (see Theorem~\ref{maintheorem} and below) will likely be crucial for this task.

	Another fundamental problem is the understanding of quantum effects near the singularity located at the terminal boundary of a black hole. As it turns out, the interior of asymptotically AdS black holes provides a simplified model to understand these effects; motivated by these considerations, Hartnoll et al. \cite{numerics.uncharged,numerics.RN,numerics.charged} studied numerically the interior of hairy black holes for various charged/uncharged matter models, and discovered a wealth of singularities (see Section~\ref{comparison.section}).  In particular, charged hairy black holes with  asymptotically AdS asymptotics are  studied in the numerical work \cite{numerics.RN}: The interior part of these black holes corresponds to  $g_{\ep}$ from Theorem~\ref{thm.intro}, see Section~\ref{hartnoll.section} for details.
	
	Lastly, we mention a soft argument in \cite{numerics.RN} ruling out \emph{smooth} Cauchy horizons (in particular, assuming a finite scalar field $\phi$), though in principle the result in \cite{numerics.RN} would still be consistent with having a singular Cauchy horizon \footnote{Note that in other contexts, singular Cauchy horizons exist due to mass inflation, see \cite{Moi,Moi4} and the discussion in Section~\ref{AF.comparision}.} (where the scalar field $\phi$ or its derivatives blow up for instance, even if the metric itself is smooth).

	\subsection{Almost formation of a Cauchy horizon and stability for intermediate times} \label{RN.stab.intro}
	
	For the purpose of this discussion, we define a time coordinate $s$ on $g_{\ep}$ by \ref{gauge}. On Reissner--Nordstr\"{o}m-(dS/AdS),\\ $s$ coincides with the $r^*$-tortoise coordinate: in particular, $s=r^* \in \RR$ and $\{s=+\infty\}$ is the Cauchy horizon, $\{s=-\infty\}$ is the event horizon. The statements  (Theorem~\ref{maintheorem} and Theorem~\ref{convergence.theorem}) we prove can be summarized as: \begin{enumerate}
		
		\item \label{AF1} the spacetime $\mathcal{M}_{\ep}$ corresponds to $s \in (-\infty,s_{\infty}(\ep))$, ending at $\mathcal{S}=\{r=0\}=\{s_{\infty}(\ep)\}$ and $s_{\infty}(\ep) \simeq \ep^{-2}<+\infty$. 
		\item   \label{AF2} $g_{\ep}$ is uniformly close to the Reissner--Nordstr\"{o}m-(dS/AdS)  metric \eqref{RN} in amplitude for all $s \ll \ep^{-2}$.
		\item The derivatives of the $g_{\ep}$ are uniformly close to the derivatives of $g_{RN}$ for all $s \ll \log(\ep^{-1})$.
		
		\item The derivatives of the $g_{\ep}$, in particular the Hawking mass, become arbitrarily large for  $\log(\ep^{-1}) \ll s \ls \ep^{-2}$.
		\item  \label{AF4}  $g_{\ep}$ converges in the sense of distributions $\mathcal{D}'(\RR_s)$ %$\mathcal{D}'(\RR_s)$ (in fact, in measure)  
		to the Reissner--Nordstr\"{o}m-(dS/AdS)  metric \eqref{RN} as $\ep \rightarrow 0$.
		\item  \label{AF5}  $g_{\ep}$ does \emph{not} converge uniformly to \eqref{RN}, or in any Lebesgue space $L^p(\RR_s)$ for $p\geq 1$.
	\end{enumerate}
	We interpret statement \ref{AF2} as the \textit{almost formation} of a Cauchy horizon, as $s$ is allowed to be larger than any $\ep$-independent constant, which, on Reissner--Nordstr\"{o}m-(dS/AdS), corresponds to being close to the Cauchy horizon.
	
	Note that we roughly have three regions (see Section~\ref{strategy.section} for details): $s \ll \log(\ep^{-1})$, where Cauchy stability prevails in the $C^1$ norm; $\log(\ep^{-1})\ll s \ll \ep^{-2}$, where $C^0$ stability still holds, but a $C^1$ blow-up, characteristic of the blue-shift instability, occurs; and lastly a region $s \simeq \ep^{-2}$ where the spacelike singularity forms, see Section~\ref{Kasner.intro.section}.
	
	\begin{rmk}Note that the convergence in distributions from statement \ref{AF4} is \textit{specifically expressed} in the $s$ coordinate system (defined in \eqref{gauge}). However, in the  Kasner-type coordinate $\tilde{\tau}$ (see \eqref{Kasner.intro.eq}), the metric converges in $L^p_{\tilde{\tau}}$ to the Minkowski metric  in the $\ep \rightarrow 0$ limit. 
	\end{rmk}
	
	\subsection{Violent nonlinear collapse and strength of the spacelike singularity} \label{violent.section}
	
	\textit{Despite} the weak stability up to large intermediate times of Reissner--Nordstr\"{o}m-(dS/AdS), a spacelike singularity forms at late times $s \simeq \ep^{-2}$. We discuss the strength of this singularity, and compare it to the Schwarzschild metric: \begin{equation} \label{gS} \tag{$g_S$}
		g_{S}:=-(1-\frac{2M}{r}) dt^2+ (1-\frac{2M}{r})^{-1} dr^2 +r^2[ d\theta^{2}+\sin(\theta)^{2}d \varphi^{2}].
	\end{equation} On the Schwarzschild metric \eqref{gS}, $r$ is also the area-radius  and the Hawking mass is constant and equal to $M>0$. We argue that the singularity $\mathcal{S}=\{r=0\}$ on $g_{\ep}$ is ``more singular'' than $\{r=0\}$ on \eqref{gS}, which justifies our denomination \textit{violent nonlinear collapse} (we explain the origin of ``nonlinear'' later in the section) for four reasons: \begin{enumerate}
		\item The Hawking mass $\rho$ on $g_{\ep}$ blows up as a power-law $r^{-b_-^{-2} \cdot \ep^{-2}+\color{black} O(\log(\epsilon^{-1})) \color{black}}$ at $r=0$, while it is constant on \eqref{gS}.
		
		\item The Kretchsmann scalar $\mathfrak{K}:= R_{\alpha\beta \gamma \delta}R^{\alpha\beta \gamma \delta}$ blows up like $r^{-2 b_-^2 \cdot \ep^{-2}+\color{black} O(\log(\epsilon^{-1})) \color{black}}$ while it blows up like $r^{-6}$ on \eqref{gS}.
		\item The quantity $r^6 \rho^{-2} \mathfrak{K} \simeq \ep^{-4}$ is finite on $g_{\ep}$ (as on \eqref{gS}) but becomes arbitrarily large for arbitrarily small $|\ep|$.
		
		\item The spacetime volume of  $\{r \geq r_0,\ 0 \leq t \leq 1\}$ is of order $r_0^{ b_-^{-2} \cdot \ep^{-2}+ O(\log(\epsilon^{-1})) \color{black}}$ for small $r_0$, as opposed to $r^3_0$ on \eqref{gS}.
	\end{enumerate} To go beyond the heuristics of this section, we refer to Theorem~\ref{maintheorem} for the precise estimates that we prove.  \\

	It is also interesting to compare $g_{\ep}$ to spherically symmetric solutions in Christodoulou's model (i.e.\ \eqref{E1}-\eqref{E5} with $F\equiv 0$, $m^2=\Lambda=0$) constructed in \cite{Christo1,Christo2}; such spacetimes converge to \eqref{gS} towards $i^+$ and admit a spacelike singularity. It has been proven recently   in \cite{DejanAn,AnZhang} that, on these spacetimes, $\mathfrak{K}$ and $\rho$ blow-up like $r^{-\alpha(t)}$ with a time-dependent rate $\alpha(t)$ converging to the Schwarzschild values as $t \rightarrow+\infty$, respectively $6$ and $0$.

	In view of this, the truly surprising fact is not the power-law blow up of $\rho$, but the $\ep^{-2}$-rate, which tends to $+\infty$ as $|\ep| \rightarrow 0$. Naively, one may expect that an $\ep$-perturbation of \eqref{RN} will give rise to $O(\ep)$ rates at the very least, allowing to recover \eqref{RN} in the (strong) limit  but we obtain a singular limit instead,  \textit{caused by the nonlinearity}.

	To understand, it is useful to look  at the linear wave equation on the Schwarzschild metric $\Box_{g_S} \psi=0$. Solutions blow-up  at $r=0$ (see \cite{GregJan}) as: $$ \psi(r,t) \approx A(t) \cdot \log(r^{-1}).$$ For constant data $\psi \equiv \ep$ on the event horizon, we get $A(t)= C(M) \cdot \ep$, consistent with the linearity of the equation.
	
	However, in the evolution $(g_{\ep},\phi_{\ep}=\ep)$ of the data of Theorem~\ref{thm.intro} we have, near the singularity $\mathcal{S}=\{r=0\}$: \begin{equation} \label{intro.violent1}
		\phi_{\ep}(r) \simeq \ep^{-1} \cdot \log(r^{-1}),
	\end{equation}a radically different rate! The main explanation behind \eqref{intro.violent1} is a nonlinear estimate that we prove on $g_{\ep}$: \begin{equation} \label{intro.violent2}
		\frac{dr}{ds} \simeq -\ep^2\cdot  r^{-1},
	\end{equation} to be compared with  $\frac{dr}{ds} \approx -M \cdot r^{-1}$ on \eqref{gS} (with the same $s$ defined by \eqref{gauge}). The striking fact, is that \eqref{intro.violent2} is a remnant of the Reissner--Nordstr\"{o}m Cauchy horizon stability for intermediate times (the phenomenon we explained in Section~\ref{RN.stab.intro})! This explains our claim that the  \textit{violent nonlinear collapse} phenomenon only occurs for nonlinear perturbations of charged black holes, and not perturbations of Schwarzschild. The numerics of \cite{numerics.uncharged}, where data are set as in Theorem~\ref{thm.intro} replacing \eqref{RN} by \eqref{gS} (equivalently assuming $F\equiv 0$), tend to confirm this expectation.

	Further details on this nonlinear collapse dynamics and the role of the different phenomena are given in Section~\ref{strategy.section}.
	
	\subsection{Kasner-like behavior and convergence in  $L^p$ spaces} \label{Kasner.intro.section} We recall the Kasner metric, a solution of the Einstein-scalar-field system i.e.\ \eqref{E1}-\eqref{E5} with $F \equiv 0$, $m^2=0$: \begin{equation} \label{Kasner.exact}
		g_{Kas}= -d\tilde{\tau}^2 +   \tilde{\tau}^{2 p_{rad}}  d\rho^2  +    \tilde{\tau}^{2 p_{x}} dx^2+ \tilde{\tau}^{2 p_{y}}dy^2, \hskip 5 mm \phi(\tilde{\tau})= p_{\phi} \cdot \log(\tilde{\tau}^{-1});
	\end{equation} \begin{equation} \label{p.Kasner}
		p_{rad}+ p_x+ p_y =1, \hskip 10 mm p_{rad}^2+ p_x^2+ p_y^2+ 2 p_{\phi}^2 =1.
	\end{equation} The Kasner solution has been used in cosmology as a model of anisotropic Big Bang singularity, in the presence of a so-called stiff fluid (here the scalar field $\phi$). In a recent breakthrough, Fournodavlos--Rodnianski--Speck \cite{FournoRodSpeck} proved that, under a sub-criticality condition on the exponents from \eqref{p.Kasner}, the family \eqref{Kasner.exact} is \textit{stable against perturbations}, with no symmetry assumptions and that the near-singularity dynamics are dominated  by monotonic blow-up.

	\begin{rmk}
		Relation \ref{p.Kasner} may seem unfamiliar, because of the factor $2$ in front of $p_{\phi}^2$. This is due to our definition in \eqref{E3} (adopted in the majority of works in the black hole interior \cite{Christo1,Christo2,MihalisPHD,Kommemi,JonathanStab} ...), as opposed to the standard definition in cosmology where the factor $2$ is absent in \eqref{E3}, see for instance \cite{FournoRodSpeck}. 
	\end{rmk}
	
	\paragraph{Anisotropic Kantowski–Sachs  metrics}
	
	We study \color{black} spherically symmetric \color{black} metrics of the form \begin{equation} \label{Lapse}
		g_{\ep}= -\Omega^2(s) ( ds^2 - dt^2)+ r^2(s) \color{black} d\sigma_{\mathbb{S}^2},
	\end{equation} where $d\sigma_{\mathbb{S}^2}$ is the standard metric on $\mathbb{S}^2$ and $r$ is a scalar function called \textit{the area-radius}, which is geometrically well-defined \color{black} in the spherically-symmetric class. \color{black}   Metrics of the form \eqref{Lapse} \color{black}  are called Kantowski–Sachs cosmologies and are used to model an anisotropic but spatially-homogeneous universe.
	
 More generally, is also possible to consider spatially-homogeneous cosmologies with a 2-surface (other than $\mathbb{S}^2$) \color{black} of constant  curvature $k \in \{-1,0,1\}$ ($\mathbb{H}^2$, $\mathbb{T}^2$, $\mathbb{S}^2$ respectively). \color{black} Our main result will also apply to these cases, with only minor modifications in the proof. We will, however, not pursue that route and stick to Kantowski–Sachs cosmologies.  \color{black}  
	
	\paragraph{Monotonic blow-up due to the scalar field} Our collapse from an almost Cauchy horizon to a spacelike singularity is caused by purely nonlinear dynamics driven by monotonicity. Specifically, the Einstein equations (\eqref{Omegastat}) give a relation where the scalar field dominates and acts as a monotonic source for the lapse $\Omega^2$ from \eqref{Lapse}:
	
	\begin{equation} \label{monotonic1}
		\frac{d^2}{ds^2} \left[\log(\Omega^2)\right](s ) = - 2 (\frac{d\phi}{ds})^2 + ...
	\end{equation}
	
	As we shall explain in Section~\ref{strategy.section},  \eqref{monotonic1} gives rise to a Kasner-type behavior, as given right below, once the lower order terms $...$ have been quantitatively controlled. We emphasize however that the sharp asymptotics of the $- 2 (\frac{d\phi}{ds})^2$ are necessary for the proof (not only its sign!).   \paragraph{Uniform estimates near a (quasi)-Kasner metric}

	Note that, as stated, the Kasner exponents are constants, but there exists generalization of \eqref{Kasner.exact} where the exponents are allowed to depend on $x$, $y$ and $\rho$ (but not on $\tilde{\tau}$). 
	Formally, our metric $g_{\ep}$, written as \eqref{Kasner.intro.eq}, corresponds (up to errors that converge uniformly to zero as $\ep \rightarrow 0$) to a Kasner metric of exponents $p_{rad}= 1-4b_-^2 \cdot \ep^2$,  $p_{x}= p_{y}= 2b_-^2 \cdot \ep^2$, $p_{\phi}= 2b_- \cdot \ep$ (up to a $O(\ep^3)$ error), which satisfy \eqref{p.Kasner} to order $O(\ep^4)$. Nevertheless, the $O(\ep^3)$ error are actually time-dependent (see Theorem~\ref{metric.Kasner.unif}), so \eqref{Kasner.intro.eq} is not necessarily an exact Kasner metric (although the time-dependence in the Kasner exponents is lower order in $\ep$). Our uniform estimates are only valid in a Kasner-time coordinate  $\tilde{\tau}$ (thus, do not contradict the convergence in distribution to Reissner--Nordstr\"{o}m-(dS/AdS) in the other time coordinate $s$ from Section~\ref{RN.stab.intro}) defined so that the metric takes the following product form: \begin{equation} \label{Kasner.time.def}
		g_{\ep} = -d\tilde{\tau}^2 + \underbrace{ \Omega^2(\tilde{\tau}) dt^2+ r^2(\tilde{\tau}) d\sigma_{\mathbb{S}^2}}_{\breve{g}}\color{black},
	\end{equation} where $\breve{g}$ does not involve $d\tilde{\tau}$. This form is often called the synchronous gauge (i.e.\ unitary lapse with zero shift).

	\paragraph{Comparison between the proper time and the area-radius} The spacelike singularity $\mathcal{S}=\{r=0\}$ corresponds to $\{\tilde{\tau}=0\}$. Another consequence of the almost stability of the Reissner--Nordstr\"{o}m-(dS/AdS) Cauchy horizon  is an important discrepancy between the area-radius $r$ (from \eqref{Lapse}) and the proper time $\tilde{\tau}$ (from \eqref{Kasner.time.def})
	\begin{equation*}
		\tilde{\tau}(r) \approx  (\frac{r}{r_-})^{ \frac{\ep^{-2}}{2 b_-^2} + O(\log(\ep^{-1}))},\ \text{or equivalently}\ r(\tilde{\tau}) \approx r_-(M,e,\Lambda) \cdot \tilde{\tau}^{2b_-^2 \ep^2+ O(\ep^3)}.
	\end{equation*}
	The above relation also explains immediately how we find Kasner exponents of the form $(1-O(\ep^2),\ O(\ep^2),\ O(\ep^2))$.

	\paragraph{Convergence to Minkowski in the $L^p_{\tilde{\tau}} $ \color{black} norm}
	The particular case where $p_{rad}=1$, $p_x=p_y=p_{\phi}=0$ in \ref{Kasner.exact}  corresponds (locally) to the Minkowski metric. If $(p_{rad},p_x,p_y,p_{\phi})=(1+\alpha_{rad}(\ep),\ \alpha_{x}(\ep),\ \alpha_{y}(\ep),\ \alpha_{\phi}(\ep))$, and $\alpha_{rad}(\ep)$, $\alpha_{x}(\ep)$, $\alpha_{y}(\ep)$, $\alpha_{\phi}(\ep)$ tend to $0$ as $\ep \rightarrow 0$, the corresponding Kasner metric \eqref{Kasner.exact} converges in $L^p([0,\ep^4]_{\tilde{\tau}})$ to Minkowski, for any $p < +\infty$; however, the convergence is not uniform. We will also show that $g_{\ep}$ converges to Minkowski as well in $L^p([0, \ep^4]_{\tilde{\tau}})$ norm (in particular, it follows from  Theorem~\ref{Kasner.thm}).  It is important to note that this convergence of the metric components expressed using $ \tilde{\tau}$ coordinates (proper time) is consistent with the convergence of the metric components in coordinates $s$ to another limit (namely the Reissner--Nordstr\"{o}m interior) mentioned in Section~\ref{RN.stab.intro}.

	To summarize our analysis from Theorem~\ref{Kasner.thm} gives us\color{black}  $$ (g_{\ep})_{\alpha \beta}= g_{Kas} + \mathfrak{E}^{unif}_{\alpha \beta}=m+ E^K_{\alpha \beta} $$ where $g_{Kas}$ is formally \eqref{Kasner.exact} with Kasner parameters  $(p_{rad}, p_{x}, p_{y}, p_{\phi} )= (1-4b_-^2 \cdot \ep^2, 2b_-^2 \cdot \ep^2, 2b_-^2 \cdot \ep^2, 2b_- \cdot \ep)+O(\ep^3)$,   $m$ is the Minkowski metric,  $\mathfrak{E}^{unif}_{\alpha \beta}$ converges to $0$ in $L^{\infty}([0,\ep^4]_{\tilde{\tau}})$ and $E^{K}_{\alpha \beta}$ converges \footnote{The convergence is implicitly understood to occur for the components of $\mathfrak{E}^{unif}_{\alpha \beta}$, $E^{K}_{\alpha \beta}$  in a Kasner-type frame, see Theorem~\ref{Kasner.thm}.} to $0$ in $L^p([0,\ep^4]_{\tilde{\tau}})$ for any $1<p<+\infty$\color{black}.

		\paragraph{Spatially-homogeneous perturbation of a Schwarzschild-(dS/AdS) spacetime} \label{Schwarz.section}
		
		In Theorem~\ref{thm.intro}, we assume that the spacetime $g_{\ep=0}$ from which $g_{\ep}$ bifurcates is a Reissner--Nordstr\"{o}m-(dS/AdS) metric, with a non-trivial Maxwell field. A natural question is to ask what happens in the limiting case where the Maxwell field is zero, meaning when $g_{\ep=0}$ is the Schwarzschild-(dS/AdS) metric \eqref{gS}. Note that \eqref{gS}, under a coordinate change,  has the same asymptotics (as $r\rightarrow 0$) as  \eqref{Kasner.exact}    with $\phi \equiv 0$ and exponents $(p_{rad},p_x,p_y)=(-\frac{1}{3},\frac{2}{3}, \frac{2}{3} )$. In the absence of a Maxwell field or angular momentum, there is no ``mitigating factor'' to collapse\footnote{In contrast to the data considered in Theorem~\ref{thm.intro}, for which the emergence of a spacelike singularity at late time is more surprising.}, and it is reasonable to expect that the metric will admit a spacelike singularity, and that as $\ep \rightarrow 0$, the space-time would converge in a weak sense to the unperturbed \eqref{gS}.
		
		\begin{rmk} If this expectation is true, then in the uncharged case, the stationary Schwarzschild-AdS black hole has a stable singularity $\mathcal{S}$ with respect to the hairy non-decaying perturbations $\phi$. Theorem~\ref{thm.intro} showed that in the charged case, in contrast, the stationary Reissner--Nordstr\"{o}m-AdS interior solution   is not stable, since the Cauchy horizon of the black hole is  destroyed by arbitrarily small perturbations and replaced by a spacelike singularity.
			
		\end{rmk}

		While the uncharged analogue of $g_{\ep}$ is not covered in our work, numerics \cite{numerics.uncharged} support the above expectations. They also suggest the following important differences between $F \equiv 0$ and the charged case from Theorem~\ref{thm.intro}: \begin{enumerate}[i.]
			\item \label{cont1}The Kasner exponents are bounded away from $0$ (with respect to $\ep$) in the uncharged case, and $p_{rad}<0$.
			
			\item  \label{cont2} As a consequence,  uncharged collapse is not \textit{violent} (in the sense of Section~\ref{violent.section}), contrary to the charged case.
			
			\item  \label{cont3}$\phi$ with data $\ep$  blows up at the singularity as $\phi \simeq \ep \cdot \log(r^{-1})$ as in the linear theory on \eqref{gS} (see Section~\ref{violent.section}).
			
			\item  \label{cont4} The background metric \eqref{gS} is Cauchy-stable for small $\ep$, and there no bifurcation in the stability analysis.
		\end{enumerate}
		We recall that points \ref{cont1}, \ref{cont2}, \ref{cont3}, \ref{cont4} are in sharp contrast with Theorem~\ref{thm.intro} (the charged case), see Section~\ref{RN.stab.intro} and \ref{violent.section}.

		\subsection{Comparison with other matter models and other results in the interior} \label{comparison.section}
		
		In this section, we mention prior results addressing the interior of a black hole that is not converging to Schwarzchild, Reissner--Nordstr\"{o}m, or Kerr. Most of these works are based on either heuristics, or numerics.
		\subsubsection{Numerics on spatially-homogeneous perturbations of the Reissner--Nordstr\"{o}m-AdS interior} \label{hartnoll.section}
		The metric $g_{\ep}$ from Theorem~\ref{thm.intro} was previously investigated in an interesting numerical work \cite{numerics.RN}. The numerical results corroborate entirely Theorem~\ref{thm.intro}; they also  consider the case of non-small perturbations (which are not covered by Theorem~\ref{thm.intro}) and suggest, that even in this case, a Kasner-like singularity forms. \cite{numerics.RN} was preceded by numerics in the uncharged case \cite{numerics.uncharged} (already mentioned in Section~\ref{Kasner.intro.section}) and succeeded by numerics \cite{numerics.charged} studying spatially-homogeneous perturbations of Reissner--Nordstr\"{o}m-AdS in which \textit{the scalar field itself carries a charge} (same model as in \cite{Moi}-\cite{Moi4}). In the charged scalar field case, these numerics suggest intriguing intermediate time oscillations impacting the late-time Kasner singularity, see Section~\ref{open.problems}.
		\subsubsection{Spatially-homogeneous Einstein--Yang--Mills--Higgs interior solutions} \label{EYMH}
		
		The Einstein--$SU(2)$-(magnetic)-Yang--Mills have long been known to admit (non-singular) particle-like asymptotically flat solutions \cite{Bartnick}, \cite{P7}, which are moreover static and spherically symmetric. An equally striking result is the mathematical construction  \cite{P8} of a discrete, infinite family of (asymptotically flat) so-called  \textit{colored black holes} which thus falsify the no-hair conjecture for the $SU(2)$-Yang--Mills matter model (see also \cite{Bizon} for pioneering numerics).
		
		The interior of such black holes is spatially-homogeneous with topology $\RR \times \mathbb{S}^2$,  i.e.\  corresponds to a Kantowski--Sachs cosmology analogous to our spacetime from Theorem~\ref{thm.intro}. Numerical studies in the interior of such black holes \cite{P4}, \cite{P9} have highlighted oscillations and sophisticated dynamics,  in which power law mass inflation/Kasner-type behavior \emph{alternates} with the almost formation of a Cauchy horizon (analogous to the one discussed in Section~\ref{RN.stab.intro}).

		The Einstein--$SU(2)$-Yang--Mills-\emph{Higgs} model, in which a massive scalar field analogous to \eqref{E5} is added, has also been studied numerically  \cite{P2,P4,P8}. These studies suggest that the dynamics are drastically changed by the Higgs field which suppresses the oscillations, and the singularity is spacelike with a power-law mass inflation: \begin{equation} \label{lambda.YM}
			\phi(r) \approx \log(r^{-\lambda}), \hskip 5 mm \rho(r) \approx r^{-\lambda},
		\end{equation}  for some $\lambda>1$, which is consistent the behavior  (\eqref{rho.intro}, \eqref{intro.violent1}) of  $g_{\ep}$ in our model. However, in these numerics, it is not clear how the constant $\lambda$ relates to the size of the initial data, and whether the \textit{violent nonlinear collapse scenario} that we put forth applies or not (recall that in our case, data proportional to $\ep$ give a rate $\lambda(\ep) \simeq \ep^{-2}$).

		\subsection{Directions for further studies} \label{open.problems}
		In this section, we provide open problems that our new Theorem~\ref{thm.intro} has prompted, and their underlying motivation.
		\subsubsection{Extensions of Theorem~\ref{thm.intro} for the same matter model}
		We define the temperature $T$ of a  Reissner--Nordstr\"{o}m black hole as $T(M,e,\Lambda):=2K_+$ the surface gravity of the event horizon.  $T(M,e,\Lambda)>0$ for sub-extremal parameters, and $T(M,e,\Lambda)=0$ corresponds to an extremal Reissner--Nordstr\"{o}m black hole.
		Theorem~\ref{thm.intro} applies to perturbations of a sub-extremal Reissner--Nordstr\"{o}m black hole with fixed parameters  $(M,e,\Lambda)$ hence $T\simeq 1$. Choosing $\ep$-dependent parameters $(M,e,\Lambda)$ allows to formulate:
		
		\begin{open}
			Study spacetimes as in Theorem~\ref{thm.intro} assuming that the parameters depend on $\ep$ and are extremal at the limit in the sense that $T(M_{\ep},e_{\ep},\Lambda_{\ep}) \rightarrow 0$ as $\ep \rightarrow 0$.
		\end{open}
		
		The numerics \cite{numerics.RN} suggest that the important scaling is $\frac{|\phi|_{|\mathcal{H}^+}}{T}$, meaning that our spacetime $g_{\ep}$, for which $|\phi|_{|\mathcal{H}^+} = \ep$ and $T \simeq 1$, should be similar to a spacetime with data $|\phi|_{|\mathcal{H}^+} \simeq  1$ and $T \simeq \ep^{-1}$. If this is true, then the spacetime with data $|\phi|_{|\mathcal{H}^+} \simeq  \ep^2$ and $T \simeq \ep$ (which is extremal at the limit $\ep \rightarrow 0$) behaves in the same way as $g_{\ep}$.

		Another natural possible direction to extend Theorem~\ref{thm.intro} is to relax the smallness assumption on the data $\phi$:
		\begin{open}
			Study spacetimes as in Theorem~\ref{thm.intro} for large data i.e.\ without assuming that $\ep$ is small.
		\end{open}
		
		We emphasize that the resolution of this problem does not follow immediately from the techniques of Theorem~\ref{thm.intro}. Nevertheless, numerics \cite{numerics.RN} suggest that, even in the case of large perturbations, the singularity is still spacelike.

		Lastly, recall that Theorem~\ref{thm.intro} only applies to \textit{almost every} but not all parameters $(M,e,\Lambda,m^2)$. This is because, for discrete values of  mass $m^2$ (so-called non-resonant masses, see Section~\ref{linear.section}), $b_-(M,e,\Lambda,m^2)=0$ and $\phi$ ``degenerates in the linear theory''. It would be interesting to see understand the singularity of $g_{\ep}$ in this case:
		
		\begin{open}
			Study spacetimes as in Theorem~\ref{thm.intro} for exceptional $(M,e,\Lambda,m^2)$ such that $b_-(M,e,\Lambda,m^2)=0$.
		\end{open}
		
		We emphasize that, from the point of view of genericity, these exceptional values, which form a set of zero measure (in fact, the zero set of a holomorphic function on $\mathbb{C}^4$) are most likely irrelevant. The case $b_-=0$ is discussed in Section 3.2 in \cite{numerics.RN}, and it is argued numerically  that the black hole terminal boundary is a Cauchy horizon  for at least countably many values of the parameters. This is consistent with Theorem~\ref{thm.intro}, and would indicate that the statement for ``almost all parameters'' in Theorem~\ref{thm.intro} cannot be improved into ``for all parameters''.

		\subsubsection{Static black hole exteriors with AdS boundary conditions}
		
		As we explained in Section~\ref{ADS.section}, our spacetime  $(\mathcal{M}_{\ep},g_{\ep})$ finds potential applications to the AdS/CFT theory. The idea is to prescribe boundary data on an asymptotically AdS boundary in the black hole exterior, which we must also construct. This construction should be much simpler in the exterior, in the absence of any singularity (as opposed to the interior where the question of stability is more subtle, and varies at different time scales, see Section~\ref{RN.stab.intro} and \ref{violent.section}).
		Based on the numerics from \cite{numerics.RN}, the following problem seems reasonable to formulate and accessible:\color{black}
		\begin{open} \label{open4}
			Construct a static two-ended black hole exterior with appropriate data on the AdS boundary, such that the spacetime $(\mathcal{M}_{\ep},g_{\ep})$ in Theorem~\ref{thm.intro} is the corresponding black hole interior, as depicted in Figure \ref{Fig_Two}.
		\end{open}
		Numerics \cite{numerics.uncharged,numerics.RN} indicate that indeed, the metric $g_{\ep}$ is the black hole interior region corresponding to an asymptotically AdS black hole with the following asymptotics towards the AdS boundary $\{r=\infty\}$: $$ \phi(r) = \frac{\delta_\ep}{r^{\frac{3}{2}-\sqrt{\frac{9}{4}+m^2}}} + o(\frac{1}{r^{\frac{3}{2}-\sqrt{\frac{9}{4}+m^2}}}), \text{ as } r \rightarrow \infty$$ where $\delta_\ep \neq 0$ is small constant and $m^2<0$. This corresponds to a choice of Neumann boundary conditions, see \cite{Claude}.

		\begin{figure} 		
			
			\begin{center}
				
				\includegraphics[width=91 mm, height=70 mm]{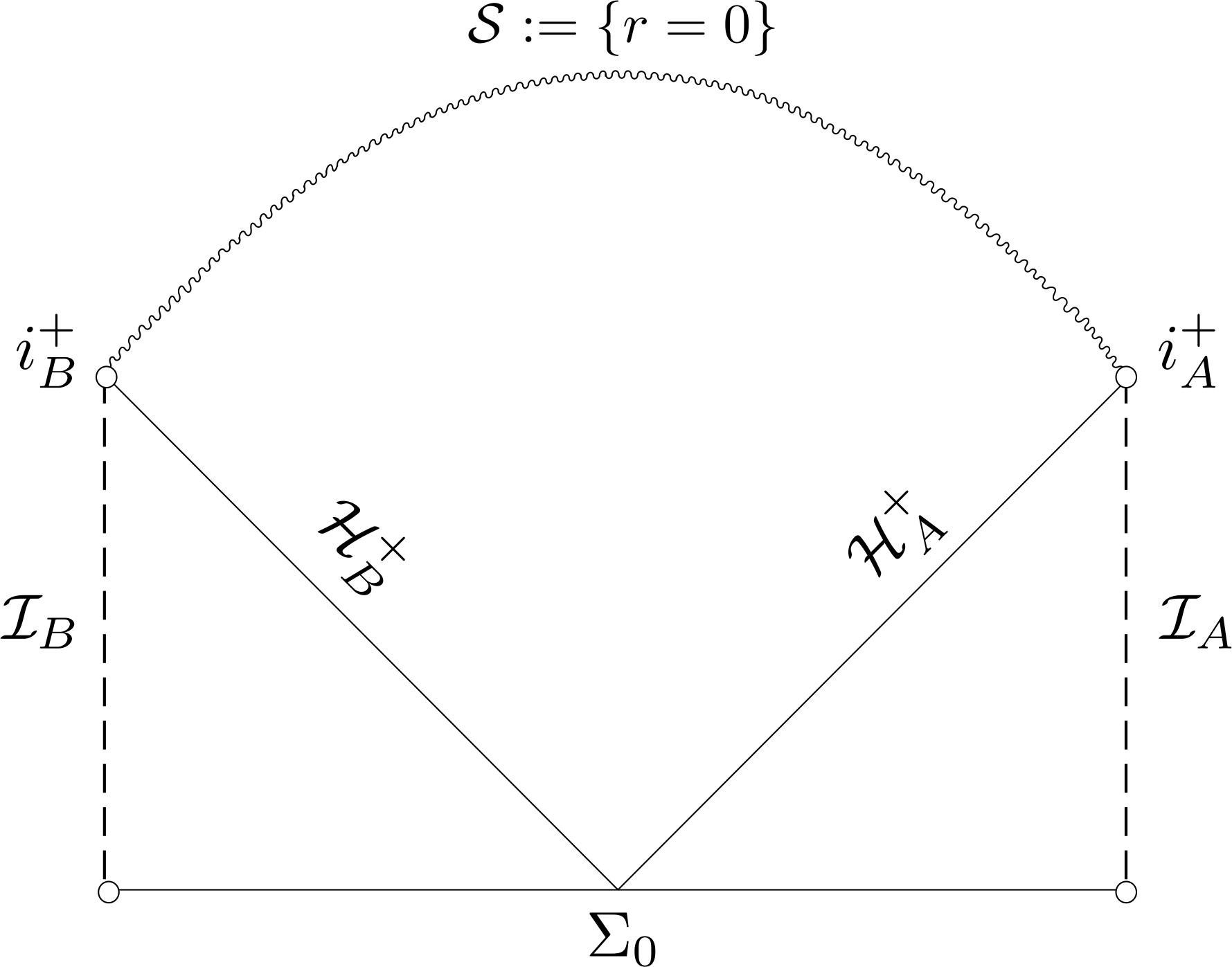}

			\end{center}

			\caption{Penrose diagram of the asymptotically AdS two-ended black hole  resolving Open Problem \ref{open4}}
			\label{Fig_Two}
		\end{figure}

		\subsubsection{Stability of $g_{\ep}$  with respect to non-homogeneous perturbations}
		After solving Open Problem \ref{open4} and obtaining the metric $g_{\ep}$ covering both the black hole interior and exterior with AdS boundary conditions, a natural (and more demanding) problem  to ask is the question of stability of $g_{\ep}$:
		\begin{open} \label{openv}
			Study the stability of $g_{\ep}$  against small non-static perturbations for the initial/boundary value problem.
		\end{open}
		
		Nonlinear stability problems with AdS boundary conditions are notoriously difficult, despite recent remarkable progress in spherical symmetry \cite{G1,G2,G3} (in the direction of instability, however). A more accessible question in the direction of Open Problem \ref{openv} is to first study \emph{only perturbations of the interior}, for  instance concrete data of the following form on the event horizon, with $f(v,\theta,\varphi) \rightarrow 0$ as $v\rightarrow +\infty$: \begin{equation} \label{pertb}
			\phi_{|\mathcal{H}^+}(v,\theta,\varphi)= \ep+ \ep^{2} \cdot f(v,\theta,\varphi).
		\end{equation} 
		
		In particular, in the case where $f(v,\theta,\varphi)=f(v)$ is spherically symmetric, it seems reasonable to expect that the quantitative methods of the present paper can be generalized to yield the stability of the interior region of $g_{\ep}$, at least if $f$ converges sufficiently fast to $0$. The non spherically-symmetric is expectedly more difficult, even though the recent breakthrough \cite{FournoRodSpeck}  proving the stability of Kasner for positive exponents (hence including the exponents that we obtain, see Section~\ref{Kasner.intro.section}) could help in this direction, at least close enough to the singularity.

		We must also point out an important connection: solving Open Problem \ref{openv}, even in spherical symmetry, will yield valuable insight on spacelike singularities even inside the \emph{non-hairy black holes} of Figure~\ref{Fig.relax} that converge to Reissner--Nordstr\"{o}m. First, recall that, despite Theorem~\ref{CH.thm} asserting the existence of a non-empty Cauchy horizon $\mathcal{CH}^+$ for such black holes, it is conjectured that there additionally exists  a spacelike singularity $\mathcal{S}=\{r=0\}$ as soon as the scalar field amplitude is large enough in the two-ended\footnote{In the more realistic one-ended case corresponding to gravitational collapse, the situation is even more drastic: a singularity $\mathcal{S}=\{r=0\}$ must always exist \cite{r=0}. Studying this case, however, requires the scalar field to be charged, and thus to generalize Theorem~\ref{maintheorem}  to  the charged model. Solving  Open Problem~\ref{openv} in this context will thus yield valuable insights on spacelike singularities in gravitational collapse.} case \cite{Mihalisnospacelike}.  
		
		By the domain of dependence, such a spacelike singularity will emerge from event horizon data which are locally identical to the data allowed  by Open Problem~\ref{openv}. Thus, it is enough to glue\footnote{However, the null constraint forbid gluing scalar field data $\phi \equiv \ep$ on a compact set to the decaying-in-time tail needed for Theorem~\ref{CH.thm}.} event horizon data corresponding to Open Problem~\ref{openv} on a compact set to decaying-in-time asymptotic tails for $\phi$ as needed for Theorem~\ref{CH.thm} to construct a large class of two-ended black holes with both a spacelike singularity $\mathcal{S}=\{r=0\}$ quantitatively described by Theorem~\ref{maintheorem} (more precisely, its analogue upon a successful resolution of Open Problem~\ref{openv}) and a null Cauchy horizon $\mathcal{CH}^+$.

		\subsubsection{Rotating hairy black holes}
		
		A remarkable mathematical construction of a one-parameter family of rotating hairy black hole exteriors namely stationary, two-ended asymptotically flat solutions of the Einstein--Klein--Gordon system \eqref{E1}-\eqref{E5} (with $F\equiv 0$) that bifurcate off the Kerr metric has been carried out by Chodosh--Shlapentokh-Rothman \cite{OtisYakov}. Nevertheless, the interior of such hairy black holes has never been studied, and it is not known whether a Cauchy horizon exists as for the Kerr metric, or if it is replaced by a spacelike singularity as for $g_{\ep}$. This prompts the following problem:
		\begin{open}
			Characterize the singularity inside the Chodosh--Shlapentokh-Rothman hairy black holes.
		\end{open}

		We plan to return to this problem in a future work. 
		
		\subsubsection{Other matter models}
		\paragraph{Models where the singularity is conjectured to be Kasner-like} 
		
		We already mentioned in detail in Section~\ref{Schwarz.section} the case of perturbations of Schwarzschild-(dS/AdS) and the main difference compared to the setting of Theorem~\ref{thm.intro}, together with the numerics \cite{numerics.uncharged} suggesting that the singularity is Kasner-like. This leads us to formulate the 
		\begin{open}
			Study spatially-homogeneous perturbations of  Schwarzschild-(dS/AdS) black holes.
		\end{open}

		As we explained in Section~\ref{EYMH}, the Einstein equations coupled with  the SU(2)-Yang--Mills--Higgs equations also admit spatially-homogeneous black hole interiors, and numerics \cite{P4,P2,P8} suggest a Kasner-like singularity:
		
		\begin{open}
			Study spatially-homogeneous solutions for the Einstein-SU(2)-Yang--Mills--Higgs equations.
		\end{open}
		
		\paragraph{Models where oscillations are conjectured}
		
		As we explained in Section~\ref{EYMH}, without a Higgs field, numerics suggest that the Einstein-SU(2)-Yang--Mills black holes admit a chaotic behavior in the interior  \cite{P4,P9}. We are hopeful that the methods we  developed in the present paper could also be adapted to address the following 
		\begin{open}
			Study spatially-homogeneous solutions for the Einstein-SU(2)-Yang--Mills equations.
		\end{open}
		
		More closely related to our work is the generalization of Theorem~\ref{thm.intro} with the same data but for a more general model where the scalar field is \textit{allowed to be charged}. In this case, the coupling with the Maxwell field is non trivial, and sourced by the charged scalar field. (This model was studied by the author \cite{Moi3Christoph,Moi4Christoph,Moi,Moi2,Moi4},  in a different setting where the scalar field decays to zero, instead of being a small constant (see Section~\ref{AF.comparision})). Numerics \cite{numerics.charged} suggest very interesting dynamics in a certain parameter range for this system, in which oscillations precede the formation of a Kasner-like singularity with $p_{rad}<0$ (different from our scenario, recall Section~\ref{Kasner.intro.section}) which is then inverted to a Kasner-like singularity with $p_{rad}>0$, a phenomenon the authors of \cite{numerics.charged} call ``Kasner-inversion''. \color{black} Interestingly, similar Kasner inversions also play an important role in the standard cosmological picture given by the BKL scenario \cite{BKL1,BKL2}. \color{black}
		\begin{open}
			Study data as in Theorem~\ref{thm.intro} for the Einstein--Maxwell-\emph{charged}--Klein--Gordon equations.
		\end{open}
		
		It would be extremely interesting to adapt the methods of the present paper to the setting of a charged scalar field, a problem we hope to return to in the future.

		\subsection{Strategy of the proof}\label{strategy.section}
		\subsubsection{Set-up of the problem and general description of the strategy} \label{setup.intro}
		The data described in Theorem~\ref{thm.intro} give rise to a \textit{spatially-homogeneous} spacetime $\mathcal{M}_{\ep}=\RR \times \mathbb{S}^2 \times (0,r_+(M,e,\Lambda)]_r$, which we can describe as a superposition of $\RR \times \mathbb{S}^2$ cylinders with variable area-radius $r$. The EMKG system \eqref{E1}-\eqref{E5} then reduces to a coupled system of nonlinear ODEs  (see Section~\ref{ODEsection}), that we analyze in physical space.
		
		One of the main ideas behind the proof is to divide the spacetime into regions of the form $\{ r_{inf}(\ep)< r < r_{\sup}(\ep)\}$ (or ``epochs'' since $r$ is a timelike coordinate in the interior). We will express these epochs using a time-coordinate $s$ (see \eqref{gauge} for a precise definition) which is monotonic with respect to $r$ (that is, $\frac{dr}{ds}<0$). $s \in (-\infty,s_{\infty}(\ep))$ can be thought of as a generalization of the standard $r^* \in \RR$ coordinate on Reissner--Nordstr\"{o}m. The data $\phi \equiv \ep \neq 0$ are posed on the event horizon $\mathcal{H}^+:=\{s=-\infty\}$ (Section~\ref{data.section}) and we split the spacetime into (see Figure \ref{Fig3})
		
		\begin{enumerate}
			\item the red-shift region $\mathcal{R}:= \{ -\infty < s \leq -\Delta(M,e,\Lambda,m^2) \}$, for $\Delta(M,e,\Lambda,m^2)>0$ sufficiently large.
			\item the no-shift region  $\mathcal{N}:= \{ -\Delta < s \leq \Delta_0(\ep) \}$,  $\Delta_0(\ep):= c(M,e,\Lambda,m^2) \cdot \log(\ep^{-1})$ for  $c(M,e,\Lambda,m^2)>0$.
			\item the early blue-shift region  $\mathcal{EB}:= \{ S''(M,e,\Lambda)< s \leq \s(\epsilon) \}$,  $\s(\epsilon):=\frac{\log(\nu(M,e,\Lambda,m^2) \cdot \epsilon^{-2})}{2|K_-|}$, where $\nu(M,e,\Lambda,m^2)>~0$ small enough and $S''(M,e,\Lambda)>0$ are fixed constants that will be determined in the proof\color{black}, and $2K_-(M,e,\Lambda)$ is the surface gravity of the Reissner--Nordstr\"{o}m-(dS/AdS) Cauchy horizon.
			\item the late blue-shift region  $\mathcal{LB}:= \{ \s(\epsilon) < s \leq \color{black} \delta_{\mathcal{C}}\color{black} \cdot \ep^{-1}\}$, where   $\delta_{\mathcal{C}}\color{black}(M,e,\Lambda,m^2)>0$. \color{black}
			\item the crushing region  $\mathcal{C}:= \{  4\s(\ep) < s < s_{\infty} \}$, and $s_{\infty}(\epsilon):= \frac{\ep^{-2}+O(\color{black}\log(\ep^{-1}))}{ 4|K_-| \cdot b_-^2}<+\infty$.  
			\item The spacelike singularity $ \mathcal{S}=\{r=0\}=\{s=s_{\infty}(\ep)\}$. 
			
		\end{enumerate}
		
		\begin{rmk}
			Note that the region  \color{black} $\mathcal{EB}$ overlaps with $\mathcal{N}$ (for small enough $\ep\color{black}>0$) and $\mathcal{LB}$ overlaps with $\mathcal{C}$. Nevertheless, our estimates in $\mathcal{EB}$ (respectively $\mathcal{LB}$) are more precise than those in $\mathcal{N}$  (respectively $\mathcal{C}$) on the overlap $\mathcal{EB}\cap \mathcal{N}$     (respectively $\mathcal{LB}\cap \mathcal{C}$). We consider  $\mathcal{C}$ in isolation because it is the largest region where the metric is Kasner-like.  \color{black} 
			
		\end{rmk}
		
		Each region is characterized by a particular dynamical regime \footnote{As a guiding principle, the early regions are in linear regime and the late ones are driven by the nonlinearity.} in which specific techniques apply, see below.

		To estimate the EMKG system, we use a standard bootstrap method (which is crucial, in view of the nonlinearities) and prove weighted estimates (drastically different in each region) integrating  along the time direction.
		
		In fact, we expect most of our estimates to generalize  in the non-homogeneous situations (see Section~\ref{open.problems}). We emphasize that our methods are entirely quantitative (and are inspired from those in \cite{Moi} which deal with a spacetime with fewer symmetries, albeit with simpler dynamics) and that, even though $g_{\ep}$ are analytic spacetimes, \textit{we do not rely on analyticity at any stage}.

		Before turning to the nonlinear estimates, we will mention in Section~\ref{linear.section} a prior linear result  proven in \cite{ChristophYakov} that is  useful in the  early blue-shift region $\mathcal{EB}$ (where the dynamics are still mostly linear) and the late blue-shift region $\mathcal{LB}$\color{black}.
		\subsubsection{Preliminary estimates on the corresponding  linearized solution} \label{linear.section}
		
		To capture the late time dynamics and the ultimate formation of the singularity (see Section~\ref{NLintro}), we require sharp estimates on the scalar field $\phi$. It is well-known that such sharp estimates cannot be obtained exclusively  through physical spaces methods (see the discussion in \cite{Moi3Christoph}).  Moreover, even for the \textit{linear wave equation} $\Box_{g_{RN}} \phil = m^2 \phil$ on a fixed Reissner--Nordstr\"{o}m-(dS/AdS), the late time behavior \textit{depends} on the parameters $(M,e,\Lambda,m^2)$ (see \cite{ChristophYakov}).

		In our specific setting, the linearized version of $\phi$ from Theorem~\ref{thm.intro} corresponds to $\phil^{(\ep)} = \ep \cdot  \phil^{(1)}$, where   $\phil^{(1)}$  is a solution of $\Box_{g_{RN}} \phil^{(1)} = m^2 \phil^{(1)}$ with constant data $ \phil^{(1)}=1$ on $\mathcal{H}^+$. $ \phil^{(1)}$ solves a linear radial ODE in $r^*$ and \begin{equation*}
			\phil^{(1)}(r^*)=  A(M,e,\Lambda,m^2) \cdot \tilde{v}_1(r^*)+  B(M,e,\Lambda,m^2)\cdot  \tilde{v}_2(r^*),\ \text{where } \tilde{v}_1(r^*) \simeq 1,\ \tilde{v}_2(r^*) \simeq r^* \text{ as } r^* \rightarrow +\infty.
		\end{equation*}
		
		One of the results of \cite{ChristophYakov} is to show that for all sub-extremal parameters $(M,e,\Lambda)$ and almost every masses $m^2$: \begin{equation} \label{Res}
			B(M,e,\Lambda,m^2) \neq 0.\tag{res}
		\end{equation} a condition that we will refer to as a \textit{scattering resonance}. When  \eqref{Res} is satisfied, $\phil^{(1)}(r^*) \approx r^*$ near the  Reissner--Nordstr\"{o}m-(dS/AdS) Cauchy horizon.  The exceptional set $D(M,e,\Lambda)\subset \RR$ of masses $m^2$ for which \eqref{Res} is not true is called the set of \textit{non-resonant masses} and it was shown in \cite{ChristophYakov} that $0 \in D(M,e,\Lambda)$, hence $D(M,e,\Lambda)$ is non-empty.

		In the early regions (specifically $\mathcal{R}$, $\mathcal{N}$ and $\mathcal{EB}$), we will estimate the difference between the actual solution $\phi$ and the linearized $\phil$ (see Section~\ref{RSNintro} for details) and show an estimate of the following form: for all $s \ls \log(\ep^{-1})$ \begin{equation} \label{diff.int} \tag{diff}
			|\color{black}\frac{d\phi}{ds}\color{black}(s) - \color{black}\frac{d\phil}{ds}\color{black}(s)| \ls \ep^3 \color{black} \cdot \log(\ep^{-1})\color{black}.
		\end{equation}
		The bound \eqref{diff.int}, however, is no longer true in future regions $\mathcal{LB}$ (for $s\geq 4  \s(\ep)$) \color{black} and $\mathcal{C}$ because the EMKG solution is dominated by the nonlinear regime at late time. The condition \eqref{Res}, in turn, is useless in $\mathcal{R}$ and $\mathcal{N}$, as it dictates the asymptotics of $\phil$ \textit{close enough} to the Cauchy horizon. But in $\mathcal{EB}$, two conditions are reunited to combine \eqref{diff.int} and \eqref{Res}: the EMKG system is still in the linear regime, and the metric is close (i.e.\ $s \sim \log(\ep^{-1})$ is large) to the Reissner--Nordstr\"{o}m-(dS/AdS) Cauchy horizon (see Section~\ref{RN.stab.intro} and Section~\ref{EB.intro.section}). Thus,  \eqref{diff.int} and \eqref{Res}  give sharp asymptotics in $\mathcal{EB}$: \begin{equation} \label{as}
			|\color{black}\frac{d\phi}{ds}\color{black}(s)- B(M,e,\Lambda,m^2) \cdot \ep | \ls \ep^3 \color{black}\log(\epsilon^{-1})\color{black}. 
		\end{equation}

		In the later regions $\mathcal{LB}$ and $\mathcal{C}$, we will build on \eqref{as} to construct the (genuinely nonlinear at this point) dynamics of $\phi$ which critically affects the behavior of the metric $g_{\ep}$, especially the spacelike singularity formation. In particular, the Kasner exponents depend (Section~\ref{Kasner.intro.section}) on what we call the \textit{resonance parameter} $b_- \in \RR$ \begin{equation} \label{b-} \tag{$b_-$}
			b_-(M,e,\Lambda,m^2):= \frac{B(M,e,\Lambda,m^2)}{2|K_-| (M,e,\Lambda)} \neq 0,
		\end{equation} where $2K_-(M,e,\Lambda)<0$  is the surface gravity  of the unperturbed Cauchy horizon (recall $b_-$ is used  in Theorem~\ref{thm.intro}).

		\subsubsection{The red-shift region $\mathcal{R}$ and the no-shift region $\mathcal{N}$: Cauchy stability regime} \label{RSNintro}
		
		Recall the definition of the regions $\mathcal{R}$ and $\mathcal{N}$ from Section~\ref{setup.intro}. These regions are easier for the following reasons:
		
		\begin{enumerate}
			\item The range of $s$ in $\mathcal{R}$ is infinite, but the spacetime volume is small. Moreover, one can exploit the classical \textit{red-shift estimates} for the wave equation (and, by extension, for the EMKG system).
			\item  The range of $s$ in $\mathcal{N}$ is finite, but large (i.e.\ $O(\log(\ep^{-1}))$). Nevertheless, the range of $s$ to the future of $\mathcal{N}$ is $O(\ep^{-2})$ and $\log(\ep^{-1}) \ll \ep^{-2}$ so $\mathcal{N}$ is considered to be ``a small region'', compared to the rest of the spacetime.
		\end{enumerate}

		In particular, in these two regions, the principle of \textit{Cauchy stability}, which roughly states that $\ep$ perturbations in the data give rise to $O(\ep^2)$ perturbations for the solutions of quadratic nonlinear PDEs, prevails.
		
		Concretely, for  all $ s \in \mathcal{R}:= \{ -\infty < s \leq -\Delta\}$ for $\Delta(M,e,\Lambda,m^2)$ large, we show the schematic estimates: \begin{align}
			&  |g(s)-g_{RN}(s)| \lesssim \epsilon^2 \cdot e^{2K_+(M,e,\Lambda) s} \ll \ep^2, \tag{R1}\\ & |\frac{dg}{ds}(s)-\frac{dg_{RN}}{ds}(s)| \lesssim \epsilon^2 \cdot e^{2K_+(M,e,\Lambda) s} \ll \ep^2, \tag{R2}\\ & |\kappa^{-1}(s)-1| \lesssim \epsilon^2 \cdot e^{2K_+(M,e,\Lambda) s} \ll  \ep^2 \label{R3}\tag{R3},\\ & |\phi(s)-\phil(s)| \lesssim \epsilon^3 \cdot e^{2K_+(M,e,\Lambda) s} \ll  \ep^3 \tag{R4},
		\end{align}where $\kappa^{-1}:= \frac{-\frac{dr}{ds}}{\Omega^2}$ ($\kappa^{-1}=1$ on Reissner--Nordstr\"{o}m-(dS/AdS)) and $2K_+(M,e,\Lambda)>0$ is  the surface gravity of the event horizon. The factor $e^{2K_+ s}$  is the sign of red-shift, and requiring $\Delta$ to be large gives $e^{2K_+ s} \ll 1$ for $s < -\Delta$.

		In the no-shift region $\mathcal{N}$ $:= \{ -\Delta < s \leq   c(M,e,\Lambda,m^2) \cdot \log(\ep^{-1})\}$, we use Gr\"{o}nwall estimates: for all $s \in \mathcal{N}$ \begin{align*}
			&  |g(s)-g_{RN}(s)|  \lesssim C^{s} \cdot  \epsilon^2 \lesssim \epsilon , \label{N1.intro}\tag{N1}\\ & |\frac{dg}{ds}(s)-\frac{dg_{RN}}{ds}(s)|  \lesssim C^{s} \cdot  \epsilon^2 \lesssim \epsilon ,   \label{N2.intro}\tag{N2}\\ & |\phi(s)-\phil(s)|  \lesssim C^{s} \cdot  \epsilon^3 \lesssim \epsilon^2  \label{N3.intro} \tag{N3},
		\end{align*}  where $C(M,e,\Lambda,m^2):=e^{c^{-1}}$. In the above estimates, the term $C^s \ls \ep^{-1}$ is indeed the sign of the loss incurred by the use of Gr\"{o}nwall. To avoid this loss, we will only use the above for $s \leq S''(M,e,\Lambda,m^2)$ where $S''(M,e,\Lambda,m^2)>0$ is a large constant but \textit{independent of $\ep$} (recall that the past boundary of the next region $\mathcal{EB}$ is at $s=S''(M,e,\Lambda,m^2)$).

		Completing these steps provides estimates on the \textit{difference} between the dynamical metric $g_{\ep}$ and its unperturbed background $g_{RN}$  (recall $g_{\ep}=g_{RN}$ on the data, see Theorem~\ref{thm.intro}) in the strong $C^1$ norm, up to $\mathcal{EB}$. This explains the stability of the Cauchy horizon claimed in Theorem~\ref{thm.intro}; in particular, for any constant $\Gamma>0$, $g_{\ep}$ is uniformly close to $g_{RN}$ for  $s< \Gamma$, for $|\ep|$ small enough (we will extend these estimates up to $s\ll\ep^{-1}$, but they fail for $s \simeq \ep^{-2}$).
		
		\begin{rmk}
			Going forward, we will actually not directly need $g-g_{RN}$ estimates in $\mathcal{C}$ to prove the formation of a spacelike singularity: nevertheless, we will need \eqref{as}, for which we do need  $g-g_{RN}$ to be small in the past of $\mathcal{EB}$. Thus, there is no obvious way to avoid difference estimates on the metric to show that the singularity is spacelike.
		\end{rmk}

		\subsubsection{The early blue-shift region $\mathcal{EB}$: Reissner--Nordstr\"{o}m-(dS/AdS) stability in strong norms} \label{EB.intro.section}
		
		In the early blue-shift region  $\mathcal{EB}= \{ S''(M,e,\Lambda,m^2)< s \leq \s(\epsilon) \} \subset \mathcal{N}$, \textit{Cauchy stability}, as defined earlier, does not give sharp bounds, so we look for different paths towards improvement.  We will exploit the so-called \textit{blue-shift effect}, which traditionally occurs at the Cauchy horizon of Reissner--Nordstr\"{o}m-(dS/AdS) for large $s$. To understand its manifestation, recall the lapse $\Omega^2$ from \eqref{Lapse}: the typical estimate we obtain in this region is $$ \bigl| \frac{d}{ds} \log(\Omega^2)(s) - 2K_-(M,e,\Lambda) \bigr| \ls e^{2K_- s} \ll \frac{|K_-|}{100},$$ where $2K_-(M,e,\Lambda) <0$ is the surface gravity of the Reissner--Nordstr\"{o}m-(dS/AdS) Cauchy horizon and $S''(M,e,\Lambda,m^2)>0$ is chosen to be sufficiently large so that the last inequality holds. This enforces a behavior of the following form: \begin{equation} \label{Bls}\tag{blue-shift}
			\Omega^2(s) \approx e^{2K_- s},
		\end{equation} an estimate which is also valid  on Reissner--Nordstr\"{o}m-(dS/AdS), close enough to the Cauchy horizon. Since $e^{2K_- S''} \ll~1$, $\Omega^2$ is small, which is helpful to close difference estimates with an (inconsequential) logarithmic loss (compare with the much worse estimates from \eqref{N1.intro}, \eqref{N2.intro}, \eqref{N3.intro}), i.e.\ for all $s \in \mathcal{EB}$:
		\begin{align}
			&  |g(s)-g_{RN}(s)| \lesssim \epsilon^2 \cdot \color{black} s^2 \color{black} \ls \ep^2 \cdot \log(\ep^{-1})^2 ,\tag{EB1}\\ & |\frac{dg}{ds}(s)-\frac{dg_{RN}}{ds}(s)|   \lesssim \epsilon^2 \cdot s \ls \ep^2 \cdot \log(\ep^{-1}), \tag{EB2}\\ & |\phi(s)-\phil(s)| \lesssim \epsilon^3 \cdot \color{black} s^2 \color{black} \ls \ep^3 \cdot \log(\ep^{-1})^2 .\tag{EB4} \label{EB4.intro}
		\end{align}
		
		The analogue of \eqref{R3} is different; using \eqref{Bls} and the Raychaudhuri equation \eqref{Raychstat}, we get for all $s\in \mathcal{EB}$  \begin{align}
			|\kappa^{-1}(s)-1| \ls \frac{\ep^2}{\Omega^2(s)} \approx \ep^2 \cdot e^{2 |K_-| s} < 0.1, \label{kappa.intro}
		\end{align}  and  the future boundary $\{s=\s \simeq \log(\ep^{-1})\}$ of $\mathcal{EB}$ is  chosen so that the RHS of \eqref{kappa.intro} is $<0.1$, hence $\Omega^2(\s) \approx \ep^2$.

		%\begin{rmk} \label{blue-shift.instab}	Going forward, we note that in $\mathcal{N}-\mathcal{EB}$ (and also in $\mathcal{LB}$), $\kappa^{-1}$ becomes larger and larger, a sign of the classical blue-shift instability (see Section~\ref{LBintro}). This growth is an obstruction to showing sharp difference bounds, which explains why we isolate $\mathcal{EB}$ inside $\mathcal{N}$, and also prompts the use of a new strategy in the next region $\mathcal{LB}$.\end{rmk}

		\subsubsection{The late blue-shift region $\mathcal{LB}$: Reissner--Nordstr\"{o}m-(dS/AdS) stability in weak norms} \label{LBintro}
		Unlike its predecessors, the late blue-shift region  $\mathcal{LB}:= \{ \s(\epsilon)< s \leq \delta_{\mathcal{C}} \cdot  \epsilon^{-1}\color{black}\}$ has  a large $s$-range of order $O(\epsilon^{-1})$ and Cauchy stability \textit{utterly fails}, two facts that render the estimates noticeably more delicate. The failure of Reissner--Nordstr\"{o}m-(dS/AdS) stability in strong $C^1$ norm can be expressed for instance in the following estimate: \begin{equation} \label{kappa.large}
			\kappa^{-1}(s) \approx  \ep^2 \cdot e^{2 |K_-| s}\ \Rightarrow\   \color{black}\kappa^{-1}( \epsilon^{-1}) \approx  \ep^2 \cdot e^{O(\epsilon^{-1})} \gg 1,\color{black}
		\end{equation} and a similar estimate holds for the Hawking mass $\rho$, which becomes large. Relatedly, in the non-hairy case when $\phi$ tends to $0$ (discussed in Section~\ref{AF.comparision}), it is known that mass inflation is caused by the blue-shift effect %(black holes converging to Reissner--Nordstr\"{o}m  
		(see \cite{Moi4}); we see some reminiscence of this phenomenon in the hairy case as reflected by \eqref{kappa.large} and the related estimate on  $\rho$.

		As should be clear, difference estimates (as we had earlier) cannot be proved in $\mathcal{LB}$. However, we still obtain some form of weak  stability in the $C^0$ norm of Reissner--Nordstr\"{o}m-(dS/AdS): recalling the definition of $r$ from \eqref{Lapse} \begin{align}
			&  |r(s)-r_-(M,e,\Lambda)| \lesssim \epsilon^2 \cdot s \ls \color{black} \ep \color{black}\leq \frac{r_-(M,e,\Lambda)}{100}, \tag{LB1}\\ & \Omega^2(s) \approx e^{2K_- s}\tag{LB1},\\ & |\phi(s)-  A(M,e,\Lambda,m^2) \cdot \ep- B(M,e,\Lambda,m^2) \cdot \ep \cdot s | \lesssim \color{black}\ep^3\color{black} \cdot s^2  \ls   \color{black}\ep^2 s\color{black}\label{LB3.intro}\tag{LB3}.
		\end{align}   At the first order we get a much sharper control,  which is more useful than  \eqref{LB3.intro} in practice\color{black}: \begin{align}
			\bigl| r^2(s) \cdot \frac{d\phi}{ds}(s)- r_-^2 \cdot B(M,e,\Lambda,m^2) \cdot \ep \bigr| \ls \ep^3 \label{LB4.intro} \tag{$\dot{\phi}$}.
		\end{align}% (but differing, however, from Reissner--Nordstr\"{o}m(dS/AdS) !)
		Going forward,	estimate \eqref{LB4.intro} will be crucial in the spacelike singularity formation in $\mathcal{C}$. Obtaining \eqref{LB4.intro} was in truth the main motivation behind all the estimates that we wrote earlier, including the metric difference estimates.
		
		Finally, we prove the following estimate on $\frac{dr}{ds}:= \dot{r}$, which will be important \footnote{\color{black}Note that \eqref{dotr.intro} is consistent with the statement previously obtained in \cite{AnZhang,DejanAn} that both $r\partial_u r$ and $r\partial_v r$ admit a finite limit towards $\{r=0\}$ in the spherically symmetric Einstein-scalar-field model.} in the crushing region $\mathcal{C}$: \begin{equation} \label{dotr.intro} \tag{$\dot{r}$}
			-\dot{r}(s) \simeq \frac{\ep^2}{r(s)}.
		\end{equation} Note that, near the future boundary of $ \mathcal{LB}$, \eqref{dotr.intro} differs drastically from the Reissner--Nordstr\"{o}m-(dS/AdS) bounds $-\dot{r}_{RN}(s)  \approx e^{2K_- s}$. This is yet another sign of the nonlinear violent collapse, see Section~\ref{NLintro} below.

		\subsubsection{The crushing region $\mathcal{C}$: nonlinear regime and spacelike singularity formation} \label{NLintro}
		
		%Recall the estimates obtained in $\mathcal{LB}$ are far from sharp in the region $\ep^{-1} \ll s \ls \ep^{-\frac{5}{4}}$; for technical reasons, 
		We introduce the crushing region  $\mathcal{C}:= \{  \color{black}4\s\color{black}< s < s_{\infty} \}$, which overlaps substantially with $\mathcal{LB}$. \color{black}The reason of this overlap is that in the region $\mathcal{LB}\cap \mathcal{C}= \{ 4\s \leq s \leq \delta_{\mathcal{C}}\color{black}  \cdot \ep^{-1}\}$ the estimates we will get on $\mathcal{C}$ are weaker than the ones already proven in $\mathcal{LB}$; however, it is desirable to have the largest possible crushing region $\mathcal{C}$ in which the metric is uniformly close to Kasner. \color{black}
		
		In contrast to the bounds obtained in past regions, that were inspired from \cite{Moi} (in the context where the Cauchy horizon is stable), the region $\mathcal{C}$ has a new behavior. A nonlinear bifurcation occurs, driving the dynamical system \textit{away} from Reissner--Nordstr\"{o}m-(dS/AdS) and towards a Kasner metric (Section~\ref{Kasner.intro.section}). The key ingredient is an estimate on the null lapse $\Omega^2$ (see \eqref{Lapse}) which ultimately makes the spacetime volume of $\mathcal{C}$ much smaller than expected (see heuristics in Section~\ref{violent.section}): we prove that $\Omega^2$ tends to $0$ as a large $r$ power. Specifically,  for all $s \in \mathcal{C}$: \begin{align} \label{C1.intro} \tag{C1}
			\Omega^2(s) \simeq ( \frac{r(s)}{r_-})^{ b_-^{-2} \cdot \ep^{-2}}.
		\end{align}
		
		By \eqref{C1.intro}, we can schematically neglect the terms proportional to $\Omega^2$  in the EMKG system of ODEs (Section~\ref{ODEsection}): \begin{align*}
			&\frac{d^2\log(r\Omega^2)}{ds^2}=-2|\frac{d\phi}{ds}|^2+\frac{ \Omega^{2}}{r^{2}}(1-\frac{3e^2}{r^2}+m^2 r^2 |\phi|^2+  \Lambda r^2 ) \approx -2|\frac{d\phi}{ds}|^2 \numberthis \label{Om},\\    &\frac{d}{ds}(-r \dot{r}) = \Omega^2 \cdot (1-\frac{e^2}{r^2}) -r^2\Omega^2(\Lambda+ m^2 |\phi|^{2}) \approx 0, \\ &  \frac{d}{ds} ( r^2 \dot{\phi})= 	-m^{2}r^2 \Omega^{2}\phi \approx 0.
		\end{align*} The three above equations are of course mere heuristics we give to explain the proof: for details, see Section~\ref{crushing.section}. Once these heuristics have been made rigorous, they show that \eqref{dotr.intro} and \eqref{LB4.intro} are \textit{still valid} in $\mathcal{C}$. 
		To prove \eqref{C1.intro} (which is as of now  just a bootstrap assumption), we plug \eqref{dotr.intro} and \eqref{LB4.intro} into \eqref{Om} and integrate: the following schematic computation undoubtedly embodies the most important nonlinear aspects of the dynamics. \begin{align*}
			\frac{d^2\log(r\Omega^2)}{ds^2} \approx -2|\frac{d\phi}{ds}|^2 \approx - \frac{\ep^2}{r^4} \approx \frac{\dot{r}}{r^3}\ \Rightarrow\ \frac{d\log(r\Omega^2)}{ds} \approx -\frac{1}{r^2} \approx \frac{ \ep^{-2} \cdot \dot{r}}{r}\ \Rightarrow\ \log(r\Omega^2) \approx \log(r^{\ep^{-2}})+\color{black} \text{constant}\color{black}.
		\end{align*} Now that we have explained the main elements behind the estimates, we make important qualitative comments: \begin{enumerate}
			\item $\mathcal{C}$ is the only region where $r$ is not bounded away from zero. While the $r$-dependence of \eqref{dotr.intro} and \eqref{LB4.intro} was inconsequential in $\mathcal{LB}$, it becomes crucial in $\mathcal{C}$, as it dictates the behavior of all the other quantities as $r\rightarrow 0$.
			\item $\dot{\phi}$ and $\dot{r}$ blow-up with the same $r$ power as on Schwarzschild at $r=0$, but the pre-factor is different, which impacts dramatically the dynamics, and is responsible for \textit{violent nonlinear collapse} (see Section~\ref{violent.section}).
			
			\item The drastic difference between \eqref{dotr.intro} and \eqref{LB4.intro} and their Schwarzschild analogue can be interpreted as a consequence of the almost stability of the Reissner--Nordstr\"{o}m-(dS/AdS) Cauchy horizon and its associated blue-shift effect, which provides ``non-typical'' initial data on $\{s=\color{black} 4\s \color{black} \}$ the past boundary of $\mathcal{C}$.
			
			\item The dynamics are governed by the nonlinearity $-2|\frac{d\phi}{ds}|^2$ in the Gauss equation \eqref{Om}. This term is responsible for a monotonicity-driven collapse, which is reminiscent of the stable blow-up dynamics near Kasner in \cite{FournoRodSpeck}.
			
			\item To show \eqref{C1.intro}, we needed sharp estimates on $\phi$ already in the past, given by \eqref{LB4.intro} (showing that $\phi$ differs from $\phil$ in $\mathcal{LB}$), itself following from \eqref{EB4.intro}, an estimate that uses the presence of a linear resonance, i.e.\ \eqref{Res}.
		\end{enumerate}

		All the power law bounds e.g.\ \eqref{rho.intro} follow directly from \eqref{C1.intro}; they are more compelling than \eqref{C1.intro} as they do not depend on the choice of the coordinate $s$: the Hawking mass and the Kretschmann scalar  are geometric quantities.
		
		As a consequence of \eqref{dotr.intro} (still valid in $\mathcal{C}$), we show that there exists $s_{\infty}<+\infty$ with $s_{\infty}(\ep) \simeq \ep^{-2}$ such that  \begin{align*}
			\lim_{s \rightarrow s_{\infty}} r(s)=0,%\ \text{and}\ \lim_{s \rightarrow s_{\infty}} g(\nabla_g s, \nabla_g s) =  \lim_{s \rightarrow s_{\infty}} \Omega^{-2}(s) =+\infty,
		\end{align*} therefore $\mathcal{S}:=\{r=0\}$ is indeed a spacelike singularity, since $s$ is finite at $\mathcal{S}$ and  $g= -\Omega^2 (-dt^2+ ds^2 ) + d\sigma_{\mathbb{S}^2}$.

		Lastly, we recall that \eqref{C1.intro} translates into a Kasner behavior as explained in Section~\ref{Kasner.intro.section}, see also Theorem~\ref{Kasner.thm}.

		\subsection{Acknowledgments} I warmly thank Mihalis Dafermos for mentioning this problem to me,  for his interest in its resolution, and many helpful comments on the manuscript. I am grateful to Christoph Kehle 		and Yakov Shlapentokh--Rothman for interesting discussions  and for explaining certain aspects of their work, and to Hans Ringstr\"{o}m for kindly answering  my questions related to cosmology. I would like to warmly thank Jorge Santos for enlightening discussions about his work, and AdS-CFT in general. I am grateful to Jonathan Luk and Harvey Reall for useful comments on the manuscript. Finally, I would like to thank an anonymous reviewer for a very thorough review and very useful suggestions that greatly improved the presentation.
		
		\subsection{Outline of the paper}\begin{itemize}
			\item In Section~\ref{prelim}, we introduce the Reissner-Nordstr\"{o}m-(dS/AdS) metric, the class of spatially-homogeneous spacetimes that we consider,  specify the data of Theorem~\ref{thm.intro}, and the reduction of the Einstein equations as a system of nonlinear ODEs. We also mention a result on the linearized dynamics from \cite{ChristophYakov} that we will be needing.
			
			\item 		In Section~\ref{precise.section}, we provide a precise statement of the results that were  condensed in Theorem~\ref{thm.intro}: Theorem~\ref{maintheorem} (main theorem), Theorem~\ref{convergence.theorem} (convergence to Reissner-Nordstr\"{o}m), Theorem~\ref{Kasner.thm} (Kasner behavior).
			\item 		In Section~\ref{main.section}, we prove Theorem~\ref{maintheorem} and provide many useful quantitative estimates.
			\item 		Section~\ref{convergence.section} is a short proof of Theorem~\ref{convergence.theorem}, entirely based on the   quantitative estimates of Section~\ref{main.section}.
			
			\item 		In Section~\ref{Kasner.section}, we prove Theorem~\ref{Kasner.thm}, mostly based again on the   quantitative estimates of Section~\ref{main.section} and also using certain renormalization procedures and convergence estimates.
		\end{itemize}

		\section{Preliminaries} \label{prelim}
		\subsection{The Reissner-Nordstr\"{o}m-(dS/AdS) interior metric} \label{RN.section}
		
		Before defining the Reissner-Nordstr\"{o}m-(dS/AdS) metric, we introduce the set of sub-extremal parameters $(M,e,\Lambda) \in \RR_+^{*} \times \RR^2$ defined as $\PP:= \PPp \cup \PPm$ where $\PPm$ consists of all  $(M,e,\Lambda) \in  \RR_+^{*} \times \RR \times \RR_-$ such that the polynomial $ X^2 -2M X+ e^2 -\frac{\Lambda X^4}{3}$ has two positive simple roots $r_-<r_+$ and $\PPp$ of of all  $(M,e,\Lambda) \in  \RR_+^{*} \times \RR \times \RR_{+}^{*}$ such that the polynomial $ X^2 -2M X+ e^2 -\frac{\Lambda X^4}{3}$ has three positive simple roots $r_-<r_+< r_c$.

		Let $(g_{RN},F_{RN})$ be an electro-vacuum solution of \eqref{E1}-\eqref{E5} i.e.\ a solution with $\phi \equiv 0$, given by 
		\begin{equation} \tag{$g_{RN}$} \label{RN}
			g_{RN}= -(1-\frac{2M}{r_{RN}}+\frac{e^2}{r^2_{RN}}-\frac{\Lambda r^2_{RN}}{3}) dt^2+ (1-\frac{2M}{r_{RN}}+\frac{e^2}{r^2_{RN}}-\frac{\Lambda r^2_{RN}}{3})^{-1} dr^2_{RN} +r^2_{RN}[ d\theta^{2}+\sin(\theta)^{2}d \varphi^{2}],
		\end{equation}
		\begin{equation}\label{F.RN} \tag{$F_{RN}$}F_{RN} =\frac{e}{r^2}    dt \wedge dr,
		\end{equation} in the coordinate range $t\in \RR$, $r_{RN} \in [r_-(M,e,\Lambda),r_+(M,e,\Lambda)]$ and where $(M,e,\Lambda) \in \PP$ and $r_{\pm}(M,e,\Lambda)>0$ are the roots of the polynomial $ X^2 -2M X+ e^2 -\frac{\Lambda X^4}{3}$.   
		For $\Lambda=0$ (resp.\ $\Lambda>0$, resp.\ $\Lambda<0$), the metric given by \eqref{RN} is called the sub-extremal Reissner-Nordstr\"{o}m (resp.\ de Sitter, resp.\ Anti-de Sitter) interior metric. Define: \begin{equation} 
			\frac{dr^*}{dr_{RN}}:=  -  \frac{1}{\Omega^2_{RN}(r_{RN})}, \hskip 2 mm \Omega^2_{RN}(r_{RN}):= -(  1-\frac{2M}{r_{RN}}+\frac{e^2}{r^2_{RN}}-\frac{\Lambda r_{RN}^2}{3}),
		\end{equation}
		\begin{equation}
			u:= \frac{r^*-t}{2}, \hskip 2 mm v:= \frac{r^*+t}{2},
		\end{equation} in the range $u \in \bar{\RR}$, $v \in \bar{\RR}$, $r^* \in \bar{\RR}$. We attach the following so-called event horizons to the space-time: \begin{equation}
			\mathcal{H}_{L}^+:=\{v=-\infty\}, \hskip 2 mm \mathcal{H}_{R}^+:=\{u=-\infty\}, \hskip 2 mm \mathcal{H}^+:= \mathcal{H}_{L}^+ \cup \mathcal{H}_{R}^+=\{r^*=-\infty\}.
		\end{equation} Note that in the $(u,v,\theta,\psi)$ coordinate system, \eqref{RN} takes the form $$ g_{RN} = -4\Omega^2_{RN} du dv  +r^2_{RN}[ d\theta^{2}+\sin(\theta)^{2}d \varphi^{2}].$$ It is a classical computation (see e.g.\ \cite{Moi}) that there exists a constant $\alpha_+(M,e,\Lambda)>0$ such that as $r^{*} \rightarrow -\infty$ \begin{equation} \label{Omega_asymp}
			\Omega^2_{RN}(r^*) \sim \alpha_+ \cdot e^{2K_+(M,e,\Lambda) \cdot r^*}= \alpha_+ \cdot e^{2K_+(M,e,\Lambda) \cdot (u+v)},
		\end{equation} where $2K_{+}(M,e,\Lambda)>0$ is the so-called surface of the gravity of the event horizon and we also define more generally \begin{equation}\tag{$K_{\pm}$}\label{K}
			2 K_{\pm}(M,e,\Lambda)= \frac{2}{r_{\pm}^2(M,e,\Lambda)} ( M - \frac{e^2}{r_{\pm}(M,e,\Lambda)}-\frac{\Lambda r_{\pm}^3(M,e,\Lambda)}{3}),
		\end{equation} noting that the so-called surface gravity of the Cauchy horizon $2K_-(M,e,\Lambda)<0$ is strictly negative. 
		
		Lastly, we also define ``regular coordinates'' $(U,V)$ as follows \begin{equation} \label{defU}
			\frac{dU}{du}= e^{2K_+ u}, \hskip 2 mm U(u=-\infty)=0.
		\end{equation}
		\begin{equation} \label{defV}
			\frac{dV}{dv}= e^{2K_+ v}, \hskip 2 mm V(v=-\infty)=0.
		\end{equation}
		\subsection{\color{black}Spherically \color{black} symmetric metrics and \eqref{E1}--\eqref{E5} in double null coordinates} \label{doublenull}
		
		We consider a \color{black}spherically \color{black} symmetric Lorentzian manifold $(M,g)$ taking the following form in  null coordinates $(u,v)$: \begin{equation} \label{guv}
			g= g_{\mathcal{Q}}+ r^2 \color{black}d\sigma_{\mathbb{S}^2}\color{black}=- 4\Omega^2 du dv + r^2 \color{black}d\sigma_{\mathbb{S}^2}\color{black},
		\end{equation} %\begin{equation}\label{surface.sym.eq} \begin{split}d\Sigma_{1} = d\theta^2 + \sin^2(\theta) d\varphi^2, \hskip 5 mm d\Sigma_{-1} = d\theta^2 + \sinh^2(\theta) d\varphi^2, \hskip 5 mm d\Sigma_{0} = dx^2+ dy^2.\end{split}\end{equation} where $d\Sigma_{k}$ is the metric of constant curvature $k \in \{-1,0,1\}$ (respectively on $\mathbb{H}^2$,  $\mathbb{T}^2$ or  $\mathbb{S}^2$) 
		and we call $r$ the area-radius. We  also define the Hawking mass by the formula: $$ \rho :=  \frac{r}{2}(1- g_{\mathcal{Q}} (\nabla r, \nabla r )).$$ In the presence of the Maxwell field \eqref{F.RN}, we define the Vaidya mass $\varpi$ and the constant-$r$ surface gravity $2K$:

		\begin{equation} \label{Kdef}
			\varpi := \rho + \frac{e^2}{2r}- \frac{\Lambda r^3}{6}, \hskip 2 mm 2K = \frac{2}{r^2}(\varpi- \frac{e^2}{r}-\frac{\Lambda r^3}{3}).
		\end{equation}

		An elementary computation shows that (note that the normalization of $\Omega^2$ differs from \cite{Moi} by a factor $4$) \begin{equation} \label{mu} \frac{- \partial_u r \partial_v r}{\Omega^2} = 1- \frac{2 \varpi}{r}+ \frac{e^2}{r^2}-\frac{\Lambda r^2}{3}.\end{equation}

		Let $\phi$ a real-valued spherically symmetric scalar field on $(M,g)$ such that $(g,F_{RN},\phi)$ is a solution of \eqref{E1}--\eqref{E5}. Then, $(r,\Omega^2,\phi)$ satisfy the following system of PDE's in any double-null $(u,v)$ coordinates:
		
		\begin{equation}\label{Radius}\partial_{u}\partial_{v}r =\frac{- \Omega^{2}}{r}-\frac{\partial_{u}r\partial_{v}r}{r}
			+\frac{ \Omega^{2}}{r^{3}} e^2 +r \Omega^2  (\Lambda+ m^2|\phi|^{2})  =-\Omega^2 \cdot 2K +r \Omega^2\cdot    m^2|\phi|^{2}, \end{equation}
		\begin{equation}\label{Omega}
			\partial_{u}\partial_{v} \log(\Omega^2)=-2\partial_u \phi \partial_v \phi+\frac{2 \Omega^{2}}{r^{2}}+\frac{2\partial_{u}r\partial_{v}r}{r^{2}}- \frac{4 \Omega^{2}}{r^{4}} e^2,
		\end{equation}
		\begin{equation}\label{RaychU}\partial_{u}(\frac {\partial_{u}r}{\Omega^{2}})=\frac {-r}{\Omega^{2}}|  \partial_{u} \color{black}\phi|^{2}, \end{equation}
		\begin{equation} \label{RaychV}\partial_{v}(\frac {\partial_{v}r}{\Omega^{2}})=\frac {-r}{\Omega^{2}}|\partial_{v}\color{black}\phi|^{2},\end{equation}
		\begin{equation}\label{Field}
			\partial_u \partial_v  \phi =-\frac{\partial_{v}r \partial_{u}\phi}{r}-\frac{\partial_{u}r\partial_{v}\phi}{r} 
			- m^{2}\Omega^{2} \phi.\end{equation}

		\subsection{Set-up of the initial data}\label{data.section}
		
		We pose characteristic data for \eqref{Radius}--\eqref{Field} on two affine-complete null hypersurfaces $\mathcal{H}^+_{L}:=\{v=-\infty, u \in \RR\}$ and $\mathcal{H}^+_{R}:=\{u=-\infty, v\in \RR\}$ intersecting at a sphere, where $(u,v)$ are renormalized by the following gauge conditions:
		\begin{equation} \label{gaugeuv}
			\frac{-\partial_u r}{\Omega^2}_{|\mathcal{H}^+_{R}}(v)=1-\frac{m^2 r_+\ \ep^2}{2K_+(M,e,\Lambda)}, \hskip 2 mm \frac{-\partial_v r}{\Omega^2}_{|\mathcal{H}^+_{L}}(u) =1-\frac{m^2 r_+\ \ep^2}{2K_+(M,e,\Lambda)}.
		\end{equation}
		The data we consider consist of sub-extremal Reissner-Nordstr\"{o}m-(dS/AdS) data for $(r,\Omega^2)$ and a constant function $\phi$ on $\mathcal{H}^+_{L} \cup\mathcal{H}^+_{R}$, more precisely for some $(e,M,\Lambda) \in \PP$ and some $\epsilon \neq 0$ (recall the definition of $\alpha_+>0$ from \eqref{Omega_asymp}): \begin{equation} \tag{$r$-data} \label{data1}
			r_{|\mathcal{H}^+_{L} \cup\mathcal{H}^+_{R}} \equiv r_+(M,e,\Lambda),
		\end{equation}
		\begin{equation} \tag{$\Omega^2$-data} \label{data2}
			\left( \frac{\Omega^2}{e^{ 2K_+ (u+v)}} \right)_{|\mathcal{H}^+_{L} \cup\mathcal{H}^+_{R}} \equiv \frac{\alpha_+(M,e,\Lambda)}{(1-\frac{m^2 r_+\ \ep^2}{2K_+(M,e,\Lambda)})^2},
		\end{equation}
		\begin{equation} \tag{$\phi$-data}  \label{data3}
			\phi_{|\mathcal{H}^+_{L} \cup\mathcal{H}^+_{R}} \equiv \epsilon,
		\end{equation}  where $2K_+(M,e,\Lambda)>0$ is given by \eqref{K}. Note that the main advantage of the term $\frac{m^2 r_+\ \ep^2}{2K_+(M,e)}$ in \eqref{gaugeuv} is that in this gauge, $\Omega^2$ is proportional to $e^{2K_+(u+v)}$ as in the $\ep=0$ case.   \color{black} One can check that such data are compatible with the (null) constraints imposed on $\mathcal{H}^+_{L} \cup\mathcal{H}^+_{R}$ by \eqref{Radius}--\eqref{Field}. Recalling that $(u,v)$ are defined by \eqref{gaugeuv}, we define coordinates $U$ and $V$ by the formulae \eqref{defU}, \eqref{defV}, %where $\ep=0$):\begin{equation} \label{defU2} \frac{dU}{du}= e^{(2K_+-m^2r_+\ep^2) u}, \hskip 2 mm U(u=-\infty)=0. \end{equation} \begin{equation} \label{defV2} \frac{dV}{dv}= e^{(2K_+-m^2r_+\ep^2) v}, \hskip 2 mm V(v=-\infty)=0. \end{equation} \color{black} 
		thus $\mathcal{H}^+_{L}:=\{V=0, u \in \RR\}$ and $\mathcal{H}^+_{R}:=\{U=0, v\in \RR\}$. Then by \eqref{Field} we have the following identities \begin{equation*}
			(\partial_v \partial_U \phi)_{|\mathcal{H}^+_{R}}=-\frac{m^2}{2K_+(M,e,\Lambda)-m^2 r_+\ \ep^2} \cdot \frac{\alpha_+(M,e,\Lambda)}{{1-\frac{m^2 r_+\ \ep^2}{2K_+(M,e,\Lambda)}}} \cdot e^{2K_+(M,e,\Lambda) v} \epsilon,
		\end{equation*}
		\begin{equation*}
			(\partial_u \partial_V \phi)_{|\mathcal{H}^+_{L}}=-\frac{m^2}{2K_+(M,e,\Lambda)-m^2 r_+\ \ep^2} \cdot \frac{\alpha_+(M,e,\Lambda)}{{1-\frac{m^2 r_+\ \ep^2}{2K_+(M,e,\Lambda)}}}\cdot  e^{2K_+(M,e,\Lambda)  u} \epsilon.
		\end{equation*} Integrating using the boundary condition  $\partial_U \phi(0,0)= \partial_V \phi(0,0)=0$ at the bifurcation sphere $ \{U=V=0\}$ gives \begin{equation} \label{dUdata}
			( \partial_U \phi)_{|\mathcal{H}^+_{R}}(v)=-m^2 \cdot \frac{ \alpha_+(M,e,\Lambda)}{(2K_+(M,e,\Lambda)-m^2 r_+ \epsilon^2)^2 } \cdot e^{2K_+(M,e,\Lambda)  v}\cdot \epsilon,
		\end{equation}
		\begin{equation} \label{dVdata}
			( \partial_V \phi)_{|\mathcal{H}^+_{L}}(u)=-m^2 \cdot \frac{ \alpha_+(M,e,\Lambda)}{(2K_+(M,e,\Lambda)-m^2 r_+ \epsilon^2 )^2  } \cdot e^{2K_+(M,e,\Lambda)  u}\cdot \epsilon.
		\end{equation} In particular, note that $( \partial_u \phi)_{|\mathcal{H}^+_{R}} \equiv 0$ and  $( \partial_v \phi)_{|\mathcal{H}^+_{L}} \equiv 0$. 
		
		\subsection{System of ODEs for static solutions and initial data} \label{ODEsection}
		
		Let $t:= v-u$ and $s:= v+u$ where $(u,v)$ are given by \eqref{gaugeuv}. For the data of Section~\ref{data.section}, it is clear that $\partial_t r$, $\partial_t \log(\Omega^2)$, $\partial_t \phi =0$ on the initial surfaces $\mathcal{H}^+_{L} \cup \mathcal{H}^+_{R}$ (recall the discussion in Section~\ref{data.section}). Therefore,  $\partial_t r$, $\partial_t \log(\Omega^2)$, $\partial_t \phi =0$ for solutions $(r,\Omega^2,\phi)$ of the 2D PDE \eqref{Radius}--\eqref{Field} with the data of Section~\ref{data.section}. In view of the identities  $$ \partial_u = \partial_s - \partial_t, \hskip 2 mm  \partial_v = \partial_s + \partial_t, $$ it is clear for that all $f \in \{r(s),\Omega^2(s),\phi(s)\}$, $ \partial_u f = \partial_v f= \frac{df}{ds}$.  In the $(t,s,\theta,\varphi)$ coordinate system, \eqref{guv} becomes \begin{equation} \label{gts}g= \Omega^2(s) (-ds^2+dt^2) + r^2(s) d\sigma_{\mathbb{S}^2}.
		\end{equation}  Note that the usual gauge freedom of spherical symmetry consisting in  $u \rightarrow F(u)$ and $v \rightarrow G(v)$ for monotonic functions $F$ and $G$ translates  into the gauge transform $s\rightarrow H(s)$ for $H(s)$ monotonic. \color{black} Using the notation $\frac{df}{ds}= \dot{f}$, the  Einstein--(Maxwell)-- Klein-Gordon equations \eqref{Radius}--\eqref{Field} can be written as
		\begin{equation}\label{Radiusstat} \ddot{r} =\frac{- \Omega^{2}}{r}-\frac{\dot{r}^2}{r}
			+\frac{ \Omega^{2}}{r^{3}} e^2 +  r\Omega^2 (\Lambda+ m^2 |\phi|^{2}) , \end{equation}  which is also equivalent to
		\begin{equation} \label{req}
			\frac{d}{ds}(-r \dot{r}) = \Omega^2 \cdot (1-\frac{e^2}{r^2}) -r^2\Omega^2(\Lambda+ m^2 |\phi|^{2}),
		\end{equation} or  \begin{equation} \label{req2}
			\frac{d}{ds}(\frac{-\dot{r}}{r}) =  2 (\frac{-\dot{r}}{r})^2+ \Omega^2 \cdot (\frac{1}{r^2}-\frac{e^2}{r^4})-\Omega^2 (\Lambda+ m^2 |\phi|^{2})
		\end{equation}
		
		\begin{equation}\label{Omegastat}
			\frac{d^2\log(\Omega^2)}{ds^2}=-2|\dot{\phi}|^2+\frac{2 \Omega^{2}}{r^{2}}+\frac{2 \dot{r}^2}{r^{2}}- \frac{ 4\Omega^{2}}{r^{4}} e^2,
		\end{equation} or equivalently 	\begin{equation}\label{Omegastat2}
			\frac{d^2\log(r\Omega^2)}{ds^2}=-2|\dot{\phi}|^2+\frac{ \Omega^{2}}{r^{2}}(1-\frac{3e^2}{r^2}+m^2 r^2 |\phi|^2+  \Lambda r^2 ),
		\end{equation} 	the Raychaudhuri equation: 
		
		\begin{equation}\label{Raychstat}\frac{d}{ds}(\frac {\dot{r}}{\Omega^{2}})=\frac {-r}{\Omega^{2}}| \dot{\phi}|^{2}, \end{equation}
		
		which we can also write,  defining the quantity $\kappa^{-1}:= \frac{- \dot{r}}{\Omega^2} $ as \begin{equation}
			\dot{ (\kappa^{-1})}= \frac {r}{\Omega^{2}}| \dot{\phi}|^{2}.
		\end{equation} The Klein--Gordon wave equation : 
		
		\begin{equation}\label{FieldODE}
			\ddot{\phi} =-\frac{2\dot{\phi}\cdot \dot{r}}{r}
			- m^{2}\Omega^{2}\phi,\end{equation}
		or equivalently   \begin{equation}\label{FieldODE2}
			\frac{d}{ds} ( r^2 \dot{\phi})= 	-m^{2}r^2 \Omega^{2}\phi,\end{equation} \begin{equation}\label{FieldODE3}
			\frac{d}{ds} ( \frac{r^2 \dot{\phi}}{\Omega^2})= 	-m^{2}r^2 \phi - \frac{r^2 \dot{\phi}}{\Omega^2} \frac{d}{ds}(\log(\Omega^2)).\end{equation}
		
		Note that $\dot{r} \leq 0$ therefore $\kappa^{-1}\geq 0$. In the $s$-variables, \eqref{mu} becomes \begin{equation} \label{mu.static}
			\kappa^{-2} \cdot \Omega^2 = -[\color{black}1-\frac{2\varpi}{r}+\frac{e^2}{r^2}-\frac{\Lambda r^2}{3}]\color{black}=-[1-\frac{2\rho}{r}].
		\end{equation}
		
		Recalling  the renormalized Vaidya mass from \eqref{Kdef} we can also derive the following equation: \begin{equation} \label{mass.static}
			\dot{\varpi}= \frac{r^2}{2} \kappa^{-1} |\dot{\phi}|^2 + \frac{m^2}{2} r^2 \dot{r} |\phi|^2.
		\end{equation}

		Note that gauge \eqref{gaugeuv} translates into the following normalization for $s$ on $\mathcal{H}^+=\{s=-\infty\}$: \begin{equation} \label{gauge}\kappa^{-1}(-\infty)=\frac {-\dot{r}(-\infty)}{\Omega^{2}(-\infty)} = 1-\frac{m^2 r_+\ \ep^2}{2K_+(M,e,\Lambda)} .
		\end{equation}

		Additionally, note that one can easily show that $\varpi$ is constant on $\mathcal{H}^+$ and equates the following value:  \begin{equation} \label{init3}
			\varpi_{|\mathcal{H}^+}=M.
		\end{equation}
		Lastly, note that \eqref{dUdata}, \eqref{dVdata} become \begin{equation} \label{dotphi.data}
			(\frac{\dot{\phi}}{\Omega^2})_{|\mathcal{H}^+}= \phi'_{|\mathcal{H}^+}=\frac{-m^2}{4 K_+^2(M,e,\Lambda)}\ \epsilon,
		\end{equation} where we introduced the notation $f':= \frac{df}{d\sr}$ for a regular coordinate $\sr$ defined as \begin{equation}
			\frac{d\sr}{ds} = \Omega^2(s), \hskip 2 mm \sr(s=-\infty)=0.
		\end{equation}

		\subsection{Linear scattering in the black hole interior}

		Consider the linear Klein-Gordon equation with mass $m^2 \in \RR$ on the Reissner--Nordstr\"{o}m-(dS/AdS) interior  \eqref{RN} with sub-extremal parameters $(M,e,\Lambda) \in \mathcal{P}_{se}$ (see Section~\ref{RN.section}):   \begin{equation} \label{linearKG}
			\Box_{g_{RN}} \phil=m^2 \phil.
		\end{equation}    Using the spherical harmonics $Y_{l,m}(\theta,\varphi)$, we write $\phil(t,s,\theta,\varphi)= \sum_{l=0}^{+\infty} \sum_{m=-l}^{l} (\phil)_{l,m}(t,s) Y_{l,m}(\theta,\varphi)$, where $s=r^*$ is defined in Section~\ref{RN.section}. \color{black}
		We consider the case of homogeneous, spherically symmetric solutions $T \phil=0$, $R \phil=0$ where $T=\partial_t$ is the Killing vector field of  Reissner--Nordstr\"{o}m-(dS/AdS) (spacelike in the interior) and $R$ is any vector field on the sphere. Since $R \phil=0$, then $(\phil)_{l,m}\equiv 0$ for any $l\neq 0$. Then we denote $u(\omega,s)=  \int_{\omega \in \RR} e^{i\omega t}[r \phi_{\mathcal{L}}](t,s) dt \color{black}$: since $T\phil=0$, then $u(\omega,s)=0$ for any $\omega \neq 0$. In what follows, we denote $u(s)=u(\omega=0,s)$. \color{black} Then,
		$u$ satisfies the following ODE in $s$ \begin{equation} \label{ODEw=0}
			\ddot{u}=V \cdot u,
		\end{equation}  where the expression of $V(s)$  can be found in (6.2) and (6.3) with $\omega=l=0$ in  [\cite{ChristophYakov}, Section 6]. Define $\tuone$ and $\tutwo$, two linearly independent solutions of \eqref{ODEw=0}, characterized by their asymptotic behavior towards the event horizon, i.e.\ as $s\rightarrow -\infty$, the following asymptotic equivalences are true: $$ \tuone(s) \sim 1,$$ $$ \tutwo(s) \sim s.$$
		Also define $\tvone$ and $\tvtwo$, two linearly independent solutions of \eqref{ODEw=0}, characterized by their asymptotic behavior towards the Cauchy horizon i.e.\ as $s\rightarrow +\infty$, the following asymptotic equivalences are true: $$ \tvone(s) \sim 1,$$ $$ \tvtwo(s) \sim s.$$
		
		\begin{lem}[Kehle--Shlapentokh--Rothman \cite{ChristophYakov}]
			Recalling the definition  of $2K_-(M,e)<0$ from Section~\ref{RN.section}, there exists $C(M,e,\Lambda,m^2)>0$ and $s_0(M,e,\Lambda,m^2) \in \RR$ such that for all $s \geq s_0$: \begin{equation} \label{tvoneasymp}
				|\tvone(s)-1|+ |\frac{d\tvone}{ds}|(s)+ |\frac{d^2\tvone}{ds^2}|(s)\leq C \cdot e^{2K_- s},
			\end{equation} \begin{equation} \label{tvtwoasymp}
				|\tvtwo(s)-s|+ |\frac{d\tvtwo}{ds}(s)-1|+ |\frac{d^2\tvtwo}{ds^2}|(s)\leq C \cdot e^{2K_- s}.
			\end{equation}
			
		\end{lem}
		
		\begin{prop}[Kehle--Shlapentokh--Rothman \cite{ChristophYakov}] \label{KSR.thm}
			There exists $A(M,e,\Lambda,m^2)\in \RR$, $B(M,e,\Lambda,m^2)\in \RR$ such that $$ \tuone= A \tvone+B \tvtwo.$$
			
			Moreover, there exists $\mathcal{Z}(M,e,\Lambda,m^2)$, the zero set of an analytic function (in particular, a subset of $\RR^4$ of zero Lebesgue measure) such that for all $ (M,e,\Lambda,m^2) \in \mathcal{P}_{se}- \mathcal{Z}$, $B(M,e,\Lambda,m^2) \neq 0$.
			
		\end{prop}

		\begin{cor} \label{linearcor}
			Let $\phil$ be a solution of \eqref{linearKG} with constant data $\phil= \epsilon \neq 0$ on $\{s=-\infty\}$  on the  Reissner--Nordstr\"{o}m-(dS/AdS) metric \eqref{RN}. Then for all parameters $ (M,e,\Lambda,m^2) \in \mathcal{P}_{se}- \mathcal{Z}$, there exists constants $C(M,e,\Lambda,m^2)>0$ and $s_0(M,e,\Lambda,m^2) \in \RR$ such that for all $s \geq s_0$:
			\begin{equation} \label{asymplin}
				|\phil(s)-A(M,e,\Lambda,m^2)\cdot \epsilon - B(M,e,\Lambda,m^2) \cdot \epsilon \cdot s|+ |\dphil(s)-B(M,e,\Lambda,m^2) \cdot \epsilon|\leq C \cdot \epsilon \cdot e^{2K_-(M,e,\Lambda) s},
			\end{equation}
			\begin{equation} \label{RNasymp1}
				\Omega^2_{RN}(s) \leq C \cdot e^{2K_-(M,e,\Lambda) s},
			\end{equation}\begin{equation} \label{RNasymp2}
				| \frac{d \log	\Omega^2_{RN}}{ds}(s) -2K_-(M,e,\Lambda)|\leq C \cdot e^{2K_-(M,e,\Lambda) s}.
			\end{equation}
			
		\end{cor}

		\section{Precise statement of the main results} \label{precise.section}
		
		In this section and all the subsequent ones, we will use the notation $f(s) \ls g(s)$ to signify that there exists a constant $\tilde{C}(M,e,\Lambda,m^2)>0$ such that $f(s) \leq \tilde{C} \cdot g(s)$ for all $s$ in the region of interest. $\gtrsim$ is defined similarly and $f(s) \simeq g(s)$ if $f(s) \ls g(s)$ and $g(s) \ls f(s)$ (note that we have used this notation already in Section~\ref{intro.section}).
		
		\subsection{Presence of a spacelike, crushing singularity and gauge-invariant estimates}
		
		\begin{thm}\label{maintheorem}
			Let $\mathcal{Z} \subset \RR^4$ be the zero Lebesgue measure set from Theorem~\ref{KSR.thm}. Then, for all $(M,e,\Lambda,m^2) \in \mathcal{P}_{se}- \mathcal{Z}$ with $e\neq0$, there exists $\epsilon_0(M,e,\Lambda,m^2)>0$ such that for all $0<|\ep| < \ep_0$, the future domain of dependence $\mathcal{M}$ of the characteristic data from Section~\ref{data.section} (i.e.\ $\phi(s=-\infty)=\ep$) terminates at a spacelike singularity at which $r=0$. 
			
			More precisely, there exists a foliation of $\mathcal{M}$ by spacelike hypersurfaces $\Sigma_{s}$ with, $s\in (-\infty,s_{\infty}(\ep))$, where $s$ is defined in \eqref{gauge}, and $s_{\infty}(\ep):=  \frac{|K_-|}{B^2} \cdot \ep^{-2} + L(\ep) \cdot \color{black}\log(\ep^{-1})\color{black}$ and $L_0(M,e,\Lambda,m^2)>0$ are such that \begin{equation}
				\lim_{s \rightarrow s_{\infty}(\ep)} r(s)=0, \hskip 5 mm |L(\ep)| \leq L_0,
			\end{equation}
			where we recall that $B(M,e,\Lambda,m^2) \neq 0$ is defined in Theorem~\ref{KSR.thm} and $2K_-(M,e,\Lambda)<0$ is defined in Section~\ref{RN.section}. $\mathcal{S}:=\{s=s_{\infty}\}$ is furthermore a spacelike singularity in the sense that for all $ p \in \{s=s_{\infty}\}$,  $J^{-}(p) \cap \mathcal{H}^+$ is compact, where $J^{-}(p)$ is the causal past  of $p$.
			
			Moreover, all the quantitative estimates stated in Proposition \ref{prop.RS.d}, \ref{prop.N.d}, \ref{EB.Prop}, \ref{LB.prop}, \ref{C.prop}, \ref{spacelike.prop} are satisfied. In particular, the following stability with respect to Reissner--Nordstr\"{o}m-(dS/AdS) estimates hold: for all $s \ls \ep^{-1}:$ \begin{equation} \label{RN.stab.main}
				\ep^{-1}\cdot \bigl| \phi(s)-\phil(s) \bigr| \ls \ep^2 \cdot s^2\ls 1, \hskip 5 mm  \ep^{-1}\cdot \bigl| \frac{d\phi}{ds}(s)-\frac{d\phil}{ds}(s) \bigr| \ls \ep^2  \cdot s \ls \ep,
			\end{equation}\begin{equation} \label{RN.stab.main2}
				\bigl| r(s)-r_{RN}(s)\bigr| \ls \ep^2 \cdot s \ls \ep, \hskip 5 mm \bigl| \log(\frac{\Omega^2(s)}{\Omega^2_{RN}(s)}) \bigr| \ls \ep^2  \cdot \color{black}s^2 \lesssim 1\color{black}.
			\end{equation}
			
			Additionally, in the crushing region $\mathcal{C}:=\{  \color{black} 4\s \color{black} \leq s < s_{\infty}\}$, \color{black} where $\s= (|K_-|)^{-1} \log(\nu \cdot \ep^{-1})$ for  $\nu(M,e,\Lambda,m^2)>0$, \color{black} there exists $ \check{D}_C(M,e,\Lambda,m^2)>0$ such that we have the following spacetime volume estimate: for all $s\in \mathcal{C}$  \begin{equation} \label{volume.estimate}
				\color{black} (\frac{r(s)}{r_-})^{ b_-^{-2} \cdot \ep^{-2} \cdot(1+\ep^2  \cdot\log(\ep^{-1}) \cdot\check{D}_C)}\ls vol\left( \bigcup_{s' \geq s} \Sigma_{s'}  \cap \{t\in [0,1]\}\right)  \ls   (\frac{r(s)}{r_-})^{ b_-^{-2} \cdot \ep^{-2} \cdot(1-\ep^2  \cdot\log(\ep^{-1}) \cdot  \check{D}_C)}.\color{black}
			\end{equation}

			Finally, there exists $ \tilde{D}_C(M,e,\Lambda,m^2)>0$ such that the following blow-up  estimates hold for all $ s \in \mathcal{C}$: \begin{equation} \label{blow.up.main.theorem.1}
				|\phi|(s) \simeq  \epsilon^{-1} \cdot \log( r^{-1}(s)), \hskip 5 mm  |\frac{d\phi}{dr}|(s) \simeq \frac{\epsilon^{-1}}{r(s)},\end{equation}
			\begin{equation} \label{blow.up.main.theorem.2}
				\color{black}\ep^4\color{black} \cdot  (\frac{r(s)}{r_-})^{-b_-^{-2}\cdot \ep^{-2} \cdot (1-\color{black}  \ep^2  \cdot\log(\ep^{-1})\color{black} \cdot  \tilde{D}_C) } \ls \rho(s)\ls \color{black}\ep^4\color{black}  \cdot  (\frac{r(s)}{r_-})^{-b_-^{-2}\cdot \ep^{-2} \cdot (1+\color{black}  \ep^2  \cdot\log(\ep^{-1})\color{black} \cdot  \tilde{D}_C) },
			\end{equation}\begin{equation} \label{blow.up.main.theorem.3}
				\frac{r^6(s) \cdot\mathfrak{K}(s)}{\rho^2(s)} \simeq \ep^{-4},\end{equation} where $\rho(s)$ is the Hawking mass and $\mathfrak{K}(s):=  R_{\alpha \beta \mu \nu} R^{\alpha \beta \mu \nu}(s)$ is the Kretchsmann scalar. Moreover for all $p>1+\ep$:   \begin{align*} 
				&\sup_{ s\in[\ep^{-1},s_{\infty}]} |\phi|(s)= +\infty, \hskip 10 mm \sup_{ s\in[\ep^{-1},s_{\infty}]} |\frac{d\phi}{dr}|(s)= +\infty,   \numberthis \label{phiblowup}\\& 
				\sup_{s\in[\ep^{-1},s_{\infty}]} \rho(s)= +\infty, \hskip 15 mm  \int_{s\in[\ep^{-1},s_{\infty}] \cap  \{t\in [0,1]\}} \rho^{p}   dvol_g:= 4\pi  \int_{\ep^{-1}}^{s_{\infty}} \rho^{p}(s')  \cdot  r^2(s') \cdot \Omega^2(s') ds' = +\infty,
				\\& 
				\sup_{s\in[\ep^{-1},s_{\infty}]}  r^6(s) \cdot\mathfrak{K}(s)= +\infty, \hskip 4 mm \lim_{\ep \rightarrow 0} \sup_{s\in[\ep^{-1},s_{\infty}]} \frac{ r^6(s) \cdot\mathfrak{K}(s)}{\rho^2(s)}= +\infty,  \hskip 4 mm   \int_{s\in[\ep^{-1},s_{\infty}] \cap  \{t\in [0,1]\}}    (r^3 \cdot \sqrt{\mathfrak{K}})^pdvol_g = +\infty.
			\end{align*}

		\end{thm}

		\subsection{Convergence to Reissner--Nordstr\"{o}m-(dS/AdS) in a weak topology}
		
		\begin{defn}[Convergence in distribution] \label{distrib.conv}
			We define the functions the real-line $s \in \RR \rightarrow r^2_{RN}(s)$ and $ s \in \RR \rightarrow \Omega^2_{RN}(s):=-(1-\frac{2M}{r_{RN}(s)}+ \frac{e^2}{r_{RN}^2(s)})$, where $r_{RN}$  is the area-radius function of the spherically symmetric Reissner--Nordstr\"{o}m-(dS/AdS) metric \eqref{RN}, and the coordinate $s=r^*$ on \eqref{RN} (see Section~\ref{RN.section}). Similarly, defining a coordinate $s$ from \eqref{gauge}, we view $r^2_{\ep}(s)$ and $  \Omega^2_{\ep}(s)$ from \eqref{gts} (where $g$ is the metric from Theorem~\ref{maintheorem}) as functions on $\{ s \in  (-\infty,s_{\infty}(\ep))\}$, and we extend these functions on the real-line as $\tilde{r}^2_{\ep}(s)=r^2_{RN}(s)$, $\tilde{\Omega^2}_{\ep}(s)=\Omega^2_{RN}(s)$ for all $s>s_{\infty}(\ep)$. Then for $(f_{\ep},\tilde{f}_{\ep},f_{RN})  \in \{\ (r^2_{\ep},\tilde{r^2}_{\ep},r^2_{RN}),\ (\Omega^2_{\ep},\tilde{\Omega^2}_{\ep},\Omega^2_{RN})\ \}$, we say that $f_{\ep}$ converges uniformly (respectively in $L^p(\RR_s)$, resp. in $\mathcal{D}'(\RR_s)$) to $f_{RN}$ if $s \in \RR \rightarrow \tilde{f}_{\ep}(s)$ converges uniformly (respectively in $L^p$, resp. in distribution) to $s\in \RR \rightarrow f_{RN}(s)$ as $\ep \rightarrow 0$, denoted  $f_{\ep} \overset{L^{\infty}(\RR_s)}{\rightarrow}f_{RN}$ (resp. $f_{\ep} \overset{L^p(\RR_s)}{\rightarrow}f_{RN}$, resp. $f_{\ep} \overset{\mathcal{D}'(\RR_s)}{\rightharpoonup}f_{RN}$).

		\end{defn}
		
		\begin{thm} \label{convergence.theorem}
			Under the assumptions of Theorem~\ref{maintheorem}, the following convergence holds (as in Definition \ref{distrib.conv}):
			\begin{align} \label{convergence}
				&\Omega^2_{\ep}  \overset{\mathcal{D}'(\RR_{s})}  \rightharpoonup \Omega_{RN}^2, \hskip 10 mm \Omega^2_{\ep}  \overset{L^{\infty}(\RR_{s})}  \rightarrow \Omega_{RN}^2, \hskip 10 mm \Omega^2_{\ep}  \overset{L^{p}(\RR_{s})}  \rightarrow \Omega_{RN}^2, \hskip 10 mm ,\\ & r^2_{\ep}  \overset{\mathcal{D}'(\RR_{s})}  \rightharpoonup r_{RN}^2, \hskip 12 mm  \lim_{\ep \rightarrow 0} \sup_{s \leq \ep^{-2+\alpha}} \bigl| r^2(s)- r^2_{RN}(s) \bigr|=0
			\end{align} for all $\alpha \in (0,2)$. 
			Nevertheless, $r_{\ep}^2$ does not converge to $r_{RN}^2$ uniformly or in $L^p(\RR_s)$ for any $p\geq 1$: we have $$ \inf_{s \geq \ep^{-1} } r^2_{\ep}(s)=0, \hskip 10 mm  \int_{\ep^{-1}}^{s_{\infty}} |r^2(s)-r^2_{RN}(s)|^p ds \gtrsim \ep^{-2}.$$

		\end{thm}
		\subsection{Uniform Kasner-like behavior}
		\begin{thm} \label{Kasner.thm}
			Under the assumptions of Theorem~\ref{maintheorem}, the metric can be expressed in the following form, defining the proper time variable $\tilde{\tau}>0$ as $g_{\tilde{\tau} \alpha}=-\delta_{\alpha}^{\tilde{\tau}}$ (i.e.\ null shift, unitary lapse gauge, also called synchronous gauge):
			\begin{equation} \label{metric.Kasner.unif}
				g_{\ep}= -d\tilde{\tau}^2 +  \left[ 1+  \mathfrak{E}^{rad}_{\ep}(\tilde{\tau})\right]\cdot  \tilde{\tau}^{2 [1-4b_-^2 \ep^2 \cdot (1+\color{black} \ep^2 \cdot \log(\ep^{-1}) \color{black} \cdot P(\tilde{\tau}))]}  d\rho^2  +    \left[1+\mathfrak{E}^{ang}_{\ep}(\tilde{\tau})\right]\cdot \tilde{\tau}^{4b_-^2\ep^2 (1+\ep \cdot P(\tilde{\tau}))} \cdot  (\color{black}\tilde{r}^2_{-}\color{black} \cdot   \color{black}  d\sigma_{\mathbb{S}^2}\color{black}),
			\end{equation}
			where \color{black} $\tilde{r}_-=  \eta^{\ep^2}\cdot r_-(M,e,\Lambda) $ for some $\eta(M,e,\Lambda,m^2)>0$ and \color{black} $P(\tilde{\tau})$, $\mathfrak{E}^{rad}_{\ep}(\tilde{\tau})$ and $\mathfrak{E}^{ang}_{\ep}(\tilde{\tau})$ obey the following estimates in the region \color{black} $  \mathcal{C}:=\{0< \tilde{\tau} < \tilde{\tau}_{\mathcal{C}}(\ep)\}$, where $\tilde{\tau}_{\mathcal{C}}(\ep)=\color{black}\Gamma \cdot \ep^4$  for some $0<\Gamma<\Gamma_0(M,e,\Lambda,m^2)$\color{black}: there exists $P_0(M,e,\Lambda,m^2)>0$, $C_0(M,e,\Lambda,m^2)>0$ such that \begin{align} \label{main.uniform.0}
				&\sup_{\tilde{\tau} \in \mathcal{C}} |P|(\tilde{\tau}) \leq P_0,\\  \label{main.uniform.1}&\sup_{\tilde{\tau} \in \mathcal{C} }|\mathfrak{E}^{rad}_{\ep}|(\tilde{\tau}) \leq C_0 \cdot\ep^2 \log^2(\ep^{-1})\color{black},\\ \label{main.uniform.2}
				& \sup_{\tilde{\tau} \in \mathcal{C} } |\mathfrak{E}^{ang}_{\ep}|(\tilde{\tau}) \leq C_0 \cdot \ep^4 \log^2(\ep^{-1}) .\end{align}
			In particular, $\mathfrak{E}^{ang}_{\ep}(\tilde{\tau})$ and $\mathfrak{E}^{rad}_{\ep}(\tilde{\tau})$ tend uniformly to $0$ in $ \mathcal{C}$  as $\ep$ tends to $0$.
			
			We also have the following estimate comparing the proper time $\tilde{\tau}$ and the re-normalized area-radius $X:=\frac{r}{r_-}$: \color{black} for some $\tau_{eff}(M,e,\Lambda,m^2)>0$ \color{black}\begin{equation} \label{Kasner.time}
				\sup_{ \tilde{\tau} \in \mathcal{C}}\ \bigl|\ \log(\frac{\tilde{\tau}  }{ \color{black} \tau_{eff}(M,e,\Lambda,m^2) \color{black}\cdot X(\tilde{\tau})^{\frac{  \ep^{-2} }{2 b_-^{2}}\cdot (1+ \color{black} \ep^2 \cdot \log(\ep^{-1}) \color{black} \cdot  P(\tilde{\tau}))^{-1}} }) \ \bigr| \ls \ep^2 \log^2(\ep^{-1}).
			\end{equation}  Additionally, the scalar field blows-up when $\tilde{\tau}$ approaches $0$: there exists $Q_0(M,e,\Lambda,m^2)>0$, $|Q(\tilde{\tau})|< Q_0$ such that  \begin{equation} \label{SF.thm}
				\sup_{\tilde{\tau} \in \mathcal{C}}\ \bigl|\ \phi(\tilde{\tau})- 2b_- \cdot \ep \cdot (1+\color{black} \ep^2 \cdot \log(\ep^{-1}) \color{black}\cdot  Q(\tilde{\tau}))\cdot \log (\tilde{\tau}^{-1})\ \bigr|\ \lesssim \color{black} \epsilon \cdot \log(\ep^{-1}). \color{black}
			\end{equation}

			Moreover, we have the following estimates in the  $L^p([0,\tilde{\tau}_{\mathcal{C}}])$ norm for $1 \leq p < \ep^{-2}$\color{black}:
			\begin{align} \label{metric.Kasner.BMO}
				& g_{\ep}= -d\tilde{\tau}^2 +  \tilde{\tau}^2 \cdot \left(1 + E^{rad}_{\ep}(\tilde{\tau})\right)  d\rho^2  +    \left( 1 +E^{ang}_{\ep}(\tilde{\tau})\right) \cdot  (r_-^2 \cdot   \color{black} d\sigma_{\mathbb{S}^2})\color{black},\\
				\label{main.BMO1}
				&  \tilde{\tau}^{-\frac{1}{p}}_{\mathcal{C}}\|   E^{rad}_{\ep}(\tilde{\tau}) \|_{L^p([0,\tilde{\tau}_{\mathcal{C}}])} \lesssim  \frac{\ep^2 \log(\ep^{-1})}{1-p\cdot \ep^2} ,\\
				\label{main.BMO2}
				& \tilde{\tau}^{-\frac{1}{p}}_{\mathcal{C}}\|   E^{ang}_{\ep}(\tilde{\tau}) \|_{L^p([0,\tilde{\tau}_{\mathcal{C}}])} \lesssim   \frac{\ep^2 \log(\ep^{-1})}{1-p\cdot \ep^2} 
			\end{align}
			In particular, $  \tilde{\tau}^{-\frac{1}{p}}_{\mathcal{C}}\color{black}E^{ang}_{\ep}$ and $ \tilde{\tau}^{-\frac{1}{p}}_{\mathcal{C}}\color{black}E_{\ep}^{rad}$ converge  to $0$ in $L^p([0,\tilde{\tau}_{\mathcal{C}}])$ \color{black} as $\ep$ tends to $0$.
		\end{thm}

		\begin{rmk}
			Up to the errors $\mathfrak{E}^{ang}_{\ep}$, $\mathfrak{E}^{rad}_{\ep}$ (which converge to $0$), $(g_{\ep},\phi)$ corresponds to a Kasner metric of the form \eqref{Kasner.exact} satisfying the relation \ref{p.Kasner} (the first one exactly, the second one to top order). Namely:  \begin{align}
				&p_{rad}=1-4  b_-^2 \ep^2 \cdot (1+\ep \cdot P(\tilde{\tau})), \hskip 3 mm p_x= p_y = 2 b_-^2 \ep^2 \cdot (1+\ep \cdot P(\tilde{\tau})), \hskip 3 mm p_{\phi}= 2 b_- \cdot \ep \cdot (1+\ep \cdot Q(\tilde{\tau})) ;\\ & \hskip 30 mm p_{rad}+p_x+ p_y =1, \hskip 10 mm p_{rad}^2+p_x^2+ p_y^2 + 2 p_{\phi}^2 =1 + O(\ep^3).
			\end{align}
			Note that $p_x$, $p_y$, $p_{rad}$, $p_{\phi}$ depend on $\tilde{\tau}$, although the dependence is lower order in $\ep$, so, even assuming $\mathfrak{E}^{ang}_{\ep}=\mathfrak{E}^{rad}_{\ep}=0$, $(g_{\ep},\phi)$ is not an exact Kasner metric. This is due to the non-linear back-reaction which already acts as a (non-trivial) perturbation in the early-time regions where the dynamics are mostly linear. 
		\end{rmk}

		\begin{rmk}	Because we work on a time-interval $[0,\tilde{\tau}_{\mathcal{C}}]$ where $\tilde{\tau}_{\mathcal{C}} \simeq \ep^4$, it is important to prove estimates in a \textit{scale-invariant} way, which is why  we considered  $\tilde{\tau}^{-\frac{1}{p}}_{\mathcal{C}}E^{rad}_{\ep}$ and $\tilde{\tau}^{-\frac{1}{p}}_{\mathcal{C}}E^{ang}_{\ep}$ instead of  $E^{rad}_{\ep}$ and $E^{ang}_{\ep}$. \end{rmk}

						\section{Proof of Theorem~\ref{maintheorem}} \label{main.section}

						\subsection{Preliminary estimates up to the no-shift region} \label{NS.section1}
						
						In this section, we prove that there exists a time $\Delta_0(\ep) \simeq \log(\epsilon^{-1})$ up to which we are in the linear regime and we derive sub-optimal estimates using a basic Gr\"{o}nwall argument. Nevertheless, although we prove the estimates up to $s\simeq \log(\epsilon^{-1})$, we will ultimately only apply them for $s \leq S''(M,e,\Lambda,m^2)$ where $S''(M,e,\Lambda,m^2)>0$ is a large constant independent of $\epsilon$. More refined estimates (building up on this section) will be derived in Section~\ref{difference.section} in the region $s \simeq \log(\epsilon^{-1})$. We consider two different regions for the moment: the red-shift region $\mathcal{R}:= \{ s \leq -\Delta\}$ where $\Delta(M,e,\Lambda)>0$ is a large number and the no-shift region $\mathcal{N}:= \{ -\Delta \leq s \leq \Delta_0(\ep)\}$ for $\Delta_0(\ep) = c(M,e,\Lambda,m^2) \cdot \log(\epsilon^{-1})>0$.
						
						\begin{lem} \label{lemma.RS1} There exist $D_R(M,e,\Lambda,m^2)>0$ and $\Delta(M,e,\Lambda)>0$  such that for all $ s \in \mathcal{R}:=\{ s \leq -\Delta\}$:
							\begin{equation} \label{RS1}
								|\phi|(s)+\frac{|\dot{\phi}|(s)}{\Omega^2(s)} \leq D_R \cdot \epsilon,
							\end{equation}
							\begin{equation} \label{RS2}
								|\kappa^{-1}(s)-1| \leq D_R \cdot \epsilon^2
							\end{equation}
							\begin{equation} \label{RS3.0}
								|r(s)-r_+(M,e,\Lambda)| +|2K(s)-2K_+|  \leq D_R \cdot \Omega^2(s),
							\end{equation}
							\begin{equation} \label{RS3}
								|\varpi(s)-M|  \leq D_R \cdot \Omega^2(s) \cdot \epsilon^2,
							\end{equation}
							\begin{equation} \label{RS4}
								|\frac{d\log(\Omega^2)}{ds}(s)-2K(s) | +	|\log(\Omega^2)(s)-2K_+\cdot s | \leq D_R \cdot[ \Omega^2(s)\color{black}+\epsilon^2] \color{black} .
							\end{equation}

						\end{lem}

						\begin{proof}
							The proof is a simpler version of Proposition 4.5 in \cite{Moi}. For the convenience of the reader we briefly sketch the argument. We make the following two bootstrap assumptions \begin{equation} \label{BSR1}
								|	\frac{d \log(\Omega^2)}{ds}(s)- 2K_+\color{black} | \leq K_+, \hskip 4 mm |\phi|(s) \leq 4 D_R \cdot \epsilon.
							\end{equation}  The key estimate is to use \eqref{FieldODE3} to derive for some $\frac{K_+}{8}<a<\frac{K_+}{4}$ $$ \frac{d}{ds} \left( \frac{r^2 \dot{\phi} e^{as}} {\Omega^2}\right)= (a -\frac{d}{ds} \log(\Omega^2)) \frac{r^2 \dot{\phi}\cdot e^{as}} {\Omega^2}- m^2 r^2 \phi \cdot e^{as},$$ then we integrate, using the fact that $(a -\frac{d}{ds} \log(\Omega^2))< -\frac{3K_+}{4}$ and Gr\"{o}nwall's inequality, together with \eqref{dotphi.data} and \eqref{BSR1}. This shows that for all $s \in \mathcal{R}$: $$ |\dot{\phi}|(s)  \lesssim D_R \cdot \Omega^2(s) \cdot \epsilon, $$ which closes the bootstrap of $\phi$ since $\int_{-\infty}^{s}\Omega^2(s') ds' \lesssim e^{-K_+ \Delta }< \frac{1}{4}$,  if  we choose $\Delta$ large enough. The other estimates follow from integration and improve the other bootstrap assumptions from \eqref{BSR1}, see \cite{Moi} for details.
						\end{proof}
						
						\begin{lem} \label{lemma.N1} There exist $C=C(M,e,\Lambda,m^2)>1$ and $D_N(M,e,\Lambda,m^2)>0$ such that the following estimates hold for all $ s \in \mathcal{N}:=\{ s \leq \Delta_0(\epsilon)\}$, where $\Delta_0(\epsilon) = \frac{\log(C(M,e,m^2))^{-1}}{2}\cdot \log(\epsilon^{-1})$ \color{black} is so that $C^{s} \cdot \epsilon = \sqrt{\epsilon}$\color{black}
							\begin{equation} \label{N1}
								|\phi|(s)+|\dot{\phi}|(s) \leq D_N \cdot C^{s} \cdot \epsilon,
							\end{equation}
							\begin{equation} \label{N2}
								|\kappa^{-1}-1|(s)\leq D_N \cdot (C^{s} \cdot \epsilon)^2
							\end{equation}
							\begin{equation} \label{N3}
								|\frac{d\log(\Omega^2)}{ds}-2K|(s) \leq D_N \cdot (C^{s} \cdot \epsilon)^2,
							\end{equation}
							\begin{equation} \label{N4}
								|\varpi-M|(s)   \leq D_N \cdot (C^{s} \cdot \epsilon)^2.
							\end{equation}
							
						\end{lem}
						\begin{proof}
							We bootstrap \eqref{N1}-\eqref{N4} replacing $C^{s}$ by $C^{\frac{5s}{4}}$. In particular, for $\epsilon$ small enough and given that $C^s \leq \epsilon^{-\frac{1}{2}}$ by definition of $\Delta_0$, the estimates \eqref{N1}--\eqref{N4} are satisfied with $C^s \ep$ replaced by $\epsilon^{\frac{3}{8}}$. We also bootstrap \begin{equation}\label{BT}
								\Omega^2(s) \leq 4D_N.
							\end{equation}  Then, coupling the bootstraps with \eqref{mu.static}, it is clear that for some $E(M,e,\Lambda)>0$ \begin{equation}\label{rN}
								r^{-1}(s) \leq E(M,e,\Lambda)
							\end{equation} for all $s \in \mathcal{N}$. Therefore, by \eqref{FieldODE2}, using the bootstraps and defining $\varphi:= r^2 \dot{\phi}$, there exists $C(M,e,\Lambda,m^2)>0$ such that $$ |\frac{d}{ds}(\varphi,\phi)| \leq C(M,e,\Lambda,m^2) \cdot |(\varphi,\phi)|. $$ Therefore, by Gr\"{o}nwall's inequality, the bootstrap corresponding to \eqref{N1} is retrieved. It is then not difficult to retrieve the other bootstraps using \eqref{N1} and improve the $C^{\frac{5s}{4}}$ weights to a smaller $C^{s}$ weight (recall $C>1$). \color{black} Note in particular that to retrieve the bootstrap assumption \eqref{BT}, we use the identity \eqref{mu.static} and the bootstrap $|\varpi-M| \leq D_N (C^{\frac{5s}{4}} \ep)^2$ \color{black} to get $$ \Omega^2 \leq 2 |1-\frac{2M}{r}+ \frac{e^2}{r^2}-\frac{\Lambda}{3} r^2| + \frac{ \sqrt{\epsilon}}{r},$$ which  with \eqref{rN} is sufficient to conclude (after taking $D_N$ large enough, say $D_N > 4 \underset{r_-< r<r_+} {\sup} |1-\frac{2M}{r}+ \frac{e^2}{r^2}-\frac{\Lambda}{3} r^2|$  and $\ep$ small enough so that $\sqrt{\ep}\cdot  E(M,e,\Lambda) \leq \frac{D_N}{2}$ \color{black}). \end{proof}

						\subsection{Estimates for the difference with  linearized quantities} \label{difference.section}
						
						In this section, we control the difference between the Reissner--Nordstr\"{o}m quantities $(r_{RN},\Omega^2_{RN}, \phil)$ and $(r,\Omega^2, \phi)$. We shall denote these differences $\dphi= \phi- \phil$, $ \domega = \Omega^2 -\Omega^2_{RN}$, $\dr=r-r_{RN}$. From now on, we will use the above notation throughout the paper.
						
						\begin{prop} \label{prop.RS.d}
							There exists $D'_{R}(M,e,\Lambda,m^2)>0$ such that the following estimates are true for all $s\in \mathcal{R}$:	\begin{equation} \label{diffR}
								|\domega|(s) \leq  D'_{R} \cdot \epsilon^2 \cdot \color{black}\Omega^2(s)\color{black}, \end{equation}\begin{equation}\label{diff.1.75}
								|\delta \log(\Omega^2)|(s) \leq  D'_{R} \cdot  \epsilon^2,
							\end{equation}\begin{equation} \label{diffR1.5}
								|\frac{d}{ds}(\delta \log(\Omega^2))|(s) \leq  D'_{R} \cdot  \epsilon^2 \cdot \Omega^2(s) \color{black},
							\end{equation}\begin{equation} \label{diffR2}
								|\dr|(s) \leq  D'_{R} \cdot \epsilon^2 \cdot \Omega^2(s),
							\end{equation} \begin{equation} \label{diffR3}
								|\ddr|(s) \leq  D'_{R} \cdot \epsilon^2 \cdot \Omega^2(s),
							\end{equation}\begin{equation} \label{diffR4}
								|\dphi|(s) \leq  D'_{R} \cdot \epsilon^3 \cdot \Omega^2(s),
							\end{equation} \begin{equation} \label{diffR5}
								|\ddphi|(s) \leq  D'_{R} \cdot \epsilon^3  \cdot \Omega^2(s).
							\end{equation}
						\end{prop}
						
						\begin{proof}   \color{black}$\delta r$,  $\delta \phi$, $\delta \dot{\phi}$,   and $  \frac{d }{ds} (\delta \log(\Omega^2))$ \color{black} are zero,  and  $   \delta \log(\Omega^2)=O(\ep^2)$ \color{black} for the initial data on $\mathcal{H}^+$. We make the following bootstrap assumptions \begin{equation} \label{BSR.d}
								\Omega^{-2} \cdot |\delta \Omega^2|,\ |\delta r | ,\ |\delta\dot{r} | ,\ \epsilon^{-1} \cdot |\delta \phi|  \leq 4 D'_R \cdot \epsilon^2.
							\end{equation}The key point is to notice $(r_{RN},\Omega^2_{RN})$ satisfy \eqref{Radiusstat}  and \eqref{Omegastat} with $\phi \equiv 0$ and $ \phil$ satisfies \eqref{FieldODE} with  $(r,\Omega^2)$ replaced by  $(r_{RN},\Omega^2_{RN})$. Then we can take the difference of the two equations: for instance for \eqref{FieldODE} $$\delta \ddot \phi= -\frac{2 (\delta \dot{\phi}) \dot{r}}{r}- 2 \dot{\phil} ( \frac{\delta \dot{r}}{r}- \frac{\dot{r_{RN}}\cdot \delta r}{r_{RN} }) - m^2 \Omega^2 \dphi -m^2 \phil\cdot  \delta \Omega^2.$$ We will keep using this technique repetitively without referring explicitly to it. Plugging \eqref{BSR.d} and the bounds of Lemma \ref{lemma.RS1} into the above we get for all $s$: $$| \delta \ddot \phi |(s) \ls D'_R \cdot  \epsilon^3  \cdot \Omega^2(s)+ |\dot{r}| \cdot | \delta \dot \phi |.$$ We can integrate this estimate using Gr\"{o}nwall's inequality, since $|\dot{r}|$ is integrable and we get \begin{equation}
								| \delta \dot \phi |(s) + | \delta \phi |(s) \ls D'_R \cdot  \epsilon^3  \cdot \Omega^2(s),
							\end{equation} where to obtain the $\delta \phi$ on the LHS we integrated a second time and took advantage of $\int_{-\infty}^{s} \Omega^2(s') ds' \ls \Omega^2(s)$. Since $\Omega^2(s) \ls e^{-2K_+ \Delta}$ in $\mathcal{R}$, the $\dphi$ part of \eqref{BSR.d} is improved, choosing $\Delta$ large enough. Then by \eqref{Omegastat} we get $$| \frac{d^2 \delta \log(\Omega^2)}{ds^2}| \ls D'_R \cdot \epsilon^2\cdot  \Omega^2,$$ which upon double integration (and noting $|\delta \Omega^2| \ls \Omega^2 \cdot |\delta \log(\Omega^2)| $ using the Taylor expansion of $\log$) gives $$   |\delta \log(\Omega^2)(s)|+ |\frac{d}{ds}(\delta \log(\Omega^2))|(s)+ \Omega^{-2}(s) \cdot |\delta \Omega^2|(s) \ls D'_R \cdot  \color{black} \epsilon^2\color{black}, $$ which improves the $\Omega^{-2}(s) \cdot |\delta \Omega^2|(s)$ part of \eqref{BSR.d} using the smallness of $\Omega^2$. The other two estimates can be improved similarly and the bootstraps close. Note that \eqref{BSR.d} is now closed with $D'_R \cdot \epsilon^2 \cdot \Omega^2$ replacing the RHS, which allows to re-do the argument with these improved weights and obtain all the claimed estimates.
							
						\end{proof}
						
						\begin{prop} \label{prop.N.d}
							
							There exists $D'_{N}(M,e,\Lambda,m^2)>0$ such that the following estimates are true for all $s\in \mathcal{N}$:
							\begin{equation} \label{diff.N.1}
								|\delta \log \Omega^2|(s) \leq D'_{N}\cdot \epsilon^2 \cdot C^s,
							\end{equation} 
							\begin{equation} \label{diff.N.2}
								|\frac{d}{ds}(\delta \log \Omega^2)|(s) \leq D'_{N}\cdot \epsilon^2 \cdot C^s,
							\end{equation}\begin{equation} \label{diff.N.3}
								|\dr|(s)\leq D'_{N}\cdot \epsilon^2 \cdot C^s ,
							\end{equation} \begin{equation} \label{diff.N.4}
								|\ddr|(s) \leq D'_{N}\cdot \epsilon^2 \cdot C^s ,
							\end{equation}\begin{equation} \label{diff.N.5}
								|\dphi|(s)\leq D'_{N}\cdot \epsilon^3 \cdot  C^s,
							\end{equation} \begin{equation} \label{diff.N.6}
								|\ddphi|(s) \leq D'_{N}\cdot \epsilon^3  \cdot C^s .
							\end{equation}
						\end{prop} \begin{proof}
							The proof is similar to that of Lemma \ref{lemma.N1}, taking advantage of the bounds already proved in Lemma \ref{lemma.N1}.
						\end{proof}
						
						\begin{cor} \label{cor.NS}There exists $S'(M,e,\Lambda)>0$, $D_-'(M,e,\Lambda,m^2)>0$, $B(M,e,\Lambda,m^2) \neq 0$ such that for all $s \in \mathcal{N}':= \{ S'(M,e,\Lambda) \leq s \leq \Delta_0(\epsilon)\}$: \begin{equation}
								|r(s)-r_-| \leq \epsilon^{\frac{3}{2}} + |r_{RN}(s)-r_-| \leq \epsilon^{\frac{3}{2}}+ D_-' \cdot e^{2K_- s},
							\end{equation}
							\begin{equation}
								D_-'^{-1} \cdot e^{2K_-s} \leq  \Omega^2(s) \leq D_-' \cdot e^{2K_-s},
							\end{equation}
							\begin{equation} \label{asymplin2}
								|\phi|(s) \leq D_-' \cdot( \epsilon \cdot s+  e^{2K_-s}),
							\end{equation}	\begin{equation} \label{asymplin3}
								|\dot{\phi}-B \cdot \epsilon|(s) \leq D_-' \cdot (\epsilon^2+ e^{2K_-s}).
							\end{equation}
							In particular, there exists $S''(M,e,\Lambda,m^2)\geq S'(M,e,\Lambda,m^2)> 0$, $D'_-(M,e,\Lambda,m^2)>0$ such that for all 	$s \in \mathcal{N}'':= \{ S''(M,e,\Lambda,m^2) \leq s \leq \Delta_0(\epsilon)\}$: 
							\begin{equation} \label{diff}
								|\frac{d\log(\Omega^2)}{ds}(s)-2K_-| \leq 2D'_- \cdot (\epsilon^{\frac{3}{2}}+   e^{2K_- s})< \frac{|K_-|}{100}. \end{equation}

						\end{cor}
						
						\begin{proof}
							Immediate from Proposition \ref{prop.N.d} and \eqref{asymplin}. Note that the existence of a $D_-'(M,e,\Lambda)>0$ and  $S'(M,e,\Lambda)>0$ such that $   D_-'^{-1} \cdot e^{2K_- s} \leq |r_{RN}(s)-r_-| \simeq \Omega^2_{RN} \leq D_-' \cdot e^{2K_- s}$ for $s \geq S'$ is an easy computation on \eqref{RN}. For \eqref{diff}, similarly there exists $s \geq S''(M,e,\Lambda,m^2)$, $D_-''(M,e,\Lambda)>0$ such that for all $s\geq S''$,  $$ |\frac{d\log(\Omega^2_{RN})}{ds}(s)-2K_-| \leq D_-'' \cdot e^{2K_- s}< \frac{|K_-|}{400}.$$
						\end{proof}
						
						Note that we have the freedom to choose $S''(M,e,\Lambda,m^2)$ large enough and Corollary \ref{cor.NS} still applies with slightly modified constant. We will make use of this fact in Section~\ref{EB.section}.
						
						\color{black}

						\subsection{Estimates on the early blue-shift region}  \label{EB.section}
						
						We work on the early blue-shift region $\mathcal{EB}:=\{ S''\leq s \leq s_{lin}(\ep)\}$ for $s_{lin}(\epsilon) := (2 |K_-|)^{-1} \log(\nu(M,e,\Lambda,m^2) \cdot \epsilon^{-2})$ with $\nu(M,e,\Lambda,m^2)>0$ a small constant to be fixed later. Note that $\mathcal{EB}$ and $\mathcal{N}$ \textit{overlap} on the region $\{S'' \leq s \leq \Delta_0\}$, but from Section~\ref{NS.section1} and \ref{difference.section}, we will only apply the estimates on $\{s=S''\}$ the past boundary of $\mathcal{EB}$.

						\begin{prop} \label{EB.Prop} There exists $D_E(M,e,\Lambda,m^2)>0$ such that for all $s \in \mathcal{EB}$, we have \begin{equation} \label{EB0}
								|	\Omega^{0.1} \cdot \delta \phi| \leq  D_E \cdot \epsilon^3, 
							\end{equation}
							\begin{equation} \label{EB1}
								|\delta (r^2\dot{\phi})|(s) \leq   D_E \cdot \epsilon^3,
							\end{equation}\begin{equation} \label{EB2}
								s^{-1}	|\delta\phi|(s)|+  |\delta\dot{\phi}|(s)| \leq D_E \cdot \epsilon^3 \cdot \color{black}s\color{black},
							\end{equation}\begin{equation} \label{EB3}
								s^{-1} \cdot |\delta \log(\Omega^2)|(s)+ |\frac{d}{ds}(\delta \log(\Omega^2))|(s) \leq D_E \cdot \epsilon^2 \color{black}\cdot s\color{black}
							\end{equation}\begin{equation} \label{EB4}
								s^{-1} \cdot |\delta r|(s)+ |\delta\dot{r}|(s) \leq D_E \cdot \epsilon^2,
							\end{equation}\begin{equation} \label{EB5}
								|\kappa^{-1}-1|(s)\leq  0.1,
							\end{equation}
							\begin{equation} \label{EB6}
								|\varpi-M|(s)   \leq  D_E \cdot \epsilon^2 \cdot s.
							\end{equation}

							Moreover, there exists $C_-(M,e,\Lambda)>0$ such that we have the following estimates at $s=\s(\epsilon)$:  \begin{equation} \label{EB7}
								|\dot{\phi}(\s)- B \cdot \epsilon| \leq  D_E \cdot \epsilon^3 \cdot \color{black}\log(\epsilon^{-1})\color{black},
							\end{equation} \begin{equation} \label{EB8}
								|\Omega^2(\s) - \epsilon^2 \cdot C_-\cdot   \color{black}\nu^{-1}\color{black}(M,e,\Lambda)|= |\Omega^2(\s) -  C_-  \cdot e^{2K_- \s}| \leq D_E \cdot \epsilon^4 \cdot \color{black}\log(\epsilon^{-1})^2\color{black},
							\end{equation} 	\begin{equation}\label{EBslin2}
								\color{black} \bigl|\ -\dot{r}(\s)-  \frac{B^2 r_-}{2|K_-|} \epsilon^2\  -\Omega^2(\s)\bigr| \leq D_E \cdot \epsilon^4 \cdot \log(\epsilon^{-1})^2 ,\color{black}
							\end{equation} \begin{equation} \label{r-r_-EB}
								|r(\s)-\color{black}[1-\frac{B^2}{2|K_-|}\epsilon^2 \s(\epsilon)] \color{black}r_-(M,e,\Lambda)| \leq D_E \cdot \epsilon^2,
							\end{equation}\begin{equation} \label{EB11} |\frac{d \log(\Omega^2)}{ds}(\s)-2K_-| \leq  D_E \cdot \epsilon^2 \log(\ep^{-1}) \color{black},\end{equation}
							\color{black}
							\begin{equation}\label{EB12}
								\bigl|  \int_{S''}^{\s} \frac{r(s)}{\Omega^2(s)} ds\ - \frac{ r_-}{2|K_-| \cdot \Omega^2(\s)}  \bigr| \leq D_E \cdot \ \log(\ep^{-1}).
							\end{equation}
							\color{black}
							By \eqref{r-r_-EB}, note in particular that for $\ep>0$ small enough, \begin{equation}\label{rEBs}
								r(\s)<r_-(M,e,\Lambda).
							\end{equation} 
							
							In view of \eqref{r-r_-EB}, we define $C_{EB}$ as follows, and it obeys the following inequality \begin{equation} \label{CEB.def}
								C_{EB}:=  \ep^{-2} \cdot \log(\frac{r(\s)}{\color{black}[1-\frac{B^2}{2|K_-|}\epsilon^2 \s] \color{black}r_-}), \hskip 5 mm |C_{EB}| \leq \frac{2 D_E(M,e,\Lambda,m^2)}{r_-(M,e,\Lambda)}
							\end{equation}

						\end{prop}
						\begin{proof} Our strategy involves a bootstrap method and follows three steps: \begin{enumerate}
								\item  we first close the bootstrap estimates on the geometric quantities $(r,\Omega^2)$.
								\item Then, we close the bootstrap estimates on the scalar field $\phi$.
								\item Finally, we obtained improved estimates on $(r,\Omega^2)$.
							\end{enumerate} \color{black}
							We bootstrap the estimates, with $B_1(M,e,\Lambda,m^2)>0$ to be determined in the proof \color{black} and $B$ defined in Theorem~\ref{KSR.thm}   \begin{equation} \label{B1EB}	|\delta \log(\Omega^2)| \leq 4B_1 \cdot \epsilon^2 \cdot \color{black}s^2\color{black}.\end{equation}
							\begin{equation}\label{B2EB}|\dot{\phi}| \leq 4B \cdot \epsilon,\end{equation}  
							\begin{equation} \label{B4EB}
								r^{-1}(s) \leq 10 r_-^{-1}(M,e,\Lambda),
							\end{equation}  where we have assumed $B_1$ large enough so that \eqref{B1EB} holds true at $s=S''$. \color{black}
							Note that taking $S''(M,e,\Lambda,m^2)$ large enough (recall that $S''$ only depends on $M$, $e$, $\Lambda$, $m^2$\color{black}) ensures that  \eqref{B4EB} is satisfied  at $s=S''$\color{black}. 
							
							For $\epsilon$ small enough, using $\s \ls \log(\epsilon^{-1})$ and \eqref{B1EB}, we have for some $C_-(M,e,\Lambda)>0$: $$  \frac{ C_-}{4} \cdot e^{2K_- s}\leq \frac{\Omega^2_{RN}(s)}{2} \leq \Omega^2(s) \leq 2 \Omega^2_{RN}(s) \leq 4 C_- \cdot e^{2K_- s}.$$
							
							Integrating \eqref{Raychstat} and using also \eqref{N2} for $s=S''$ we get for some $D''_N(M,e,\Lambda,m^2)>0$: \begin{equation} \label{kappa.proof.EB}
								|\kappa^{-1}-1| \leq D''_N \cdot \epsilon^2 + D''_N \cdot B^2 \cdot \epsilon^2 \cdot e^{2|K_-| s}\leq 0.1,
							\end{equation}  after taking $\nu$ small enough so that $D''_N \cdot B^2 \cdot \epsilon^2 \cdot e^{2|K_-| \s}= D''_N \cdot B^2 \cdot \nu <0.05$. 
							
							From this and \eqref{B1EB} we get $$ |\delta{ \dot{r}}| \leq 8 C_- \epsilon^2 \color{black}\cdot ( B_1 \cdot e^{2K_- s} \cdot \color{black}s^2\color{black}+ 2 D_N'' \cdot B^2 ).$$  where we used $\delta \Omega^2 \leq 2C_- e^{2K_- s}\cdot  \delta \log(\Omega^2)$ and chose $S''$ large enough. Integrating and using \eqref{diff.N.3} we get $$ |\delta{ r}| \leq    \frac{8C_- \cdot \epsilon^2}{|K_-|}  B_1 \cdot e^{2K_- S''} \color{black}[S'']^2\color{black}+3\color{black} D_N'' \cdot B^2 \cdot \epsilon^2 \cdot  s \leq   \frac{8C_- \cdot \epsilon^2}{|K_-|}  B_1 +3 D_N'' \cdot B^2 \cdot \epsilon^2 \cdot  s,$$ where we used the integrability\footnote{Note indeed that for any $N>0$, we have $\int_{S''}^{s} e^{2K_- x} x^N dx \leq\frac{e^{2K_- S''}( S'')^N}{|K_-|}  $, if $S''$ is large enough.} of $e^{2K_- s} \cdot \color{black}s^2\color{black}$,  and chose $D_N''$ large enough so that $D_N' C^{S''} \leq D_N'' B^2 S'' \leq   D_N'' B^2 s$\color{black}.
							
							Plugging these bounds into \eqref{Omegastat} and using \color{black}\eqref{B2EB}\color{black}, \eqref{B4EB}, the smallness of $\epsilon$ and $\Omega^2 \leq 1$,\color{black} we get $$ |\frac{d^2}{ds^2} ( \delta \log(\Omega^2))|(s) \leq P(M,e,\Lambda) 
						 B_1\color{black}   \cdot\epsilon^2   \cdot
							e^{2K_- s} \cdot \color{black}s^2\color{black}+ D_N''' \cdot \color{black} B^2 \epsilon^2\color{black},$$ where $ P(M,e,\Lambda,B)>0$, $D_N'''(M,e,\Lambda,m^2)>0$\color{black}. Integrating the above using \eqref{diff.N.2} for $s=S''$ gives,  $$ |\frac{d}{ds}( \delta \log(\Omega^2))| \leq \frac{P(M,e,\Lambda)}{|K_-|}   B_1\color{black} \cdot \ep^2 \cdot   e^{2K_- S''} \cdot S''^2\color{black} + D_N''' \cdot \color{black}  B^2 \epsilon^2 s 
							\leq  (B_1+D_N''' \cdot  B^2 )\epsilon^2 s  ,$$ where in the last line we have chosen $S''(M,e,\Lambda,m^2)$ large enough so that  $\frac{P(M,e,\Lambda)}{|K_-|}   \cdot   e^{2K_- S''} \cdot S''^2<1$ (note that $P(M,e,\Lambda)$ does not have any implicit dependence on $S''$). \color{black} Integrating another time \color{black} improves bootstrap \eqref{B1EB} \color{black} if we choose $B_1> D_N'''\color{black}B^2$ and $\epsilon$ small enough. In what follows, we fix $B_1(M,e,\Lambda,m^2)$ in such a way, and thus we need not worry about  the dependence on $B_1$ anymore\color{black}. Bootstrap \eqref{B4EB}   is improved immediately after multiplying \eqref{kappa.proof.EB} and integrating. \eqref{B2EB} remains: for this, note that the following identity for all $f$ follows from $\frac{d\log(\Omega^2)}{ds}<0$: \begin{equation}\label{int.ID.EB}
								\frac{d}{ds}(\Omega\cdot  |f|) \leq \Omega \cdot |\dot{f}|,
							\end{equation} To prove it, consider the derivative of $\Omega^2\cdot  |f|^2$, see Proposition 6.5 in \cite{Moi3Christoph} for details. Then, upon integration with the usual rules  gives \begin{equation}\label{int.ID.EB2}
								\Omega(s)\cdot  |f|(s) \leq \Omega(S'')\cdot  |f|(S'') +\frac{2\Omega(S'')}{|K_-| \cdot r_-^2} \cdot \sup_{S'' \leq t \leq s}|r^2\dot{f}|(t).
							\end{equation}
							
							Now come back to \eqref{FieldODE2} and take differences: we get  $$ |\frac{d}{ds}( \delta(r^2 \dot{\phi}))| \ls \epsilon^3 \cdot e^{2K_- s} \color{black}s^2 \color{black} + e^{K_- s} \cdot \Omega \cdot|\dphi|. $$ To address $e^{K_- s} \cdot \Omega \cdot|\dphi|$,  apply \eqref{int.ID.EB2} to $f= \dphi$: we get using also \eqref{diff.N.5} :  $$ |\frac{d}{ds}( \delta(r^2 \dot{\phi}))|(s) \ls  \epsilon^3 \cdot e^{2K_- s} \color{black}s^2 \color{black} + e^{K_- s} \cdot  \Omega(S'') \cdot \left(  \epsilon^3  + \sup_{S'' \leq t \leq s}|r^2\cdot  \delta \dot{\phi}|(t)\right).$$
							Writing $r^2 \delta \dot{\phi}= \delta(r^2 \dot{\phi})- (r+r_{RN}) \cdot \delta r \cdot \dot{\phil}$, and using the previously proven bounds this implies \begin{equation} \label{EB.wave.hard.est}
								|\frac{d}{ds}( \delta(r^2 \dot{\phi}))|(s) \ls \epsilon^3 \cdot e^{2K_- s} \color{black}s^2 \color{black} + e^{K_- s} \cdot  \Omega(S'') \cdot \left(   (1 +s)\cdot \epsilon^3  + \sup_{S'' \leq t \leq s}|\delta(r^2\cdot   \dot{\phi})|(t)\right) .
							\end{equation} Integrating \eqref{EB.wave.hard.est} using \eqref{diff.N.3}, \eqref{diff.N.6} for $s=S''$ gives, using the Gronwall inequality\color{black}: $$  \sup_{S'' \leq t \leq s}|\delta(r^2 \dot{\phi})|(t) \ls \exp(\int_{S''}^{s}  e^{K_- s} \cdot  \Omega(S'') ds)  \epsilon^3 \ls \ep^3,$$ thus, using the bounds of  \eqref{asymplin} is more than sufficient to improve the bootstrap \eqref{B2EB}.

						Now that all bootstraps are closed, we turn to improved estimates.	\color{black}
							For \eqref{EBslin2} and \eqref{r-r_-EB},	note in particular the following formula: $$   \delta \dot{r}= -\Omega^2 [\kappa^{-1}-1] - \delta \Omega^2,$$ and we can integrate \eqref{Raychstat} to obtain \begin{equation*}
								|- \delta \dot{r}(s) - \Omega^2(s)\int_{S''}^{s} \frac{r |\dot{\phi}|^2}{\Omega^2} ds| \leq | \delta \Omega^2|(s) +  \Omega^2(s)\ | \kappa^{-1}(S'')-1|.
							\end{equation*} So by \eqref{N2} applied at $s=S''$ and \eqref{EB3} we get  \begin{equation}\label{drEB}
								|- \delta \dot{r}(s) - \Omega^2(s)\int_{S''}^{s} \frac{r |\dot{\phi}|^2}{\Omega^2} ds| \lesssim \epsilon^2\cdot s^2 e^{2 K_- s}.
							\end{equation} 
							
							Now from \eqref{drEB},\eqref{EB2} and the inequality $$ |\dot{\phil}(s) - B \cdot \epsilon|\lesssim \epsilon \cdot  e^{2K_- s}$$ we get \begin{equation*}
								|- \delta \dot{r}(s) - B^2 \epsilon^2\cdot  \Omega^2(s)\int_{S''}^{s} \frac{r }{\Omega^2} ds| \lesssim \epsilon^2\cdot s^2 e^{2 K_-s} + \epsilon^4 s.
							\end{equation*}
							Note also that by \eqref{EB4}, \eqref{EB3} and using the fact that $r_{RN}-r_-= O (\Omega^2_{RN})$ we have  \begin{equation*}
								|- \delta \dot{r}(s) - B^2 r_- \cdot  \epsilon^2\cdot  \Omega^2(s)\int_{S''}^{s} \frac{ds}{\Omega^2} | \lesssim \epsilon^2\cdot s^2 e^{2 K_-s} + \epsilon^4 s.
							\end{equation*} Now we use identity $\Omega^{-2} = \frac{\frac{d}{ds} (\Omega^{-2})}{-\frac{d}{ds} \log(\Omega^{2})}$; again by \eqref{EB3} and the fact that   $\frac{1}{\frac{d}{ds} \log(\Omega^{2}_{RN})}-\frac{1}{2K_-}= O (\Omega^2_{RN})$ we get  \begin{equation}
								|- \delta \dot{r}(s) - \frac{B^2 r_- }{2|K_-|} \ \epsilon^2  | \lesssim \epsilon^2\cdot s^2 e^{2 K_-s} + \epsilon^4 s,
							\end{equation} which gives \eqref{EBslin2}. Integrating again using \eqref{diff.N.3} at $s=S''$ finally gives \eqref{r-r_-EB}, given that $|r_{RN}(\s)-r_-| \lesssim \epsilon^2$.

							\color{black}
							From there on, all the other claimed estimates can be retrieved quite easily. 
						\end{proof}

						\subsection{Estimates on the late blue-shift  region}\label{LBsection}
						
						We define the late blue-shift region $\mathcal{LB}:=\{ \s \leq s \leq  \color{black} \delta_{\mathcal{C}}\color{black}(M,e,\Lambda,m^2) \cdot \epsilon^{-1} \color{black}\}$ \color{black} with $\delta_{\mathcal{C}}\color{black}(M,e,\Lambda,m^2)>0$ small  to be determined later. We have the following estimates:

						\begin{prop} \label{LB.prop}
							Recalling the notations from Proposition \ref{EB.Prop}, \color{black} defining $b_-= \frac{B}{2|K_-|} \neq 0$ and \color{black}   \\$C_{eff}= \color{black}e^{\color{black}b_-^{-2}\cdot [\frac{C_- \cdot \nu^{-1}}{|K_-| \cdot r_-}  +C_{EB}]\color{black}}\cdot  C_-   \color{black}$, there exists $D_{LB}(M,e,\Lambda,m^2)>0$ such that for all $ s \in \mathcal{LB}$: 
							\begin{equation} \label{LB1}
								|r(s)-[1-\frac{B^2}{|K_-|} \ep^2 s]r_-(M,e,\Lambda)| \leq D_{LB} \cdot \epsilon^2 ,
							\end{equation}	\color{black}  \begin{equation}\label{LB/C.radius}	r(s) < r_-(M,e,\Lambda,m^2),\end{equation}\color{black}
							\begin{equation} \label{LB2}
								|-\dot{r}(s)- \frac{B^2 \cdot r_-(M,e,\Lambda)\color{black}}{2|K_-|} \cdot \epsilon^2 -\color{black} \Omega^2(s)  \color{black}| \leq  D_{LB} \cdot \color{black} \epsilon^4 \cdot s \color{black},
							\end{equation} 
							\begin{equation} \label{LB2.5}
								|	\frac{-\dot{r}(s)}{r(s) }- \frac{B^2}{2|K_-|} \cdot \epsilon^2 -\color{black} \frac{\Omega^2(s) }{r_-} \color{black}| +
								|	-r(s)  \cdot \dot{r}(s)  - \frac{B^2 \cdot r_-^2}{2|K_-|} \cdot \epsilon^2- \color{black} r_- \cdot \Omega^2(s)  \color{black} | \leq  D_{LB} \cdot \color{black} \epsilon^4 \cdot s \color{black},
							\end{equation} 
							\begin{equation} \label{LB3}
								|\frac{d}{ds} \log(\Omega^2)(s)-2K_-| \leq  D_{LB} \cdot    \color{black} \epsilon^2\cdot  s \color{black},
							\end{equation}
							\begin{equation} \label{LB3.5}
								\bigl |\frac{d}{ds}\log \left( r(s)^{-b_-^{-2} \cdot \ep^{-2}} \cdot \Omega^2(s)\right)\color{black}-\frac{b_-^{-2}\cdot \ep^{-2}}{ r_-} \cdot \Omega^2(s) \color{black} \bigr| \leq D_{LB} \cdot  \color{black} \epsilon^2 \cdot s,
							\end{equation}
							\begin{equation} \label{LB4}
								\bigl |\log(\frac{\Omega^2(s)}{C_- \cdot  e^{2K_- s}} ) \bigr| \leq D_{LB} \cdot\color{black} \ep^2 \cdot s^2\color{black},
							\end{equation}
							\begin{equation} \label{LB4.5}
								\bigl |\log\left([\frac{r(s)}{r_-}]^{-b_-^{-2} \cdot \ep^{-2}}\frac{\Omega^2(s)}{\color{black}C_{eff} \color{black}} \right) \bigr| \leq D_{LB}\left[ \color{black} \epsilon^2 \cdot s^2 \color{black} + \ep^{-2} \cdot \Omega^2(s)\right], \color{black}
							\end{equation}\begin{equation}\label{LB5}
								\Omega^{0.01}|\phi|(s)+|\dot{ \phi}|(s) \leq  D_{LB} \cdot \epsilon,
							\end{equation}
							\begin{equation}\label{LB6}
								|r^2(s)\cdot \dot{\phi}(s)-r_-^2 \cdot B \cdot \epsilon| \leq D_{LB} \cdot  \epsilon^3,
							\end{equation}	\begin{equation}\label{LB7}
								| \phi(s) -A \cdot \epsilon - B \cdot \epsilon \cdot s| \leq D_{LB} \cdot  \epsilon^3 \cdot s^2.	\end{equation}

						\end{prop}

						\begin{proof}

							We make the following bootstrap assumptions, introducing $B_3=B_3(M,e,\Lambda,m^2)$, $B_4=B_4(M,e,\Lambda,m^2)$ (to be chosen later) and recalling $D_E=D_E(M,e,\Lambda,m^2)>0$: \begin{equation} \label{BLB1}
								|\frac{d \log(\Omega^2)}{ds}-2K_-| \leq |K_-|,
							\end{equation} \begin{equation} \label{BLB2}
								|\phi|(s) \leq   10 D_E\color{black}\cdot  s,
							\end{equation}  \begin{equation} \label{BLB3}
								e^{\frac{K_-}{10}(s-\s)}\cdot	|\delta\phi|(s) \leq 10 D_E\color{black} \cdot \epsilon^2 \cdot s,
							\end{equation}  \begin{equation} \label{BLB4}
								|\dot{r}|(s) \leq B_3 \cdot \epsilon^2,
							\end{equation}  \begin{equation} \label{BLB5}
								|\delta \Omega^2|(s) \leq B_4 \cdot \Omega^2 \cdot \color{black}\epsilon^2 \cdot s^2 \color{black}.
							\end{equation}  Note that by the estimates of Proposition \ref{EB.Prop} and \eqref{BLB1}, \eqref{BLB2}, \eqref{BLB3}, \eqref{BLB4}, \eqref{BLB5} are satisfied in a neighborhood of  $s=\s$ (for $B_3>2  C_- \cdot \nu$). As a consequence of \eqref{BLB1}, we have $$ \Omega^2(s) \lesssim \epsilon^2 \cdot e^{K_- \cdot (s-\s)},$$ and integrating \eqref{BLB4} using \eqref{r-r_-EB} at $s=\s$ we have $$ |r-r_-(M,e,\Lambda)| \leq D_E\color{black} \cdot  \epsilon^2 \cdot \s(\ep) + B_3 \cdot \epsilon^2 \cdot s \ls \color{black} \epsilon \color{black} < \frac{r_-}{100}. $$ Thus, integrating \eqref{FieldODE2} using \eqref{EB7} we get : $$ |\dot{\phi}|(s) \ls r^2 |\dot{\phi}|(s)   \lesssim \epsilon + \epsilon^2 \cdot \s \ls  \epsilon + \epsilon^2 \cdot \log(\epsilon^{-1}) \ls \ep,$$ which improves bootstrap \eqref{BLB2} for small enough $\epsilon$, after a second integration \color{black} and gives $$ |\phi|(s) \lesssim \ep \cdot s .$$ \color{black}
							
							From \eqref{BLB4}, notice that $|\delta \dot{r}| \lesssim \epsilon^2$ (trivial estimate) hence using \eqref{EB4} at $s=\s$ and since $s\geq \s \gtrsim \log(\epsilon^{-1})$: $$ |\delta r| \lesssim \epsilon^2 \cdot[ s+ \log(\epsilon^{-1})]\ls \ep^2 \cdot s.$$Plugging these estimates and the bootstraps (notably \eqref{BLB5}) in \eqref{FieldODE2} gives: $$|\frac{d}{ds}( \delta(r^2 \dot{\phi}))| \lesssim \epsilon^5 \cdot e^{K_- \cdot (s-\s)} \cdot s^2+  \color{black}\epsilon^5 \cdot e^{K_- \cdot (s-\s)} \cdot s^3\color{black} +\epsilon^4 \cdot e^{\frac{9K_-}{10} \cdot (s-\s)} \cdot s,$$ which we integrate to get using \eqref{EB1} at $s=\s$:  \begin{equation} \label{LB.est-1}
								|\delta(r^2 \dot{\phi})|(s) \lesssim \ep^3.
							\end{equation}   which also implies $$ |\delta( \dot{\phi})|(s) \lesssim \ep^3 \cdot s.$$  Integrating, we get, also using \eqref{EB1} at $s=\s$ $$|\delta\phi|(s) \lesssim \ep^3 \cdot  s^2  ,$$ which is sufficient to improve boostrap \eqref{BLB3} for small enough $\epsilon$, in view of the bound $e^{\frac{K_-}{10}(s-\s)} s^2 \ls \log^2(\epsilon^{-1})$.

							For  bootstrap \eqref{BLB1}, we take difference into \eqref{Omegastat2} and get, using the previously proven estimates \begin{equation} \label{LB.est0}
								|\frac{d^2}{ds^2} (\delta \log(r\Omega^2))| \ls \Omega^2(s) \epsilon^2 s^2 + \epsilon^2 \color{black},
							\end{equation} thus, integrating and using \eqref{EB11} and
							the fact that $\int_{\s}^{s} (s')^a \cdot \Omega^2(s') ds' \ls \epsilon^2 \cdot \s^a$ for all $a>0$ and $s \geq \s \gtrsim \log(\epsilon^{-1})$: $$ |\frac{d}{ds} (\delta \log(r\Omega^2))| \ls \epsilon^2 \cdot s \color{black} \ls \color{black}\ep \color{black}< \frac{|K_-|}{100} ,$$
							Using (trivially) the bound on $\dot{r}$ and \eqref{RNasymp2}, we get \begin{equation} \label{LB.est1}
								|\frac{d}{ds} \log(\Omega^2)(s)-2K_- | \ls  \epsilon^2 \cdot \color{black} s \color{black} < \frac{|K_-|}{50},
							\end{equation} and clearly \eqref{LB.est1} improves \eqref{BLB1}. Another integration of   \eqref{LB.est1} \color{black} also improves \eqref{BLB5} using $|\delta \Omega^2| \ls \Omega^2 \cdot |\delta \log(\Omega^2)|$ \color{black}(note that we used the fact that $\epsilon \cdot s \leq \delta_{\mathcal{C}}\color{black}(M,e,\Lambda,m^2)$ is bounded, hence $|\delta \log(\Omega^2)|$ is bounded by a constant depending only on $M$, $e$, $\Lambda$ and $m^2$)\color{black}.

							For bootstrap \eqref{BLB4}, notice that by \eqref{Raychstat} and \color{black} \eqref{N2} at $s=S''$ \color{black} we have for all $s$: $$ |\kappa^{-1}(s) - \int^{s}_{\color{black}S''\color{black}} \frac{r |\dot{\phi}|^2(s')}{\Omega^2(s')}ds'-1| \color{black} \ls \ep^2\color{black}.$$
							
							Now evaluate, \color{black} using also the estimates of Proposition~\ref{EB.Prop} \color{black} $$   |\int^{s}_{\color{black}S''\color{black}} \frac{r |\dot{\phi}|^2(s')}{\Omega^2(s')}ds'- \int^{s}_{\color{black}S''\color{black}} \frac{r |\dot{\phil}|^2(s')}{\Omega^2(s')}ds' | \lesssim \int^{s}_{\color{black}S''\color{black}} \frac{r |\delta\dot{ \phi}|(|\dot{\phil}|+ |\dot{\phi}|)(s')}{\Omega^2(s')}ds' \lesssim \epsilon^4 \int^{s}_{\color{black}S''\color{black}}  s' \Omega^{-2}(s') ds' \lesssim \epsilon^4  \cdot s\cdot  \Omega^{-2}(s),$$ whence for some $D'_{LB}(M,e,\Lambda,m^2)>0$: $$ | |\dot{r}| - \Omega^2(s) \cdot \int^{s}_{\color{black} S'' \color{black}} \frac{r |\dot{\phil}|^2(s')}{\Omega^2(s')}ds'-\Omega^2(s)| \leq  D'_{LB}\cdot \left[ \epsilon^4 \cdot s + \color{black} \ep^2 \cdot \Omega^2(s) \color{black} \right].$$
							Therefore, using \eqref{asymplin} and choosing $\ep$ small enough we get \begin{equation} \label{LB.est2}
								| |\dot{r}|(s) - B^2 \cdot \epsilon^2 \cdot \Omega^2(s) \cdot \int^{s}_{\color{black} S'' \color{black}} \frac{r(s')}{\Omega^2(s')}ds'- \color{black} \Omega^2(s) \color{black}| \leq   D'_{LB}\cdot \left[ \epsilon^4 \cdot s + \color{black} \ep^2 \cdot \Omega^2(s) \color{black} \right].
							\end{equation}
							
							Now write the following identity \begin{align}  \label{LB.est3}
								&\int_{\s}^{s} \frac{r(s')}{\Omega^2(s')}ds'=  \int_{\s}^{s} \frac{\frac{d}{ds}(\frac{r(s')}{\Omega^2(s')})}{-\frac{d}{ds}(\log(r^{-1} \Omega^2))}ds'\\ =   &\frac{r(s)}{2|K_-| \cdot\Omega^2(s)}- \frac{r(\s)}{2|K_-| \cdot\Omega^2(\s)} +\int_{\s}^{s}  \frac{d}{ds}(\frac{r(s')}{\Omega^2(s')}) \cdot (\frac{1}{-\frac{d}{ds}(\log(r^{-1} \Omega^2))}- \frac{1}{|2K_-|})ds'. \nonumber
							\end{align}
							
							Next we use \eqref{LB.est1} to  estimate the last term \begin{equation} \label{LB.est4}
								\bigl|\int_{\s}^{s} \frac{\frac{d}{ds}(\frac{r(s')}{\Omega^2(s')})}{-\frac{d}{ds}(\log(r^{-1} \Omega^2))+2K_-}ds' \bigr| \ls \epsilon^2 \cdot (\color{black}s\color{black}+\epsilon^2 \cdot s^2) \cdot \frac{r(s)}{\Omega^2(s)} \ls \ep^2  \cdot \color{black}s\color{black} \cdot  \Omega^{-2}(s) .
							\end{equation} 
							Combining \color{black} \eqref{EB12}, \color{black} \eqref{LB.est2}, \eqref{LB.est3} and \eqref{LB.est4} shows that for some  $D''_{LB}(M,e,\Lambda,m^2)>0$ and assuming   small enough $\epsilon$: \begin{equation}
								| |\dot{r}|(s) - B^2 \cdot \epsilon^2 \frac{r(s)}{2|K_-|}  - \color{black} \Omega^2(s) \color{black}| \lesssim [ \color{black}\ep^2 \cdot \color{black}\Omega^2(s) +   \epsilon^4 s] \color{black} \lesssim \ep^4 \cdot s \color{black}.
							\end{equation}
							
							Thus, \eqref{BLB4} is improved if we choose say $B_3 > \max \{2 C_- \cdot \nu, 2 B^2 \cdot  \frac{r_-}{2|K_-|} \}$ large enough. \color{black} Note in particular
							\begin{equation}\label{LB.int}
								|  (-2K_-) - b_-^{-2} \cdot \ep^{-2}\ \frac{|\dot{r}(s)|}{r(s)}+ \frac{ b_-^{-2} \cdot \ep^{-2}}{r_-}\ \Omega^2(s) | \lesssim \ep^2 \cdot s.
							\end{equation}		 \color{black} To obtain \eqref{LB4.5}, note in particular that combining \eqref{LB.est1} and \eqref{LB.int}  and integrating gives \begin{equation*}
								\bigl |\log\left(\left[\frac{r(s)}{\color{black}[1-\frac{B^2}{2|K_-|} \epsilon^2 \s(\epsilon)]r_-}\right]^{-b_-^{-2} \cdot \ep^{-2}}\frac{\Omega^2(s)}{e^{b_-^{-2} \cdot C_{EB}}\cdot  C_-   \cdot \nu^{-1}\cdot \ep^2} \right) -  \color{black} \frac{ b_-^{-2} \cdot \ep^{-2}}{r_-} \int_{\s}^{s} \Omega^2(s) ds \color{black} \bigr| \lesssim   \epsilon^2\cdot  s^2,
							\end{equation*} \color{black}which then gives  \begin{equation}\label{extraLB}
								\bigl |\log\left([\frac{r(s)}{\color{black}[1-\frac{B^2}{2|K_-|} \epsilon^2 \s(\epsilon)]r_-}]^{-b_-^{-2} \cdot \ep^{-2}}\frac{\Omega^2(s)}{e^{b_-^{-2} \cdot C_{EB}}\cdot  C_-    \cdot \nu^{-1}\cdot \ep^2} \right) -  \color{black} \frac{ b_-^{-2} \cdot \ep^{-2}}{r_- \cdot 2 |K_-|}\ \Omega^2(\s)  \color{black} \bigr| \lesssim   \epsilon^2\cdot  s^2 \color{black}+ \ep^{-2} \cdot \Omega^2(s).
							\end{equation}
							
							\color{black}
							Note also that, using the Taylor expansion of $\log$, we obtain (recalling the definition $\s(\epsilon)= 2|K_-|^{-1}\log(\nu \epsilon^{-2})$): \begin{equation} \label{extraLB2}
								\bigl| b_{-}^{-2} \epsilon^{-2} \log(1-\frac{B^2}{2|K_-|} \epsilon^2 \s(\epsilon)) + 2|K_-| \cdot \s(\epsilon) \bigr| \lesssim  \epsilon^2 [\log(\epsilon^{-1})]^2,
							\end{equation} from which we get \eqref{LB4.5}\color{black}, noting that $\frac{ b_-^{-2} \cdot \ep^{-2}}{r_- \cdot 2 |K_-|}\ \Omega^2(\s) =  \frac{ C_- \cdot b_-^{-2} \cdot \nu^{-1}}{r_- \cdot 2 |K_-|} + O (\ep^2 \cdot \log(\ep^{-1}))$ \color{black}.

							\color{black}All the other claimed bounds follow (or were already proven in the course of the proof).
							
							\color{black}

						\end{proof}

						\color{black}
						
						\subsection{Estimates on the crushing region}\label{crushing.section}
						
						We define the crushing region as $ \mathcal{C}:= \{  \color{black} 4\s\color{black} \leq s < s_{\infty}\}$ \color{black} region. \color{black}  Note that for $\epsilon$ small enough, $ \delta_{\mathcal{C}}\color{black}\cdot\epsilon^{-1}> 4\s(\ep) \approx \log(\ep^{-1})$ hence $\mathcal{C}$ and $\mathcal{LB}$ overlap. We will thus only apply Proposition \ref{LB.prop} to obtain estimates for  $s=4\s$, and prove new estimates for $s> 4\s$. We proceed in this way to obtain ``unified'' estimates in $\mathcal{C}$ which we will use in Section~\ref{Kasner.section} to show the uniform proximity of $g_{\ep}$ to a Kasner metric in $\mathcal{C}$.\color{black}

						\begin{prop} \label{C.prop} There exists $D_{C}(M,e,\Lambda,m^2)>0$ such that for all $s\in \mathcal{C}$: \begin{equation} \label{C1} 
								\bigl| \phi(s) - b_-^{-1} \cdot \ep^{-1} \cdot \log(\frac{r_-}{r})\bigr| \leq D_C \cdot \ep \cdot \log(\ep^{-1})\color{black} [1+  \cdot \color{black}|\log(\frac{r_-}{r(s)})|],
							\end{equation} \begin{equation} \label{C2}
								\bigl|r^2 \dot{\phi}(s) - B \cdot r_-^2 \cdot \ep \bigr| \leq D_C \cdot \ep^3,
							\end{equation}
							\begin{equation} \label{C3}
								\bigl|r^2(s) -r_-^2 + \frac{B^2 \cdot r_-^2}{|K_-|}\cdot \ep^2 \cdot s \bigr| \leq D_C \cdot \color{black}\ep^2 \cdot \log(\ep^{-1})\color{black} ,
							\end{equation} \begin{equation} \label{C4}
								\bigl| r(s)\cdot |\dot{r}|(s) -  \frac{B^2\cdot r_-^2 \cdot \epsilon^2}{2|K_-|}\bigr|= \bigl| r(s)\cdot |\dot{r}|(s) - 2|K_-| \cdot b_-^2  \cdot r_-^2 \cdot \epsilon^2\bigr| \leq D_C \cdot \color{black}\ep^4 \cdot \log(\ep^{-1})\color{black} ,
							\end{equation} \begin{equation} \label{C5}
								\bigl|\frac{d}{ds}    \log(r^{ -b_-^{-2} \cdot \ep^{-2}} \Omega^2(s)) \bigr| \leq    \frac{  D_C \cdot \color{black}\ep^2 \cdot \log(\ep^{-1}) \color{black}} {r^2},
							\end{equation}
							\begin{equation} \label{C7}
								|\color{black}\frac{r^2}{r_-^2}\color{black}\frac{d}{ds} \log(\Omega^2(s)) - 2K_-| \color{black}\leq D_C \cdot \color{black}\ep^2 \cdot \log(\ep^{-1}),\color{black}
							\end{equation}

							Moreover, recalling the notations from Proposition \ref{EB.Prop} and recalling $\color{black}C_{eff}:= e^{-b_-^{-2} \cdot C_{EB}} \cdot  C_-  \color{black}>0$, we have	\begin{equation} \label{C6}
								\bigl |\log\left([\frac{r(s)}{r_-}]^{-b_-^{-2} \cdot \ep^{-2}}\frac{\Omega^2(s)}{\color{black}C_{eff} \color{black}} \right) \bigr|  \leq    D_C \cdot \left[\color{black} \ep^2 \cdot \log(\epsilon^{-1})^2\color{black}+\color{black}\log(\ep^{-1})\color{black} \cdot 	|\log(\frac{r_-}{r(s)})|\right].	\end{equation}

						\end{prop}

						\begin{proof}
						The key point (and the main difference with $\mathcal{LB}$), which we will use repetitively when integrating, is that $\Omega^2(4\s) \sim \epsilon^8$ (by \eqref{LB4}) \color{black} so that the term $\Omega^2(s)$ on the RHS of \eqref{LB2}, \eqref{LB2.5} \color{black} or the $\ep^{-2} \cdot \Omega^{-2}$ in \eqref{LB3.5}, \eqref{LB4.5}  can be ignored for $s=4\s$.

							We bootstrap the following estimates for $C(M,e,\Lambda,m^2)>0$,  $B_5(M,e,\Lambda,m^2)>0$ to be determined later: \begin{equation} \label{B1TR}
								\color{black}	|\phi|(s) \leq \epsilon^{-1.1} \cdot\left[ |\log(\frac{r_-}{r})|+1 \right],
							\end{equation} \begin{equation} \label{B2TR}
								\color{black}		\Omega^2(s) \leq   e^{1.9K_- s}\cdot (\frac{r}{r_-})^{\ep^{-\frac{1}{2}}} .\color{black}
							\end{equation} 
							
						Note that \eqref{B1TR} and \eqref{B2TR} hold initially in a neighborhood of $4\s$ by the previous propositions. Note also that $e^{1.9K_- s} \ls \ep^{7.8}$ in this region hence \eqref{B2TR} also provides smallness in the subsequent estimates. \color{black}	
							Integrating \eqref{FieldODE2} using \eqref{B1TR} and \eqref{B2TR} with \eqref{LB6} for \color{black}$s= 4\s$:
							
							\begin{equation}\label{C.est1}
								|r^2 \dot{\phi }(s) - B \cdot r_-^2 \cdot \epsilon| \ls  \ep^3+ \int_{4\s}^{s} r^2 \Omega^2 |\phi| ds' \ls \ep^3+\ep^{-1.1} \int_{4\s}^{s} e^{2K_- s'} ds'   \color{black}\ls \ep^3.
							\end{equation} 
							
							Then integrate \eqref{req} using \eqref{B1TR} and \eqref{B2TR} with \eqref{LB2.5} for \color{black}$s=  4\s$ \color{black} gives:
							
							\begin{equation}\label{C.est2}
								| -r(s)\cdot \dot{r}(s) -  \frac{B^2\cdot r_-^2 \cdot \epsilon^2}{2|K_-|}| \lesssim  \color{black} \ep^4 \cdot \log(\ep^{-1})\color{black}.
							\end{equation} 
							
							Integrating \eqref{C.est2} and using \eqref{LB1} for \color{black}$s=  4\s$ \color{black} we obtain the following 	\begin{equation}\label{C.est2.5}
								\bigl| \frac{B^2}{|K_-|} \cdot \ep^2 \cdot s +(\frac{r(s)}{r_-})^2-1 \bigr| \ls \color{black} \ep^4 \cdot \log(\ep^{-1})\color{black}\cdot s,
							\end{equation}  in particular for $\ep$ small enough 	\begin{equation}\label{C.est2.75}
								\frac{B^2}{|K_-|} \cdot s\leq 1.1 \cdot \ep^{-2}. \end{equation}

							Moreover, as a consequence of \eqref{C.est1} and \eqref{C.est2} we obtain (recalling that $b_-:= \frac{B}{2|K_-|}$): \begin{equation} \label{C.est3}
								\bigl| \dot{\phi} - b_-^{-1} \cdot \epsilon^{-1} \cdot \frac{|\dot{r}|}{r} \bigr | \ls  \ep\log(\ep^{-1})\color{black}\frac{|\dot{r}|}{r} \ls \frac{\color{black}\ep^3\log(\ep^{-1})\color{black}}{r^2}.
							\end{equation} 
							
							\color{black} Now we can integrate it on $[4\s,s]$ using \eqref{LB1} and \eqref{LB7}, which improves bootstrap \eqref{B1TR} and gives \begin{equation} \label{C.est4}
								\bigl| \phi(s) - b_-^{-1} \cdot \epsilon^{-1}\cdot  \log(\frac{r_-}{r(s)}) \bigr| \ls \ep \cdot \log(\ep^{-1})+\ep \log(\ep^{-1})\color{black} \cdot |\log(\frac{r_-}{r(s)})|.
							\end{equation}

							Then from \eqref{Omegastat}, \eqref{C.est1} and using \eqref{B2TR} we get $$ | \frac{d^2 \log(r \Omega^2)}{ds^2} +2 B^2 \cdot \ep^2 \cdot \frac{ r_-^4}{r^4}| \lesssim  \frac{\ep^4}{r^4} +  \ep^{7.5}\color{black} \cdot (\frac{r}{r_-})\color{black}^{-5 + \ep^{-\frac{1}{2}}} \ls \frac{\ep^4}{r^4}, $$
							which we re-write using \eqref{C.est2} as \begin{equation} \label{C.est5}
								\bigl|\frac{d^2 \log(r \Omega^2)}{ds^2} +4 |K_-|\cdot r_-^2 \cdot\frac{ |\dot{r}|}{r^3}|   \ls \frac{\ep^2  \log(\ep^{-1})\color{black}  \cdot |\dot{r}|}{r^3}.
							\end{equation}
							
							We then \color{black}	integrate using \eqref{LB1}, \eqref{LB2}, \eqref{LB2.5} and \eqref{LB3}  for $s=  4\s$ \color{black} (recall that $\dot{r}<0$) we obtain  
							\begin{equation}  \label{C.est6}\bigl|\frac{d \log( \Omega^2)}{ds} -(2K_-) \cdot \frac{r_-^2}{r^2}\bigr|   \ls \bigl|\frac{d \log(r \Omega^2)}{ds} -(2K_-) \cdot \frac{r_-^2}{r^2}\bigr|  + \frac{\ep^2  \log(\ep^{-1})\color{black}  }{r^2} \ls   \frac{\color{black}\ep^2 \cdot \log(\ep^{-1}) \color{black}}{r^2}. \end{equation}
							
							In particular, this shows that $\frac{d\log( \Omega^2)}{ds}\leq 1.99\color{black}K_- $ so by \eqref{LB4}
							\begin{equation} \label{C.est.int.new}
								\color{black}\Omega^{2}(s) \leq  \Omega^{2}(4\s) \cdot e^{1.99\color{black}K_- \cdot (s-4\s)} \leq 2 C_- \cdot \nu^{-2} \cdot \ep^8.
							\end{equation}
							
							Now we can re-write \eqref{C.est6} using \eqref{C.est2} as: there exists a constant $E_C(M,e,\Lambda,m^2)>0$ such that for all $s\in \mathcal{C}$:
							\begin{equation}  \label{C.est6.5}
								\bigl|\frac{d \log( \Omega^2)}{ds} + (b_-)^{-2}\cdot \ep^{-2}\cdot \frac{ |\dot{r}|}{ r} \bigr|= \bigl|\frac{d}{ds}  \left(  \log(r^{- b_-^{-2}\ep^{-2}} \Omega^2 )\right) \bigr| \leq   \frac{E_C \cdot \color{black}\log(\ep^{-1}) \color{black} \cdot|\dot{r}|}{r}.
							\end{equation}
							
							Therefore, integrating \eqref{C.est6.5} on $[4\s,s]$ using \eqref{LB4.5} at $s= 4\s$ and also \eqref{r-r_-EB}, we obtain 	\begin{equation}  \label{C.est7}
								|\log\left((\frac{r(s)}{r_-})^{- b_-^{-2}\ep^{-2}} \frac{\Omega^2(s)}{\color{black}C_{eff} \color{black}} \right)|\ls  \color{black} \ep^2 \cdot \log(\epsilon^{-1})^2\color{black}+  \color{black}\log(\ep^{-1}) \color{black}  \cdot | \log(\frac{r_-}{r}) |.
							\end{equation}
							
							In particular, for $\ep$ small enough and taking the exponential (recall \eqref{CEB.def} and the definition of $C_{eff}$): \begin{equation} \label{C.est.int.new2}
								\Omega^{2}(s) \color{black}\ls \color{black} (\frac{r(s)}{r_-})^{ \frac{b_-^{-2}}{2}   \cdot \ep^{-2}}.
							\end{equation}

							Combining \eqref{C.est.int.new} and \eqref{C.est.int.new2} (for instance, take  \eqref{C.est.int.new} \color{black} to the power $0.999$ and \eqref{C.est.int.new2} to the power $0.001$ \color{black} and multiply them) improves bootstrap \eqref{B2TR} \color{black} for $\ep$ small enough \color{black}; all the other claimed estimates follow from the proof.

						\end{proof}
						
						\subsection{Spacelike singularity and blow-up estimates}

						\begin{prop} \label{spacelike.prop}There exists $s_{\infty}(\ep)>0$ such that $$ \lim_{s \rightarrow s_{\infty}} r(s)=0,$$ and moreover there exists  $D_C'(M,e,\Lambda,m^2)>0$, $D_C''(M,e,\Lambda,m^2)>0$ such that \begin{equation}\label{sinfty}
								|s_{\infty}-\epsilon^{-2} \cdot \frac{|K_-|}{B^2}| \leq D_C \cdot \color{black}\log(\ep^{-1}) \color{black} ,
							\end{equation}  \begin{equation}\label{rs}
								| (s_{\infty}-s) - \epsilon^{-2}\cdot\frac{|K_-|}{B^2}\cdot  \frac{r^2(s)}{r_-^2} |	\leq D'_C \cdot \color{black}\log(\ep^{-1}) \color{black} \cdot r^2(s). \color{black}
							\end{equation}  
							
							Additionally, $\{s=s_{\infty}\}$ is a spacelike singularity in the sense that for all $ p \in \{s=s_{\infty}\}$,  $J^{-}(p) \cap \mathcal{H}^+$ is compact, where $J^{-}(p)$ is the domain of dependence of $p$.

							Moreover, there exists $\tilde{D}_C(M,e,\Lambda,m^2)>0$ such  that the following blow-up estimates hold: for all $ s\in \mathcal{C}$  \begin{equation} \label{rho.blowup.equation}
								\tilde{D}_C^{-1} \cdot \ep^4\cdot  (\frac{r(s)}{r_-})^{-b_-^{-2} \cdot \ep^{-2}\cdot (1+\color{black}\ep^2 \cdot \log(\ep^{-1}) \color{black} \cdot  \tilde{D}_C)} \leq \rho(s) \leq \tilde{D}_C \cdot \ep^4 \cdot  (\frac{r(s)}{r_-})^{-b_-^{-2}\cdot \ep^{-2} \cdot (1-\color{black}\ep^2 \cdot\log(\ep^{-1}) \color{black} \cdot  \tilde{D}_C) },
							\end{equation}
							\begin{equation} \label{K.blowup.equation}
								\tilde{D}_C^{-2} \cdot \ep^4 \cdot  (\frac{r(s)}{r_-})^{-2b_-^{-2}\cdot \ep^{-2}\cdot (1+2  \color{black} \ep^2 \cdot \log(\ep^{-1}) \color{black} \cdot  \tilde{D}_C)} \leq \mathfrak{K}(s) \leq \tilde{D}_C^2 \cdot \ep^4 \cdot  (\frac{r(s)}{r_-})^{-2b_-^{-2}\cdot \ep^{-2} \cdot (1-2 \color{black} \ep^2 \cdot \log(\ep^{-1}) \color{black} \cdot  \tilde{D}_C) },
							\end{equation}
							
							where $\rho$ is the Hawking mass and $\mathfrak{K}(s):=  R_{\alpha \beta \mu \nu} R^{\alpha \beta \mu \nu}(s)$ the Kretchsmann scalar. In particular, for all $p>1$: \begin{equation} \label{rho.blow.up}
								\sup_{s\in \mathcal{C}} \rho(s)= +\infty, \hskip 6 mm \int_{\mathcal{C} \cap  \{t\in [0,1]\}} \rho^{p}   dvol_g:= 4\pi  \int_{\ep^{-1}}^{s_{\infty}} \rho^{p}(s')  \cdot  r(s') \cdot \Omega^2(s') ds' = +\infty.
							\end{equation}
							
						\end{prop}
						\begin{proof}
							The existence of $s_{\infty}$, \eqref{sinfty}, \eqref{rs} follows directly from \eqref{C3}, \eqref{C4}. The spacelike character of $\mathcal{S}$ follows immediately. \eqref{rho.blowup.equation} (and by extension \eqref{rho.blow.up}, using also \eqref{C4}) follow from \eqref{C6} and \eqref{mu.static}. For a metric of the form \eqref{guv}, $\mathfrak{K}$ is given by (c.f.\ \cite{Kommemi}, Section 5.8): $$\mathfrak{K} = 16\Omega^{-4} [\frac{d^2 \log(\Omega^2)}{ds^2}]^2+ \frac{16 \rho^2}{r^6} + 24 \left( \frac{\rho}{r^3} + 2 \Omega^{-2} r  |\dot{\phi}|^2-m^2 |\phi|^2\color{black}\right)^2 + 48 \Omega^{-4} |\dot{\phi}|^4 $$

							In view of \eqref{mu.static} and \eqref{Omegastat} and Proposition \ref{C.prop}, we see that the $\Omega^{-4} |\dot{\phi}|^4$ term dominates, therefore for all $s\in \mathcal{C}$: $$ \mathfrak{K} (s) \simeq \frac{\ep^4}{r^8 (s) \cdot \Omega^4(s)} \simeq \frac{ \ep^{-4}\cdot\rho^2(s)}{r^6(s)},$$ and \eqref{K.blowup.equation} then follows from \eqref{rho.blowup.equation}.

						\end{proof}

						\section{Convergence to Reissner--Nordstr\"{o}m-(dS/AdS) in a weak topology} \label{convergence.section}
						
						In this section, we state, for the convenience of the reader, the convergence results in the $s$ coordinate defined by \eqref{gauge}, which follow immediately from the estimates of Section~\ref{main.section}. We refer the reader to Definition \ref{distrib.conv} for notations.
						
						\begin{prop}
							Recall the definition of $\tilde{\Omega^2}_{\ep}(s)$ from Definition \ref{distrib.conv}. Then  $\tilde{\Omega^2}_{\ep}(s)$ converges uniformly to $\Omega^2_{RN}(s)$. More precisely, we have the following estimate \begin{equation} \label{OmegaLinfty}
								\sup_{-\infty< s< +\infty} |\tilde{\Omega^2}_{\ep}(s)- \Omega^2_{RN}(s)| =\sup_{-\infty< s< s_{\infty}(\ep)} |\Omega^2(s)- \Omega^2_{RN}(s)| \ls \ep^2.
							\end{equation}
							
							Moreover,  $\tilde{\Omega^2}_{\ep}(s)$ converges in $L^1(\RR_s)$ to $\Omega^2_{RN}(s)$. More precisely, we have the following estimate \begin{equation} \label{OmegaL1}
								\int_{-\infty}^{+\infty} |\tilde{\Omega^2}_{\ep}(s)- \Omega^2_{RN}(s)| ds=\int_{-\infty}^{s_{\infty}(\ep)} |\Omega^2(s)- \Omega^2_{RN}(s)| ds  \ls \ep^2.
							\end{equation}
						\end{prop}
						
						\begin{prop}
							
							Recall the definition of $\tilde{r}_{\ep}(s)$ from Definition \ref{distrib.conv}.  Then for all $s \in \RR$, we have the point-wise convergence: $$  \tilde{r}_{\ep}(s) \underset{\ep \rightarrow 0}{\rightarrow} r_{RN}(s).$$
							
							Additionally, we have the following convergence in distribution: $$  \tilde{r}_{\ep} \overset{\mathcal{D}'(\RR_{s})}  \rightharpoonup r_{RN}.$$	
							Moreover, we have the following convergence in $L^{\infty}_{loc}$ in the following sense: for all $S \in \RR$, there exists $\epsilon_0(S)>0$ such that for all $0<|\ep| < \ep_0$:
							\begin{equation} \label{r-rRNLinfty}
								\sup_{-\infty<s \leq S} |r^2(s)-r_{RN}^2(s)| \leq \sup_{-\infty<s \leq \color{black}\ep^{-1}\color{black}} |r^2(s)-r_{RN}^2(s)| \leq C \cdot \color{black}\ep\color{black}.
							\end{equation}
							
							Nevertheless,  $\tilde{r}_{\ep}^2(s)$ does not converge to $r_{RN}(s)$ in $L^1(\RR_s)$ (or a forciori in any $L^p(\RR_s)$, $1\leq p<\infty$) in particular	\begin{equation} \label{r-rRNL1blowup}
								\int_{ s\in \RR } |r^2(s)-r_{RN}^2(s)| ds\geq 	\int_{\ep^{-1} \leq s < s_{\infty}(\ep) } |r^2(s)-r_{RN}^2(s)| ds \geq C \cdot \ep^{-2}.
							\end{equation}
							Also, $\tilde{r}_{\ep}(s)$ does not converge uniformly to $r_{RN}(s)$ either: for all $0\leq \alpha < 1$, there exists $s_{\alpha}(\ep) \in (\delta_{\mathcal{C}}\color{black} \cdot\ep^{-1},s_{\infty})$ with \begin{equation} \label{not.unif}
								\lim_{\ep \rightarrow 0} r(	s_{\alpha}(\ep))= \alpha \cdot r_-(M,e,\Lambda).
							\end{equation}
							
							Lastly, $\tilde{r}_{\ep}(s)$ does not converge to $r_{RN}(s)$ in $\dot{W}^{1,p}$ for any $p\geq1$ : we have the following estimate for all $s < s_{\infty}$:
							\begin{equation} \label{W11blowup}
								\int_{s}^{+\infty} \bigl| \frac{d\tilde{r}_{\ep}}{ds}(s')- \frac{dr_{RN}}{ds}(s') \bigr|  ds' \geq 	\int_{\ep^{-1}}^{s_{\infty}} \bigl| |\dot{r}|(s')-\Omega^2_{RN}(s') \bigr|  ds' \geq \frac{r_-}{2},
							\end{equation} \begin{equation} \label{H1blowup}
								\int_s^{s_{\infty}} (\dot{r})^2(s') ds' = +\infty.
							\end{equation}
						\end{prop}
						\begin{proof}
							All the estimates follow from \eqref{diffR2}, \eqref{diff.N.3} (for $s \leq \s$), \eqref{EB4}, \eqref{LB1} and \eqref{C3}. \\Note that 
							$s_{0}(\ep):= s_{\infty}-\ep^{-1}> \ep^{-1}$ for small enough $\ep$, by \eqref{sinfty}.
						\end{proof}
						
						\section{Convergence to a Kasner-type metric in re-normalized coordinates} \label{Kasner.section}

						In this section, we will define new variables to show that the metric is uniformly close to a Kasner metric of exponents $$
						(1-4b_-^2 \epsilon^2,\ 2b_-^2 \epsilon^2,\ 2b_-^2 \epsilon^2)+O_{s}(\epsilon^3)$$ in the crushing region $\mathcal{C}=\{\color{black}4\s \color{black}\leq s < s_{\infty}(\ep)\}$, and close in $BMO$ (in Kasner-type coordinates)  to Minkowski.

						\color{black}  To express the metric in Kasner-type coordinates, we define (a multiple of) the proper time $\tau(s)>0$: 
						\begin{equation}\label{taudef}
							\frac{d\tau}{ds} = - \tau_0(M,e,\Lambda,m^2,\ep) \cdot 
							\Omega(s),
						\end{equation}
						\begin{equation*} \label{Tdef2}
							\tau(s_{\infty})=0,
						\end{equation*} where $\tau_0(\ep)>0$ (to be estimated later, see already Lemma \ref{tau0.lemma}) is chosen so that on the past boundary on $\mathcal{C}$ (i.e.\ $s=4\s$), \color{black} as per \eqref{LB4} \color{black} \begin{equation}\label{tau0def}
							\tau(4\s) = \Omega(4\s)= \color{black} \sqrt{C_-} \color{black} \cdot \nu^{-2} \cdot \epsilon^4 + O(\epsilon^6 [\log(\epsilon^{-1})]^2).
						\end{equation}\color{black}
						We have the following immediate formulae for $\tau$:
						\begin{equation}\label{tau.formula1}
							\tau(s)= \tau_0 \cdot
							\int_{s}^{s_{\infty}} \Omega(s') ds',
						\end{equation}
						\begin{equation}\label{tau.formula2}
							\frac{d}{ds} \log(\tau^{-1})(s) =\frac{\Omega(s)}{\int_{s}^{s_{\infty}} \Omega(s') ds'}.
						\end{equation}

						It will also be convenient to introduce the  the re-scaled area-radius \begin{equation} \label{X.def}
							X(s):= \frac{r(s)}{r_-}. 
						\end{equation}\color{black}
						Note that for all $s \in \mathcal{C}$, $X(s) \in (0,1)$ (This is a direct consequence of the monotonicity of $r$ and \eqref{LB/C.radius}).
						
						Substituting $s$ for $\tau$ in \eqref{guv} gives a metric \color{black} of the following form: \begin{equation} \label{gTX}
							g = - \tau_0^{-2}
							d\tau ^2 + \Omega^2(s) dt^2 + r_-^2 \cdot X^2(s) \color{black} d\sigma_{\mathbb{S}^2} \color{black}.
						\end{equation}
						\color{black}
						Now the goal will be to show that for  a bounded function $P(\tau)$ we have for all $s \in \mathcal{C}$ \begin{align}
							&   \Omega^2(s) \approx  \tau^{2(1-2b_- \epsilon^2\cdot(1+\epsilon^2 \cdot \log(\ep^{-1}) \cdot P(\tau)))}, \label{crucial1}\\ & X^2(s) \approx (\frac{\tau}{\sqrt{C_{eff}}})^{4b_- \epsilon^2\cdot(1+\epsilon^2 \cdot \log(\ep^{-1}) \cdot P(\tau))}\label{crucial2}. 
						\end{align} 
						
						Once \eqref{crucial1} and \eqref{crucial2}, we (trivially) re-scale by constants the coordinates in the following fashion \begin{align}
							&\tilde{\tau}= \frac{\tau}{\tau_0}, \label{scale1}\\ & \rho = \tau_0^{1-4b_- \epsilon^2}\ t,  \label{scale2}\\ &   d\tilde{\sigma}_{\mathbb{S}^2} =  r_-^2 \cdot  (\frac{\tau_0}{\sqrt{C_{eff}}})^{4b_- \epsilon^2} d\sigma_{\mathbb{S}^2} \label{scale3},
						\end{align} and thus \eqref{gTX} becomes schematically a Kasner metric of exponents $(1-4b_- \epsilon^2,\ 2b_- \epsilon^2,\ 2b_- \epsilon^2)$ (up to $O(\epsilon^2 \cdot \log(\ep^{-1}))$ errors)  \begin{equation*}
							g \approx  -d \tilde{\tau}^2 + \tilde{\tau}^{2(1-4b_- \epsilon^2)} d\rho^2 +  \tau^{4b_- \epsilon^2} d\tilde{\sigma}_{\mathbb{S}^2}.
						\end{equation*}

						The key will be to estimate $X^2(s)$ in terms of $\Omega^2(s)$ in the following manner \begin{equation}\label{X.rough.estimate}
							|\frac{|K_-| \int_s^{s_{\infty}} \Omega(s')ds'}{\Omega(s) \cdot X^2(s)}-1|=| \frac{|K_-|\cdot \tau(s)}{\tau_0 \cdot \Omega(s) \cdot X^2(s)}-1| \lesssim \epsilon^2 \cdot \log(\ep^{-1}).
						\end{equation} 
						\color{black}
						\subsection{Preliminary estimates}\label{prel}
						In this section, we will give two types of estimates that are basically translation in the $\tau$, $X$ variables of the estimates of Section~\ref{LBsection} and Section~\ref{crushing.section}. 
						
						\subsubsection{Estimates on $\tau_0$ }\label{LBtau.section}
						
						In this section, we estimate the constant $\tau_0$ from \eqref{tau0def}. This will be useful both for the $L^{\infty}$ and $BMO$ estimates. Note that by our definition of $\tau$ in \eqref{taudef},\eqref{tau0def} we have $\log(\frac{\Omega(4\s)}{\tau(4\s)})=0$ on the past boundary of $\mathcal{C}$.

						\begin{lem} \label{tau0.lemma}
							\begin{equation} \label{tau0.est}
								\bigl|\tau_0(\ep) - |K_-|(M,e,\Lambda) \bigr| \ls \color{black}  \ep^2\cdot \log(\epsilon^{-1})^2\color{black}.
							\end{equation} In particular, $\underset{\ep \rightarrow 0}{\lim} \ \tau_0(\ep)= |K_-|$.
						\end{lem}
						
						\begin{proof}

							We start with the following estimate which follows immediately from \eqref{LB4}: \begin{equation}\label{int1}
								\Omega(4\s) = \sqrt{C_- }\cdot \nu^{-2} \cdot \ep^4 +O(\epsilon^{6} \log(\epsilon^{-1})^2)
							\end{equation}
							Now, by \eqref{LB4} and \eqref{C7} [which guarantees that $\Omega^2$ is decreasing in $\mathcal{C}$], it is easy to see that \begin{equation}\label{int2}
								\int_{4\s}^{s_{\infty}} \Omega(s) ds =  \int_{4\s}^{100\s} \Omega(s) ds + O( \epsilon^{99}) = \sqrt{C_- } \int_{4\s}^{100\s} e^{K_-s}ds + O(\epsilon^{4} \log(\epsilon^{-1})^2)= \frac{\sqrt{C_- }\cdot \nu^{-2}}{|K_-|}\cdot \epsilon^4+ O(\epsilon^{6} \log(\epsilon^{-1})^2).
							\end{equation}  \eqref{tau0.est} then follows immediately from \eqref{tau0def}, \eqref{int1} and \eqref{int2}. 
						\end{proof}
						\color{black}
						\subsubsection{Preliminary estimates on $\Omega^2$ and $X^2$ in $\mathcal{C}$}
						The purpose of this section is to translate the estimates of Section~\ref{crushing.section} into the new $\tau$, $X$ functions. This will be very useful in the proof of \eqref{X.rough.estimate}. One of the key quantities we introduce is $\log(\Delta_X)$, which is very small $O(\epsilon^3 \log(\ep^{-1}))$.
						\color{black}
						\begin{lem} \label{computation.Kasner.lemma}

							We define the variable $D(s):= \color{black}\ep^{-2} \color{black} \cdot  ( \frac{b_-^2 \ep^2\log(\frac{\Omega^2(s)}{\color{black}C_{eff} \color{black} })-\log(X(s))}{2 b_-^2 \color{black}\ep^2 \color{black} \cdot \log(\ep^{-1})^2 \color{black}+ \color{black}  \ep^2 \cdot \log(\ep^{-1}) \color{black}\log(X(s))})$ or equivalently \begin{equation} \begin{split} \label{Omega.formula}
									\Omega(s)= C_{eff}^{\frac{1}{2}}   \cdot e^{ \color{black}   \ep^2 \cdot \log(\ep^{-1})^2\cdot \color{black} D(s)} \cdot  X^{   \frac{b_-^{-2}\cdot \ep^{-2} }{2} \cdot(1+ \color{black} \ep^2 \cdot \log(\ep^{-1}) \color{black}\cdot D(s))}
								\end{split}
							\end{equation}
							Then there exists $D_0(M,e,\Lambda,m^2)>0$, $D'_0(M,e,\Lambda,m^2)>0$ such that for all $s\in \mathcal{C}$:
							\begin{equation} \label{D0bound}
								|D|(s) \leq D_0,
							\end{equation}
							\begin{equation} \label{D0'bound}
								\bigl|\frac{dD(X)}{dX}\bigr| \leq \frac{D'_0 }{X \cdot \log(X^{-1})}.
							\end{equation}
							
							Then, defining $\Delta_X(s):= \frac{\Omega^{2 b_-^2\ep^2}(s) }{\color{black} C_{eff}^{b_-^2 \epsilon^2}\cdot\color{black}[X(s)]^{1+\color{black}\ep^2 \cdot \log(\ep^{-1}) \color{black} \cdot D(s)} }$, we have  for all $s \in \mathcal{C}$: \begin{equation} \label{DeltaX.bound}
								\bigl|\log(\Delta_X(s))\bigr| \ls \color{black}\ep^4 \cdot \log(\ep^{-1})^2.\color{black} 
							\end{equation}
							
							Now, defining $F(s):= \frac{1}{\color{black}\ep^2 \cdot \log(\ep^{-1}) \color{black} }\cdot \left(\frac{ 2 |K_-| \cdot X(s)\cdot |\dot{X}(s)| }{B^2\cdot \ep^2}-1\right)$ or equivalently \begin{equation} \label{F0def}
								\frac{ 2 |K_-| \cdot X(s)\cdot |\dot{X}(s)| }{B^2 \cdot \ep^2} = 1+ F(s) \cdot \color{black}\ep^2 \cdot \log(\ep^{-1}) \color{black}.
							\end{equation} Then there exists $F_0(M,e,\Lambda,m^2)>0$ such that for all $s \in \mathcal{C}$: \begin{equation} \label{F0bound}
								| F|(s) \leq F_0.
							\end{equation}

						\end{lem}
						
						\begin{proof}
							\eqref{D0bound} is a mere-writing of \eqref{C6}. For \eqref{D0'bound}, we note that \eqref{C5} gives the equivalent (in view also of \eqref{C4}) $$ |D'(s)| \ls \frac{\ep^2}{X^2(s) \cdot |\log(X(s))|} \ls \frac{|\dot{X}(s)|}{ X(s) \cdot |\log(X(s))|}, $$ and \eqref{D0'bound} follows immediately by the chain rule. For \eqref{DeltaX.bound}, pass \eqref{Omega.formula} to the power $ 2 b_-^2 \cdot \ep^2$. \eqref{F0bound} is a re-writing of \eqref{C4}. 
						\end{proof}

						\color{black}

						\subsection{The proof of the key estimate \eqref{X.rough.estimate}}
						
						Now we turn to the proof of \eqref{X.rough.estimate}, which is probably the most delicate estimate linking $X$, $\Omega^2$ and $\tau$. We do the computation in two steps: in the first one we reduce the computation of $\frac{\tau} {\Omega}$ to a complicated integral $\mathcal{M}(s)$. 
						
						\color{black}
						
						\begin{lem} Recall the definition of the constant $F$ from \eqref{F0bound}, we have the following identities for all $s\in \mathcal{C}$:
							\begin{equation} \label{M.def.decomp}
								\frac{\int_{s}^{s_{\infty}}\Omega(s')ds'}{\Omega(s)} = \frac{2|K_-|}{B^2} \cdot \mathcal{M}(s)+ \mathfrak{E}_1(s)+ \mathfrak{E}_2(s),
							\end{equation}\begin{equation} \label{M.def}
								\mathcal{M}(s):= \ep^{-2} X^2(s) \int_{0}^1 [X(s)]^{\frac{b_-^{-2} \color{black}\log(\ep^{-1})\color{black}}{2}\cdot [D (Z X(s))- D ( X(s))]} Z^{1+\frac{b_-^{-2} \ep^{-2}}{2}\cdot (1+\color{black}\ep^2 \cdot \log(\ep^{-1}) \color{black} \cdot D(X(s) Z))}dZ,
							\end{equation}\begin{equation}\label{E1bound}
								| \mathfrak{E}_1|(s) \leq F_0 \cdot \color{black}\ep^2 \cdot \log(\ep^{-1}) \color{black} \cdot \frac{\int_{s}^{s_{\infty}}\Omega(s')ds'}{\Omega(s)},
							\end{equation}\begin{equation}\label{E2bound}
								| \mathfrak{E}_2|(s) \leq 4 D_0 \cdot \color{black}\ep^2 \cdot \log(\epsilon^{-1})^2\cdot  \color{black} \frac{\int_{s}^{s_{\infty}}\Omega(s')ds'}{\Omega(s)}.
							\end{equation}
						\end{lem}
						
						\begin{proof}
							First, we use \eqref{Omega.formula}  to write \begin{align*}
								\frac{\int_{s}^{s_{\infty}}\Omega(s')ds'}{\Omega(s)}  = \frac{ \int_{s}^{s_{\infty}}e^{ \color{black}  \ep^2 \cdot \log(\epsilon^{-1})^2 \cdot  \color{black}  D(s')} \cdot  [X(s')]^{   \frac{b_-^{-2}\cdot \ep^{-2} }{2} \cdot(1+\color{black}\ep^2 \cdot \log(\ep^{-1}) \color{black}  \cdot D(s'))}ds'}{e^{  \color{black}  \ep^2 \cdot \log(\epsilon^{-1})^2 \color{black}\cdot D(s)} \cdot  [X(s)]^{   \frac{b_-^{-2}\cdot \ep^{-2} }{2} \cdot(1+\color{black}\ep^2 \cdot \log(\ep^{-1}) \color{black} \cdot D(s))}}\end{align*} Then, in the integral in the numerator, we write  \eqref{F0def} as $1= \frac{2|K_-| X(s') |\dot{X}|(s')}{B^2 \ep^2}- F(s) \cdot \ep^2 \log(\ep^{-1})$, and $\frac{e^{   \ep^2 \cdot \log(\epsilon^{-1})^2 \cdot  D(s')}}{e^{   \ep^2 \cdot \log(\epsilon^{-1})^2 \cdot  D(s)}}= 1+ [e^{  \ep^2 \cdot \log(\epsilon^{-1})^2 \cdot  (D(s')-D(s))}-1]$. The first transformation will generate an error $\mathfrak{E}_1$, while the second transformation generates an error  $\mathfrak{E}_2 $, and we have \color{black} 
							
							\begin{align*}	  &\ 	 \frac{\int_{s}^{s_{\infty}}\Omega(s')ds'}{\Omega(s)}  =  \frac{2|K_-| \cdot  \ep^{-2}}{B^2}  \underbrace{\frac{ \int_{s}^{s_{\infty}}|\dot{X}|(s') [X(s')]^{ 1+  \frac{b_-^{-2}\cdot \ep^{-2} }{2} \cdot(1+\color{black}\ep^2 \cdot \log(\ep^{-1}) \color{black} \cdot D(s'))}ds'}{  [X(s)]^{   \frac{b_-^{-2}\cdot \ep^{-2} }{2} \cdot(1+\color{black}\ep^2 \cdot \log(\ep^{-1}) \color{black} \cdot D(s))}}}_{\mathcal{M}(s)} + \mathfrak{E}_1 + \mathfrak{E}_2, \\
								&  \mathfrak{E}_1(s):=-\color{black}\ep^2 \cdot \log(\ep^{-1}) \color{black} \cdot  \frac{ \int_{s}^{s_{\infty}} F(s')e^{ \color{black} \ep^2 \cdot \log(\epsilon^{-1})^2 \cdot  \color{black} D(s')} \cdot  [X(s')]^{   \frac{b_-^{-2}\cdot \ep^{-2} }{2} \cdot(1+\color{black}\ep^2 \cdot \log(\ep^{-1}) \color{black}   \color{black} \cdot D(s'))}ds'}{e^{   \color{black}  \ep^2 \cdot \log(\epsilon^{-1})^2   \color{black}\cdot D(s)} \cdot  [X(s)]^{   \frac{b_-^{-2}\cdot \ep^{-2} }{2} \cdot(1+\color{black}\ep^2 \cdot \log(\ep^{-1}) \color{black} \cdot D(s))}}; \\
								& \mathfrak{E}_2(s)= \frac{ \int_{s}^{s_{\infty}}[e^{ \color{black}  \ep^2 \log(\epsilon^{-1})^2   \color{black} \cdot [D(s')-D(s)]}-1] \cdot  [X(s')]^{   \frac{b_-^{-2}\cdot \ep^{-2} }{2} \cdot(1+\color{black}\ep^2 \cdot \log(\ep^{-1}) \color{black} \cdot D(s'))}ds'}{  [X(s)]^{   \frac{b_-^{-2}\cdot \ep^{-2} }{2} \cdot(1+\color{black}\ep^2 \cdot \log(\ep^{-1}) \color{black} \cdot D(s))}}.
							\end{align*}
							Then	\eqref{E1bound} follows immediately from \eqref{F0bound} \color{black} and \eqref{Omega.formula}\color{black}. \eqref{M.def} is obtained by the substitution $Z(s'):=\frac{X(s')}{X(s)}$ for all fixed $s$.
							
							By \eqref{D0bound} we also obtain \eqref{E2bound} as such \begin{align*}
								\mathfrak{E}_2(s) \leq\ & 2 D_0 \cdot \color{black} \ep^2 \cdot \log(\epsilon^{-1})^2  \color{black} \cdot \frac{ \int_{s}^{s_{\infty}} [X(s')]^{   \frac{b_-^{-2}\cdot \ep^{-2} }{2} \cdot(1+\color{black}\ep^2 \cdot \log(\ep^{-1}) \color{black} \cdot D(s'))}ds'}{  [X(s)]^{   \frac{b_-^{-2}\cdot \ep^{-2} }{2} \cdot(1+\ep \cdot D(s))}}  \\ \leq\ &  4 D_0 \cdot  \color{black} \ep^2 \cdot \log(\epsilon^{-1})^2 \color{black} \cdot \frac{ \int_{s}^{s_{\infty}} e^{ \color{black} \ep^2 \cdot \log(\epsilon^{-1})^2   \color{black}\cdot [D(s')-D(s)]}[X(s')]^{   \frac{b_-^{-2}\cdot \ep^{-2} }{2} \cdot(1+\color{black}\ep^2 \cdot \log(\ep^{-1}) \color{black} \cdot D(s'))}ds'}{  [X(s)]^{   \frac{b_-^{-2}\cdot \ep^{-2} }{2} \cdot(1+\color{black}\ep^2 \cdot \log(\ep^{-1}) \color{black} \cdot D(s))}}.
							\end{align*} 
						\end{proof}
						
						\color{black} Now we turn to estimating the main term $\mathcal{M}(s)$ in the earlier proposition. To quantify the realization of \eqref{X.rough.estimate}, we will obtain as a result control over a quantity $\log(\Delta_{\Omega^2})$ which we show is small of order $O(\epsilon^2 \cdot \log(\ep^{-1}) )$. \color{black}
						\begin{prop} The following uniform estimate holds true:
							\begin{equation}\label{tildeM.main.bound}\sup_{s \in \mathcal{C}}\bigl| \frac{\mathcal{M}(s)}{X^2(s)}- 2 b_-^2 \bigr| \ls \color{black} \ep^2 \cdot \log(\ep^{-1}) \color{black} .\end{equation} 
							Defining $ \Delta_{\Omega}:=\frac{|K_-| \cdot \int_s^{s_\infty}\Omega(s')ds'}{\Omega(s) \cdot X^2(s)} $, the above bound implies that
							\begin{align} \label{DeltaOmegabound}
								&  \sup_{s\in\mathcal{C}}\bigl| \Delta_{\Omega}(s)-1\bigr| ,\ \sup_{s\in\mathcal{C}}\bigl| \log(\Delta_{\Omega}(s))\bigr| 
								\ls \color{black} \ep^2 \cdot  \log^2(\ep^{-1}).
							\end{align}
							
						\end{prop}
						
						\begin{proof}
							
							We have, defining: \begin{align}
								&\tilde{M}(s):= \frac{ \mathcal{M}(s)}{X^2(s)}= \int_0^{1} 
								\Gamma(Z) Z^{1+\frac{b_-^{-2} \ep^{-2}}{2} }dZ\\ &\  \Gamma(Z)=  [X(s)]^{\frac{b_-^{-2} \log(\ep^{-1})}{2}\cdot [D (Z X(s))- D ( X(s))]}Z^{ \frac{b_-^{-2}}{2}  \cdot \log(\ep^{-1}) \cdot D(X(s) Z)}.
							\end{align}

							First, note by \eqref{D0'bound} that \begin{align*}
								|D (Z X)- D(X)| \leq \int_{Z\cdot X}^{X} |\frac{dD(Y)}{dY}| dY \leq D'_0\int_{Z\cdot X}^{X} \frac{dY}{Y \ln(Y^{-1})} = D'_0 \cdot  \ln(1+ \frac{\ln(Z)}{\ln(X)})= D'_0\cdot  \ln(1+ \frac{\ln(Z)}{\ln(X)}).
							\end{align*} Therefore, using the inequality $\ln(X^{-1}) \ln(1+\frac{a}{\ln(X^{-1})})\leq a$ for all $a \geq 0$ and $X \in (0,1)$ we have \begin{align*}
								|\ln(X) \cdot [D (Z X)- D(X)]| \leq D'_0\cdot  \ln(X^{-1}) \cdot \ln(1+ \frac{\ln(Z)}{\ln(X)}) \leq D'_0\cdot    \ln(\frac{1}{Z}).
							\end{align*}
							
							Therefore \color{black} $|\frac{\ln(X) [D(ZX)-D(X) ]}{\ln(Z)}|\leq D'_0$ is bounded, and \color{black}  also using \eqref{D0bound},   we see that \color{black} for all $0<Z<1$: \begin{align*}
								\Gamma(Z) &\ = e^{\ln(X(s))\cdot [\frac{b_-^{-2}\cdot \color{black}\log(\ep^{-1}) \color{black} }{2}  \cdot[ D(X(s) Z )- D(X(s))]]+  \ln(Z) \cdot  \frac{b_-^{-2}\cdot \color{black}  \log(\ep^{-1}) \color{black} }{2}  \cdot D(X(s)Z))} \\  &\ \color{black}= e^{ \ln(Z) \cdot  \frac{b_-^{-2}\cdot \color{black} \log(\ep^{-1}) \color{black} }{2} \cdot \left[ D(X(s)Z)+ \frac{\ln(X(s))  \cdot[ D(X(s) Z )- D(X(s))]}{\ln(Z)} \right] } = e^{\ln(Z) \cdot O(\log(\ep^{-1}))},
							\end{align*} from which we deduce $$   \tilde{M}(s) = \frac{1}{2+ \frac{b_-^{-2} \ep^{-2}}{2} + O(\log(\ep^{-1}))} = 2b_-^2 +O (\ep^2 \log(\ep^{-1})).$$

						\end{proof}\color{black}
						\subsection{The $L^{\infty}$ estimates}

						In this section, we show that the metric is uniformly close ($L^{\infty}$) to a Kasner metric. We start with two immediate identities which will be of convenience in the proof.
						\color{black}
						
						\begin{lem} \label{hep.lemma1}
							We have the following identities
							
							\begin{equation} \label{X.formula}
								(1+  \color{black} \ep^2 \cdot \log(\ep^{-1}) \color{black}\cdot D(\tau)+ 4 b_-^2 \ep^2) \cdot \log(X)= 2 b_-^2 \cdot \ep^2 \log(\tau)-  2 b_-^2 \cdot \ep^2 \log(\frac{\tau_0}{|K_-|})-2 b_-^2\cdot \ep^2 \log(\Delta_{\Omega}(\tau))- \log(\Delta_X(\tau))\color{black}-b_-^2 \ep^2 \log(C_{eff})\color{black},
							\end{equation}
							\begin{equation} \label{Omega.new.formula}
								\log( \frac{\tau}{\Omega^{1+\frac{4b_-^2 \ep^2}{1+\color{black} \ep^2 \cdot \log(\ep^{-1}) \color{black}\cdot  D(\tau)}}})=-2\ \color{black} \frac{ \log(\Delta_X(\tau))}{1+\color{black} \ep^2 \cdot \log(\ep^{-1}) \color{black} \cdot D(\tau)}+ \log(\Delta_{\Omega}(\tau))+\log(\frac{\tau_0}{|K_-|})\color{black} - \frac{2b_-^2 \ep^2 \cdot \log(C_{eff})}{1+\ep^2 \cdot \log(\ep^{-1}) \cdot D(\tau)}\color{black}.
							\end{equation}
						\end{lem}
						\begin{proof}
							\eqref{X.formula} and \eqref{Omega.new.formula} are immediate re-writings of  the identities given \color{black} in Section~\ref{prel}\color{black}.
						\end{proof}
						
						\color{black} Next, we give the $L^{\infty}$ estimates for $X$ (they immediately follow from the estimates from the earlier sections on $\Delta_X$ and $\Delta_{\Omega}$). \color{black}
						
						\begin{prop} \label{Prop.X}   Let $P(\tau):= \frac{1}{\color{black} \ep^2 \cdot \log(\ep^{-1}) \color{black}} \cdot \left[(1+ \color{black} \ep^2 \cdot \log(\ep^{-1}) \color{black} \cdot D(X(\tau))+ 4 b_-^2 \ep^2)^{-1}-1\right]$. \\There exists $P_0(M,e,\Lambda,m^2)>0$ such that  \begin{equation} \label{P0.bound}
								\sup_{ \tau \in \mathcal{C}}|P(\tau)| \leq P_0,	\end{equation}	\begin{equation} \label{XKasner1}
								\sup_{ \tau \in \mathcal{C}}\	\bigl| \ \log(\frac{X(\tau)^{1+\color{black} \ep^2 \cdot \log(\ep^{-1}) \color{black} \cdot D(\tau)+4b_-^2 \ep^2}}{\color{black}(\frac{\tau}{\sqrt{C_{eff}}})\color{black}^{ 2b_-^2\ep^2 }})\  \bigr|   \ls \color{black}\ep^{4}\log(\epsilon^{-1})^2\color{black},
							\end{equation}	\begin{equation} \label{XKasner1.5}
								\sup_{ \tau \in \mathcal{C}}\ \bigl|\ \log(\frac{X^2(\tau) }{\color{black}(\frac{\tau}{\sqrt{C_{eff}}})\color{black}^{4 b_-^2\ep^2 \cdot (1+ \color{black} \ep^2 \cdot \log(\ep^{-1}) \color{black}\cdot P(\tau))}})\bigr|=\sup_{ \tau \in \mathcal{C}}\ \bigl|\ \log(\frac{X(\tau) }{\color{black}(\frac{\tau}{\sqrt{C_{eff}}})\color{black}^{2b_-^2\ep^2 \cdot (1+ \color{black} \ep^2 \cdot \log(\ep^{-1}) \color{black} P(\tau))}})\bigr| \ls\color{black}\ep^{4}\log(\epsilon^{-1})^2\color{black},
							\end{equation}
							\begin{equation} \label{XKasner2} \begin{split}
									\sup_{ \tau \in \mathcal{C}}	\bigl| X^2(\tau)-\color{black}(\frac{\tau}{\sqrt{C_{eff}}})\color{black}^{4b_-^2\ep^2 \cdot (1+ \color{black} \ep^2 \cdot \log(\ep^{-1}) \color{black} P(\tau))}
									\bigr| 		\ls \color{black}\ep^{4}\log(\epsilon^{-1})^2\color{black} \cdot \tau^{4b_-^2\ep^2 \cdot (1+ \color{black} \ep^2 \cdot \log(\ep^{-1}) \color{black} P(\tau))} \ls   \color{black}\ep^{4}\log(\epsilon^{-1})^2\color{black}  ,\end{split} 	\end{equation}
						\end{prop}

						\begin{proof}
							\eqref{P0.bound} follows immediately from \eqref{D0bound} and \eqref{DeltaX.bound}\color{black}. \eqref{XKasner1} is immediate combining  \eqref{DeltaOmegabound}, \eqref{X.formula} and \eqref{tau0.est}. \eqref{XKasner1.5}, \eqref{XKasner2} are just re-writings of \eqref{XKasner1}. 
							
						\end{proof}
						\color{black} Next, we turn to the uniform $\Omega$ estimates. \color{black}
						
						\begin{prop} \label{Prop.Omega}   Note the identity $(1+\frac{4b_-^2 \ep^2}{1+\color{black} \ep^2 \cdot \log(\ep^{-1}) \color{black} D(\tau)})^{-1}=1-4b_-^2 \ep^2 \cdot (1+\color{black} \ep^2 \cdot \log(\ep^{-1}) \color{black} \cdot P(\tau))$, where $P$ is defined in Proposition \ref{Prop.X}. Then, for all such that for all $\tau \in \mathcal{C}$:
							
							\begin{equation} \label{OmegaKasner1.5}
								\sup_{ \tau \in \mathcal{C}}\ \bigl|\ \log(\frac{\Omega^2(\tau) }{\tau^{2[1-4b_-^2 \ep^2 \cdot (1+\color{black} \ep^2 \cdot \log(\ep^{-1}) \color{black} \cdot P(\tau))]}})\ \bigr| =2	\sup_{ \tau \in \mathcal{C}}\ \bigl|\ \log(\frac{\Omega(\tau) }{\tau^{1-4b_-^2 \ep^2 \cdot (1+\color{black} \ep^2 \cdot \log(\ep^{-1}) \color{black} \cdot P(\tau))}} )\ \bigr| \ls  \ep^2  \log^2(\ep^{-1})\color{black},
							\end{equation}
							\begin{equation} \label{OmegaKasner2} \begin{split}
									\sup_{ \tau \in \mathcal{C}}\	\bigl|\ \Omega^2(\tau)-\tau^{2[1-4b_-^2 \ep^2 \cdot (1+\color{black} \ep^2 \cdot \log(\ep^{-1})  \cdot P(\tau))]}\
									\bigr| 	 \ls  \ep^2 \color{black}\cdot \tau^{2[1-4b_-^2 \ep^2 \cdot (1+ \ep^2 \cdot \log(\ep^{-1})  \cdot P(\tau))]} 	\ls  \epsilon^{10}  \log^2(\ep^{-1}).\end{split} 	\end{equation}
							
						\end{prop}
						
						\begin{proof}

							\eqref{OmegaKasner1.5} and \eqref{OmegaKasner2} follow directly from \eqref{Omega.new.formula}, \eqref{DeltaOmegabound}, \eqref{DeltaX.bound}, \eqref{D0bound} and \eqref{tau0.est}.

						\end{proof}
						
						\color{black} Finally, for completeness we also give the $\phi$ estimates from Section~\ref{crushing.section} in terms of the new $\tau$ coordinate.
						\color{black}
						\begin{prop}
							For all $\tau\in \mathcal{C}$, we have \begin{equation} \label{phi.Kasner}
								\bigl| \phi(\tau) - 2 b_- \cdot \ep \cdot \log(\tau^{-1}) \bigr| \ls \color{black} \ep \cdot \log(\ep^{-1}) \color{black}+\color{black} \ep^3 \cdot \log(\ep^{-1}) \cdot \log(\tau^{-1}).\color{black}
						\end{equation}   \end{prop}
						\begin{proof}
							This follows immediately from \eqref{XKasner2} and \eqref{C1}.
						\end{proof}
						
						\color{black}

							\color{black}
							\subsection{Concluding the proof of Theorem~\ref{Kasner.thm}}
							Finally we gather all the estimates we proved to finish the proof of Theorem~\ref{Kasner.thm}.
							\color{black}
							
							\begin{cor}
								\color{black}	Define the variables $\rho$ and $\tilde{\tau}$ by \eqref{scale1}, \eqref{scale2}\color{black}. Then the metric satisfies \eqref{metric.Kasner.unif} and \eqref{metric.Kasner.BMO} with $P$ satisfying \eqref{main.uniform.0} and  $\mathfrak{E}^{ang}_{\ep}(\tau)$, $\mathfrak{E}^{rad}_{\ep}(\tau)$, $E^{ang}_{\ep}(\tau)$ and $E^{rad}_{\ep}(\tau)$ satisfying \eqref{main.uniform.1}, \eqref{main.uniform.2}, \eqref{main.BMO1}, \eqref{main.BMO2} and \color{black}\eqref{SF.thm}\color{black}.
							\end{cor}\begin{proof}
								This follows from Proposition \eqref{Prop.X}, Proposition \eqref{Prop.Omega}, and \color{black}\eqref{tau0.est}\color{black}, \eqref{LB1}, \eqref{LB4}, \color{black}\eqref{phi.Kasner}. For \eqref{SF.thm}, note that, recalling $\tilde{\tau} = \frac{\tau}{\tau_0}$, we have $$ \ep  \log(\tilde{\tau})= \ep [\log(\tau)- \log(\tau_0)]= \ep \log(\tau) + O(\ep), $$ which is why \eqref{phi.Kasner} is still true with $\tilde{\tau}$ replacing $\tau$.
								
								\color{black}

							\end{proof}

						\section*{Statements and declarations} \textbf{Competing Interests}: No competing interests.\\ \textbf{Data availability}: Data sharing not applicable to this article as no datasets were generated or analysed during the current study.

					\end{document}